\documentclass[authoryear,12pt]{article}

\usepackage{natbib}
\setcitestyle{square,numbers}
\usepackage[title]{appendix}
\usepackage[english]{babel}
\usepackage{amsmath}
\usepackage{latexsym}
\usepackage{amssymb}
\usepackage{amsmath}
\usepackage{amsthm}
\usepackage{mathtools}
\usepackage{multirow}
\usepackage{multicol}
\usepackage{float}
\usepackage{enumerate}
\usepackage{graphicx}
\usepackage{booktabs}
\usepackage{hyperref}
\usepackage{enumitem}
\usepackage{url}
\usepackage{hyperref}
\usepackage{caption}
\usepackage{epsfig}
\usepackage{lscape}
\usepackage[normalem]{ulem}
\usepackage[usenames,dvipsnames]{xcolor}
\usepackage{tikz}
\usetikzlibrary{bayesnet}
\usepackage{bbm}
\usepackage{tocloft}
\usepackage[blocks]{authblk} 

\usepackage{xr}
\RequirePackage[figuresright]{rotating}

\makeatletter
\newcommand*{\centernot}{%
	\mathpalette\@centernot
}
\def\@centernot#1#2{%
	\mathrel{%
		\rlap{%
			\settowidth\dimen@{$\m@th#1{#2}$}%
			\kern.5\dimen@
			\settowidth\dimen@{$\m@th#1=$}%
			\kern-.5\dimen@
			$\m@th#1\not$%
		}%
		{#2}%
	}%
}
\makeatother

\newcommand{\vertiii}[1]{{\left\vert\kern-0.20ex\left\vert\kern-0.20ex\left\vert #1 
		\right\vert\kern-0.20ex\right\vert\kern-0.20ex\right\vert}}

\newcommand{\norm}[1]{\left\lvert#1\right\rvert}
\newcommand{\lnorm}[1]{\left\lVert#1\right\rVert}

	\DeclareMathOperator*{\argmax}{arg\,max}
	
	\newtheorem{theorem}{Theorem}
	\numberwithin{theorem}{section}
	\newtheorem{lemma}[theorem]{Lemma}
	\newtheorem{proposition}[theorem]{Proposition}
	\newtheorem{corollary}[theorem]{Corollary}
	\newtheorem{defin}{Definition}

	\newcounter{rmk}
	\newtheorem{rem}[rmk]{Remark}
	\theoremstyle{remark}
	\newtheorem{example}{Example}
	\newtheorem{excont}{Example}

	\newtheorem{exnet}{Example}
	
	\newcommand{\N}{\mathbb{N}}
	\newcommand{\Z}{\mathbb{Z}}
	\newcommand{\R}{\mathbb{R}}
	\newcommand{\E}{\mathrm{E}}

	\newcommand{\Fb}{\mathcal{F}}

	\author[1,2,*]{Mirko Armillotta}
	\author[3,**]{Konstantinos Fokianos}
	
	\affil[1]{Department of Econometrics and Data Science, Vrije Universiteit Amsterdam}
	\affil[2]{Tinbergen Institute}
	\affil[3]{Department of Mathematics and Statistics, University of Cyprus} %
	
	{
		\makeatletter
		\renewcommand\AB@affilsepx{: \protect\Affilfont}
		\makeatother
		
		\affil[ ]{Email}
		
		\makeatletter
		\renewcommand\AB@affilsepx{, \protect\Affilfont}
		\makeatother
		
		\affil[*]{m.armillotta@vu.nl}
		\affil[**]{fokianos.konstantinos@ucy.ac.cy}
	}

	\title{Nonlinear Network Autoregression}
	\date{30th November 2023}
	
\begin{document}
		\maketitle
		\begin{abstract}
			\noindent 
				We study general nonlinear models for  time series networks  of integer and  continuous valued data.  The vector of high dimensional responses, measured on the nodes of a known network, is regressed non-linearly  on its lagged value and on lagged values of the neighboring  nodes by  employing a  smooth link function.  We study stability conditions for such multivariate process and develop quasi maximum likelihood inference when the network dimension is increasing. In addition,  we study linearity score tests by treating separately the cases of identifiable and non-identifiable parameters. In the case of identifiability, the test statistic converges to  a  chi-square distribution. When the parameters are  not-identifiable, we develop a supremum-type test whose $p$-values are approximated adequately by employing a feasible bound and bootstrap methodology. Simulations and data examples support further our findings.
		\end{abstract}
		
		\advance\cftsecnumwidth 0.5em\relax
		\advance\cftsubsecindent 0.5em\relax
		\advance\cftsubsecnumwidth 0.5em\relax
		
		\tableofcontents
		
		\section{Introduction} \label{intro}

		The availability of network data recorded over a timespan in several applications (social networks, GPS data, epidemics, etc.) requires assessing the effect of a network structure to a multivariate time series process. This problem has attracted considerable attention. Recently 
		\cite{zhu2017} proposed a Network Autoregressive model (NAR), under Independent Identically  Distributed ($IID$) innovation sequence, where a continuous response variable is observed for each node of a network. The high dimensional vector of such responses is modelled   linearly on the past values of the response, measured on the node itself and the average lagged response of the neighbours connected to the node. Motivated by the fact that real data networks are usually of large dimension, the authors study least squares  inference under two large-sample regimes (a) with increasing time sample size, i.e. $T\to\infty$, and fixed network dimension, say $N$, and (b) with $N\to \infty$ and $T_N\to\infty$, where the temporal size depends on $N$. In this contribution, we extend this work in various directions by proposing appropriate inference and testing methodology which is applicable to continuous and integer valued data. 
		\subsection{Related works}
		The NAR model has been the focus of recent research e.g. logistic network models \cite{zhou2017dynamic}, network quantile model \cite{zhu2019}, grouped least squares estimation \cite{zhu2020}, network GARCH models \cite{zhou2020}, and applications \cite{armillotta2023unveiling}.
		In addition, \cite{Knightetal(2020)} consider 
		the Generalized NAR model, for the continuous case, which incorporates the effect of several layers of connections between  nodes of the network. This study is accompanied by appropriate   \texttt{R} software. 
		As pointed out by  \cite{zhu2017}, discrete response variables are frequently  encountered in applications and are strongly related to network data. For example, in the social network analysis data are counts, e.g. number of characters contained in posts of single users, number of posts shared, etc. Models for count time series have been studied by  \cite{armillotta_fokianos_2021} who introduced  linear and log-linear Poisson network autoregression models (PNAR). 
		Such extensions show that  the NAR model is a member of the broad class of     Generalized Linear Models (GLM) \cite{McCullaghandNelder(1989)};
		for the count case,  the observations are in fact marginally Poisson distributed, conditionally to their past history. The joint dependence  is imposed by employing a copula construction, as introduced by \cite{fok2020} and is outlined in the Supplement \ref{Supp:copula-Poisson}. In addition, \cite{armillotta_fokianos_2021} have studied thoroughly the two related types of asymptotic inference (a)-(b) discussed above, in the context of quasi maximum likelihood inference (QMLE) \citep{hey1997}.
		
		\subsection{Nonlinear models \& testing linearity}
		Theory related to NAR model relies on the assumption of linearity. However, there are many real world examples where a non-linear model might be more appropriate. For instance, in economics,  the theory  supports occasionally   nonlinear behaviour, see \cite[Ch.~2]{Terasvirtaetal(2010)} for several examples. In modelling economic/financial time series, it seems natural to allow for the existence of different states, or regimes (e.g. expansion/crisis), such that  dynamics depend upon the specific regime; \cite[Ch.~18]{zivot_2006}. Government agencies, research institutes and central banks may typically employ nonlinear models \cite[p.~16]{Terasvirtaetal(2010)}. 
		As far as  social network analysis concerned,  nonlinear phenomena are frequently observed e.g. "superstars" with huge number of followers having an exponentially higher impact on other users' behaviour when compared to the "standard" user \cite{zhu2017}. So from both theoretical and applied point of view, there exists a necessity to develop non-linear autoregression model theory for the case of network time series.
		Literature in univariate non-linear time series  is vast and well developed, in particular for continuous random variables. The interested reader is referred to \cite{tong_1990}, \cite{chan_tong_2001chaos} and \cite{fan_yao_2003}, among many others. For integer-valued data the literature is more recent, although still flourishing;  see \cite{davis2003}, \cite{fok2009}, \cite{fokianos2012nonlinear}, \cite{davis2009negative}, \cite{chri2014}, \cite{Wangetal(2014)} and \cite{gorgi2020beta}. 
		General results are given by  \cite{Woodardetall(2010)}, \cite{douc2013}, \cite{fra2016}, \cite{davis2016}, \cite{fok2017} and \cite{armillotta_2022_EJS}.

		Estimation for non-linear models has been traditionally accompanied by tests who examine the assumption of linearity.
		Such tests are routinely used  because they provide a sound framework where evidence of linearity can be examined thoroughly. In addition, they  offer guidance about the specific non-linear model to be  fitted; see \cite[Sec.~3]{terasvirta_1994aspects} who suggests that  	"The first step of a specification strategy for any type of nonlinear model should therefore consist of testing linearity".  Furthermore, proper inference is developed especially when  the linear model is nested within a non-linear model and as such the resulting estimators obtained after fitting a non-linear model may be inconsistent; see \cite[Sec.~5.1,5.5]{Terasvirtaetal(2010)}. Finally, it is always important to have additional tools for both practical usefulness and for explanatory data analysis;   detecting latent variables, change point testing etc.
		These points motivated the growth of a large literature on linearity tests, especially for continuous-valued random variables.  
		A survey of general type results for test statistics whose application depend on identifiable/non-identifiable  parameters are given by \cite{francq_etal_2010} and  \cite{hansen_1999testing}.   Finally, \cite{andrews_1994} and \cite{hansen_1996} established a general framework for testing linearity when some parameters are non-identifiable under the hypothesis of linearity. Non-linear models for count time series and the associated testing linearity problem has been studied by   \cite{christou_fokianos_2015} who  suggest a quasi-score test for (mixed) Poisson random variables.  All the above works are concerned with univariate time series.  Related literature on multivariate observation-driven models for discrete-valued data considers only linear cases, see \cite{pedeli2011,pedeli2013,pedeli22013} and \cite{fok2020}, among others.

		\subsection{Main challenges}
		Existing theory does not cover the case of NAR models which are  multivariate
		and their properties depend  on both $N$ (size) and $T$ (time). 
		Therefore, asymptotically, both indices, $N$ and $T$, tend to infinity and it is a great challenge to address
		the properties of such multivariate processes.  Moreover, QMLE inference  breaks down  when estimating network models because $N$ is large.  Consequently, non standard proofs  are required for establishing  stationarity of infinite-dimensional processes and to obtain sound inference  within the  double asymptotic regime (b). Note that even a simple weak law does not hold under regime (b). In particular,  quantities related to the inference are of the order $\mathcal{O}(N)$. Consider, for example,  
		the sample information matrix  which depends on the network structure and diverges with $N\to\infty$. Then the covariance of estimators explodes and 
		proving  existence of limiting  Hessian and Fisher matrices is a challenging problem.  These issues  become more persistent when  testing   linearity, especially  in the case of non-identifiable parameters. Then, a double indexed  asymptotic theory with infinite dimensional vectors  over a uniform metric space for the score and related matrices is relevant and asks for development.  
		
		\subsection{Our contribution}
		The main results of our contribution are the following: i) Specification of a novel general nonlinear network autoregressive model for both continuous and discrete valued multivariate network observations (Section~\ref{SEC: model}).
		ii) Under mild  conditions, stationarity results (Section~\ref{stability conditions}-\ref{stability conditions div N}) and asymptotic theory of QMLE are established, when both time and network dimensions increases (Section~\ref{SEC: inference}). These are new results because  nonlinear NAR models have not been treated in the literature. iii)  Development of testing procedures for examining linearity  of the model by employing 
		a quasi-score based test under both asymptotic regimes (a)-(b); see  Sec. ~\ref{SEC: linearity test}. 
		We focus on  score  tests, as they require fitting  NAR models under the null hypothesis. This  is 
		computationally simpler task.  Their asymptotic distribution is  (non-central) chi-square   even when the parameters to be tested  lie on the boundary of the parameter  space.
		iv) Finally, we consider the situation where non identifiable parameters, say  $\gamma$, are present under the null. In such case the results of Section~\ref{SEC: model}-\ref{SEC: linearity test} are  extended. This is done by  proving stochastic equicontinuity of the score with diverging number of nodes and double-indexed convergence of Hessian/information matrices uniformly over $\gamma$. Then, as $N\to\infty$ and $T_N\to\infty$, we show that the quasi-score linearity test asymptotically approximates a (non-central) chi-square process (Section~\ref{SEC: identifiability}).  We  discuss two ways to approximate the $p$-values of the tests:  
		by the upper bound developed in \cite{davies_1987}, and by bootstrap approximation relying on stochastic permutations \cite{hansen_1996}. The double asymptotic convergence of bootstrap $p$-values to their theoretical counterpart is also established. We are not aware of other contributions, to the best of our knowledge,  attacking the problem of general asymptotic inference with  increasing dimension 
		network time series models for both count and continuous data.
		The last sections  discuss results of a simulation study (Section~\ref{SEC: simulations}) and real data examples (Section~\ref{SEC: applications}). All the methodology is implemented in the new released \texttt{R} package \texttt{PNAR} \cite{armillotta_et_al_rpackage_2022,armillotta_2022r}. 
		The paper concludes with Appendix \ref{SEC appendix}, containing the proofs for Sections \ref{SEC: model} and \ref{SEC: identifiability}. A Supplementary Material \ref{SEC supplementary material} (henceforth Supp. Mat.) contains the proofs for Sections \ref{SEC: inference} and  \ref{SEC: linearity test} and additional  results for Threshold Network Autoregressive model, under asymptotic regime (a).

		\subsection{Notation}  We denote $|x|_r=(\sum_{j=1}^{p}\norm{x_j}^r)^{1/r}$ the $l^r$-norm of a $p \times 1$ vector $x$. If $r=\infty$, $|x|_\infty=\max_{1\leq j\leq p}|x_{j}|$. Let $\lVert X\rVert_r=(\sum_{j=1}^{p}\E(|X_j|^r))^{1/r}$ the $L^r$-norm for a random vector $X$.  For a $q \times p$ matrix $M=(m_{ij})$, ${i=1,\ldots,q, j=1,\ldots,p}$, denote the generalized matrix norm $\vertiii{M}_{r}= \max_{\norm{x}_{r}=1} \norm{Mx}_{r}$. If $r=1$, then $\vertiii{M}_1=\max_{1\leq j\leq p}\sum_{i=1}^{q}|m_{ij}|$. If $r=2$, $\vertiii{M}_2=\rho^{1/2}(M^\prime M)$, where $\rho(\cdot)$ is the spectral radius. If $r=\infty$, $\vertiii{M}_\infty=\max_{1\leq i\leq q}\sum_{j=1}^{p}|m_{ij}|$. If $q=p$, these norms are matrix norms. Define the entry-wise norms $\norm{M}_{r}=(\sum_{i=1}^{q}\sum_{j=1}^{p}\norm{m_{ij}}^r)^{1/r}$. If $q=p$ and $1\leq r \leq 2$, these are matrix norms. Define $\norm{x}^r_{vec}=(\norm{x_1}^r,\dots,\norm{x_p}^r)^\prime$, $\lnorm{X}_{r,vec}=(\E^{1/r}\norm{X_1}^r,\dots,\E^{1/r}\norm{X_p}^r)^\prime$, $\norm{M}_{vec}=(\norm{m_{ij}})_{(i,j)}$ and $\preceq$ a partial order relation on $x,y\in\R^p$ such that $x\preceq y$ means $x_i\leq y_i$ for $i=1,\dots,p$. The same notation holds for random vectors $X,Y$ such that $X \preceq Y$ means $X_i\leq Y_i$ almost surely (a.s.) for $i=1,\dots,p$. 
		Set the compact notation $\max_{1\leq i< \infty}x_i=\max_{i\geq 1}x_i$.
		The notations $C_r$ denote a constant which depend on $r$, where $r\in\N$, and $C$ is a generic constant. The symbol $I$ denotes an identity matrix, $1$ ($0$) a vector of ones (zeros), whose dimension  depends on  context.
		Let $\Rightarrow$ denote weak convergence with respect to the uniform metric. 
		Finally, the notation $ \left\lbrace N, T_N \right\rbrace \to\infty$ will be used as a shorthand for $N\to\infty$ and $T_N\to\infty$. 

		\section{Nonlinear NAR model specification} \label{SEC: model}
		
		Consider a network with $N$ nodes (network size) indexed by $i=1,\dots N$. The neighborhood structure of the network is explicitly described by its adjacency matrix $A=(a_{ij})\in\R^{N\times N}$ where $a_{ij}=1$, if there is a directed edge from $i$ to $j$ (e.g. user $i$ follows $j$ on Twitter, a flight take off from airport $i$ landing to airport $j$), and $a_{ij}=0$ otherwise. Undirected graphs are allowed ($A=A^\prime$) but self-relationships are excluded i.e. $a_{ii}=0$ for any $i=1,\dots,N$. This is a typical and realistic assumption, e.g. social networks, see \cite{wass1994} and \cite{kola2014}, among others. The network structure, equivalently  the matrix $A$, is assumed to be non-random. A row-normalized adjacency matrix is defined by $W=\mathrm{diag}\left\lbrace n_1,\dots,n_N\right\rbrace^{-1}A $ where $n_i=\sum_{j=1}^{N}a_{ij}$ is the so called out-degree, the total number of edges starting from the node $i$. Then, $W$ satisfies $\vertiii{W}_\infty=1$ and $W 1 = 1$. Moreover, define  $e_i$ the $N$-dimensional unit vector with 1 in the $i^{\text{th}}$ position and 0 everywhere else, such that $w_i^\prime=e_i^\prime W=(w_{i1}\dots,w_{iN})$ the  $i^{\text{th}}$ row of the matrix $W$, with $w_{ij}=a_{ij}/n_i$.
		
		Define a $N$-dimensional vector of time series $\left\lbrace Y_t=(Y_{1,t} \dots Y_{i,t} \dots Y_{N,t})^\prime,\,t=1,2\dots,T\right\rbrace $ which is observed on a given network; i.e. a univariate time series is measured for each node, with rate $\lambda_{i,t}$. Denote by
		$\left\lbrace \lambda_t\equiv\E(Y_t|\Fb_{t-1})=(\lambda_{1,t} \dots \lambda_{i,t} \dots \lambda_{N,t})^\prime,\,t=1,2\dots,T\right\rbrace $, the corresponding conditional expectation vector, and denote the history of the process by $\Fb_t=\sigma(Y_s: s\leq t)$ . Assume that  $\left\lbrace Y_t: t\in\Z \right\rbrace $ is integer-valued and consider the following first order nonlinear Poisson Network Autoregression (PNAR)
		\begin{equation}
			Y_t=N_t(\lambda_t),\quad\quad \lambda_t=f(Y_{t-1}, W, \theta),
			\label{nonlinear_pnar}
		\end{equation}
		where $\left\lbrace N_t \right\rbrace $ is a sequence of $IID$ $N$-variate copula-Poisson process with intensity 1, counting the number of events in $[0,\lambda_{1,t}]\times\dots\times[0,\lambda_{N,t}]$ and $f(\cdot)$ is a deterministic function depending on the past lag values of the count process, the known network structure $W$ and an $m$-dimensional parameter vector $\theta$. Examples will be given below. More precisely, the conditional marginal probability distribution of the count variables is $Y_{i,t}|\mathcal{F}_{t-1}\sim Poisson(\lambda_{i,t})$, for $i=1,\dots,N$, 
		and the joint distribution is generated by a copula, which depends on a parameter $\rho$, say $C(\cdot, \rho)$ and it is imposed on waiting times of a Poisson process specified as in 
		\cite[Sec.~2.1]{armillotta_fokianos_2021}; See Supp. Mat.~\ref{Supp:copula-Poisson}. Several alternative models resembling multivariate Poisson distributions have been proposed in the literature; see \cite[Sec.~2]{fokianos_2021} for a discussion about the issues of available  multivariate count distributions. 
		A copula-based approach for the data generating process (henceforth DGP) is used throughout this paper. Results for higher order models are derived straightforwardly; see Remark~\ref{Rem: p lag}.
		
		Similar to the case of integer-valued time series, we define the nonlinear Network Autoregression (NAR) for continuous-valued time series by
		\begin{equation}
			Y_t=\lambda_t+\xi_t,\quad\quad \lambda_t=f(Y_{t-1}, W, \theta)
			\label{nonlinear_nar}
		\end{equation} 
		where $\xi_{i,t}\sim IID(0,\sigma^2)$, for $1\leq i \leq N$ and $1\leq t\leq T$ and $\lambda_{t}=\E(Y_t|\mathcal{F}_{t-1})$.
		
		Denote by $X_{i,t}=n_i^{-1}\sum_{j=1}^{N}a_{ij}Y_{j,t}$ the network effect, i.e. the average impact of node $i$'s connections. Consider the partition of the parameter vector $\theta=( \theta^{(1)\prime}, \theta^{(2)\prime}) ^\prime$, where the vectors $\theta^{(1)}$ and $\theta^{(2)}$ are of dimension $m_1$ and $m_2$, respectively, such that $m_1+m_2=m$. 
		For $t=1\dots,T$, both \eqref{nonlinear_pnar}-\eqref{nonlinear_nar} have element-wise components
		\begin{equation}
			\lambda_{i,t}=f_i(X_{i,t-1}, Y_{i,t-1}; \theta^{(1)}, \theta^{(2)})\,,\quad i=1,\dots,N\,,
			\label{general model}
		\end{equation}
		where $f_i(\cdot)$ is the $i^{\text{th}}$ component of the function $f(\cdot)$ depending on the specific model of interest, which can contain linear and nonlinear effects. In general, $\theta^{(1)}$ will denote an $m_1\times1$ vector associated with linear model parameters, whereas $\theta^{(2)}$ will denote the $m_2\times 1$ vector of nonlinear parameters. Some examples are given below.
		
		\subsection{Examples}
		
		\begin{example}
			Consider \eqref{nonlinear_nar} and  the first order linear NAR(1),
			\begin{equation}
				\lambda_{i,t}=\beta_0+\beta_1X_{i,t-1}+\beta_2Y_{i,t-1}\,,
				\label{nar_1}
			\end{equation}
			which is a special case of \eqref{general model}, with $\theta^{(1)}=(\beta_0, \beta_1, \beta_2)^\prime$.  Model \eqref{nar_1} was originally introduced by \cite{zhu2017} for the case of continuous random variables ${Y_{t}}$, such that $Y_{i,t}=\lambda_{i,t}+\xi_{i,t}$. 
			For each single node $i$, model \eqref{nar_1} allows the conditional mean of the process to depend on the past of the variable itself, for the same node $i$,
			and the average of the other nodes $j\neq i$ by which the focal node $i$ is connected. Implicitly, only the nodes directly connected with the focal node $i$ can impact on the conditional mean process $\lambda_{i,t}$. This is reasonable assumption in many applications; for example, in the social network analysis, if the focal node $i$ does not follows a node $l$, so $a_{il}=0$, the effect of the activity related to the latter do not affect the former. 
			The parameter $\beta_1$ is called network effect, as it measures the average impact of
			the $i$'th  node connections.  The coefficient $\beta_2$ is called autoregressive effect because it determines the impact of the lagged variable $Y_{i,t-1}$. Model \eqref{nar_1} has been extended to the case of count time series by  \cite{armillotta_fokianos_2021};  it is called  the linear PNAR(1) with $	Y_{i,t}|\mathcal{F}_{t-1}\sim Poisson(\lambda_{i,t})$ for $i=1,\dots,N$ and the copula-based DGP, as described earlier.  
			
		\end{example}
		\begin{example}  \label{ex nonlinear}
			A nonlinear deviation of \eqref{nar_1}, when $Y_t$ takes integer values is given by 
			\begin{equation}
				\begin{aligned}
					\lambda_{i,t}=\frac{\beta_0}{(1+X_{i,t-1})^{\gamma}}+\beta_1X_{i,t-1}+\beta_2Y_{i,t-1}\,,\\
				\end{aligned}
				\label{nonlinear}
			\end{equation}
			where $\gamma\geq0$. Clearly, \eqref{nonlinear} approaches a linear model for small values of $\gamma$, and $\gamma=0$ reduces to the linear model \eqref{nar_1}. Instead, when $\gamma$ is larger than  zero, \eqref{nonlinear} introduces a perturbation, deviating from the linear model \eqref{nar_1}. Hence, \eqref{nonlinear} is a special case of \eqref{general model}, with $\theta^{(1)}=(\beta_0, \beta_1, \beta_2)^\prime$ and $\theta^{(2)}=\gamma$. Model \eqref{nonlinear} introduces a nonlinear drift in the intercept so that the baseline effect varies over time as a function of the network. 
			If   $Y_{i,t}$  counts activities of users in a social network (likes, reactions, etc.) and  the  community becomes more active, then  the average magnitude of $X_{i,t-1}$ grows and thus the baseline for each node $i$  varies.
			When  $Y_t\in\R^N$, the following model
			\begin{equation}
				\begin{aligned}
					\lambda_{i,t}=\frac{\beta_0}{(1+\norm{X_{i,t-1}})^{\gamma}}+\beta_1X_{i,t-1}+\beta_2Y_{i,t-1}\,,\\
				\end{aligned}
				\label{nonlinear_cont}
			\end{equation}
			is analogous to \eqref{nonlinear} but for continuous valued time series. 
			To the best of our knowledge, we are not aware of any stability or inferential  results for models \eqref{nonlinear}-\eqref{nonlinear_cont}  when $ \left\lbrace N, T_N \right\rbrace \to\infty$. When $N=1$ and  $T\to\infty$, such non-linear models have been studied by 
			by \cite{gao_etal_2009, fokianos2012nonlinear}, among others.
			
		\end{example}
		
		\begin{example}  \label{ex star}
			Another example  of  \eqref{general model} is given by the  Smooth Transition version of the NAR model, say STNAR(1),
			\begin{equation}
				\begin{aligned}
					\lambda_{i,t}=\beta_0+(\beta_1+\alpha\exp(-\gamma X_{i,t-1}^2))X_{i,t-1}+\beta_2Y_{i,t-1}\,,\\
				\end{aligned}
				\label{stnar}
			\end{equation}
			where $\gamma\geq0$; see  \cite{terasvirta1994} for an introduction to STAR models.  This models introduces a smooth regime switching behavior of  the network effect making it possible to vary smoothly  from $\beta_1$ to $\beta_1 + \alpha$, as $\gamma$ varies from large to small values.	When $\alpha=0$ in \eqref{stnar},  the linear NAR model \eqref{nar_1} is obtained. Moreover, \eqref{stnar} is a special case of \eqref{general model}, with $\theta^{(1)}=(\beta_0, \beta_1, \beta_2)^\prime$ and $\theta^{(2)}=(\alpha,\gamma)^\prime$.
			In the case of univariate count time series see \cite{fok2009, fokianos2012nonlinear}, for more. 
			
		\end{example}
		
		\begin{example}  \label{ex tar}
			Define the  Threshold NAR model (\cite{lim_tong_1980}), say TNAR(1),  by
			\begin{equation}
				\begin{aligned}
					\lambda_{i,t}=\beta_0+\beta_1X_{i,t-1}+\beta_2Y_{i,t-1}+(\alpha_0+\alpha_1X_{i,t-1}+\alpha_2Y_{i,t-1})I(X_{i,t-1}\leq \gamma)\,,\\
				\end{aligned}
				\label{tnar}
			\end{equation}
			where $I(\cdot)$ is the indicator function and $\gamma$ is the threshold parameter. When $\alpha_0=\alpha_1=\alpha_2=0$, model \eqref{tnar} reduces to the linear model  \eqref{nar_1}. In this case, $\theta^{(1)}=(\beta_0, \beta_1, \beta_2)^\prime$ and $\theta^{(2)}=(\alpha_0, \alpha_1, \alpha_2, \gamma)^\prime$ show that \eqref{tnar} is a special case of \eqref{general model}. In the case of univariate count time series see \cite{Woodardetall(2010), douc2013, Wangetal(2014), christou_fokianos_2015}, for more. 
			
			Nonlinear functions, such  as \eqref{stnar}-\eqref{tnar}, provide examples of  switching models accounting for    regime specific dynamics of the observed process. The switching mechanism depends on the network effect.  For 
			example,  consider  $A$ as the network matrix connecting regional districts which share at least a border.  Let $Y_{i,t}$  denote the numbers of reported cases for some disease in each of these districts. Then, for each district $i$, the historical average  of neighbours ($X_{i,t-1}$)  determines a switching effect, say from exponentially expanding pandemic to dying out pandemic (and vice versa).  Note that  for \eqref{stnar}, the  network effect is regime-dependent but this can be modified suitably as in \eqref{tnar}. In conclusion, the  dichotomy  between STNAR and TNAR  models is that the former accounts for smooth transitions  while the latter models sudden changes;  see \cite{Terasvirtaetal(2010)} for more on nonlinear modeling of time series.
		\end{example}

		\subsection{Stability conditions for fixed network size} \label{stability conditions}
		Set $f(\cdot, W, \theta)=f(\cdot)$.  
		\begin{theorem} 
			\label{Thm. Ergodicity of nonlinear model}
			Consider model \eqref{nonlinear_pnar}, with fixed  $N$. Define $G=\mu_1 W + \mu_2 I$, where $\mu_1$, $\mu_2$ are non-negative constants such that $\rho(G)<1$ and assume that  for $y,y^* \in \N^N$
			\begin{equation}
				\norm{f(y)-f(y^*)}_{vec}\preceq G\norm{y-y^*}_{vec}\,.
				\label{contraction}
			\end{equation}
			
			Then, the process $\{ Y_t,~ t \in \mathbb{Z} \}$ is stationary, ergodic and  $\mbox{E}\norm{Y_t}_a^{a}  < \infty$ for any $a\geq1$. 
		\end{theorem}
		The parallel result for continuous variables is also established.
		\begin{theorem} 
			\label{Thm. Ergodicity of nonlinear continuous model}
			Consider model \eqref{nonlinear_nar}, with fixed $N$. Define $G=\norm{\mu_1} W + \norm{\mu_2} I$, where $\mu_1$, $\mu_2$ are real constants such that $\rho(G)<1$ and the contraction condition \eqref{contraction} holds.
			Then, the process $\{ Y_t,~ t \in \mathbb{Z} \}$ is stationary ergodic with $\mbox{E}\norm{Y_t}_1  < \infty$. Moreover,  if $\E\norm{\xi_t}^a_a<\infty$ for some $a\geq 8$, then $\E\norm{Y_t}^a_a<\infty$.
		\end{theorem}
		
		The proof of Theorems~\ref{Thm. Ergodicity of nonlinear model}-\ref{Thm. Ergodicity of nonlinear continuous model} is given in  Appendix   \ref{Proof Ergodicity of nonlinear model}. 
		Theorem~\ref{Thm. Ergodicity of nonlinear model} extends  \cite[Prop.~1]{armillotta_fokianos_2021}, which was established for the linear PNAR model \eqref{nar_1}. Theorem~\ref{Thm. Ergodicity of nonlinear continuous model}, similarly,  extends \cite[Thm.~1]{zhu2017}.  In particular,  existence of some moments for  $\left\lbrace \xi_t : t\in\Z \right\rbrace $ guarantees  the conclusions of Theorem~\ref{Thm. Ergodicity of nonlinear continuous model}. Such assumption is not necessary in the linear case considered by \cite{zhu2017}--see their eq. (2.1)--because of the assumed normality.  
		
		For each $i=1\dots,N$, the contraction condition  \eqref{contraction} follows by assuming that for $x_i,x^*_i\in\R_+$ and $y_i,y_i^*\in\N$
		\begin{equation}
			\norm{f_i(x_i,y_i)-f_i(x^*_i,y^*_i)}\leq\mu_1\norm{x_i-x^*_i}+\mu_2\norm{y_i-y^*_i} \label{contraction_2}
		\end{equation}
		
		because the left hand side of \eqref{contraction_2} is bounded by 
		$\mu_1|\sum_{j=1}^{N}w_{ij}(y_{j}-y_{j}^*)|+\mu_2|y_{i}-y_{i}^*|
		\leq ( \mu_1w_i^\prime+\mu_2e_i^\prime ) |y-y^*|_{vec}$,		
		where  $\mu_1w_i^\prime+\mu_2e_i^\prime=e_i^\prime G$ is the $i^{\text{th}}$ row of the matrix $G$. Condition \eqref{contraction_2}  is verified  element-wise.  When the nonlinear functions $f_i(\cdot)$ cannot be  expressed  in a vector form. e.g.  $f=(f_1,\dots,f_N)^\prime$, verification of \eqref{contraction_2} is helpful; see \eqref{nonlinear}-\eqref{stnar}. 
		Moreover, the condition $\rho(G)<1$ of Theorem~\ref{Thm. Ergodicity of nonlinear model} is implied by \eqref{contraction_2} when $\mu_1+\mu_2<1$, because $ \rho(G) \leq \vertiii{G}_{\infty} \leq \mu_{1} \vertiii{W}_\infty+\mu_{2} \leq \mu_{1}+\mu_{2}$, 
		since $\vertiii{W}_\infty=1$, by construction.  
		
		Theorem~\ref{Thm. Ergodicity of nonlinear continuous model} follows again  by \eqref{contraction_2}  but with  $\norm{\mu_s}$,  for $s={1,2}$ and assuming that  $\norm{\mu_1}+\norm{\mu_2}<1$.  Some illustrative examples are given below. 
		\begin{excont} 
			For model \eqref{nar_1},  $\lambda_t=\beta_0 1+ GY_{t-1}$, with $G=\beta_1W+\beta_2I$. In this case, the sharp condition $\rho(G)<1$ is easily verifiable and, under \eqref{contraction}, it implies the results of Theorem~\ref{Thm. Ergodicity of nonlinear model}. However,  the assumptions of the theorem are also satisfied by the set of sufficient conditions \eqref{contraction_2} with $\mu_1=\beta_1$, $\mu_2=\beta_2$ and $\beta_1+\beta_2 <1$, for integer-valued processes. For  the continuous-valued case, a similar argument shows  that  $\norm{\beta_1}+\norm{\beta_2}<1$. 
		\end{excont}
		\begin{excont} 
			Consider model \eqref{nonlinear}. By the mean value theorem (MVT)
			
			\begin{align} 
				\norm{f(x_i,y_i)-f(x^*_i,y^*_i)}&\leq  \max_{x_i\in\R_+}\norm{\frac{\partial f(x_i,y_i)}{\partial x_i}}\norm{x_i-x^*_i} + \max_{y_i\in\N}\norm{\frac{\partial f(x_i,y_i)}{\partial y_i}}\norm{y_i-y^*_i} \nonumber\\
				&\leq  \beta_1^*\norm{x_i-x^*_i}+\beta_2 \norm{y_i-y^*_i} \nonumber 
			\end{align}
			where $\beta_1^*=\max\left\lbrace \beta_1, \beta_0\gamma-\beta_1 \right\rbrace$. 
			Theorem~\ref{Thm. Ergodicity of nonlinear model} holds with  $G=\beta_1^*W+\beta_2I$  and $\beta_1^*+\beta_2<1$.
			Similarly to model \eqref{nonlinear}, by considering all the possible combinations of signs of $x,\beta_0$ and $\beta_1$ in model \eqref{nonlinear_cont}, we have $\norm{\partial f(x_i,y_i)/\partial x_i}=\norm{\beta_1-\beta_0\gamma/(1+x_i)^{\gamma+1}x_i/\norm{x_i}}\leq\bar{\beta}_1 \equiv \max\left\lbrace \norm{\beta_1}, \norm{\beta_0\gamma-\beta_1},\norm{\beta_1-\beta_0\gamma} \right\rbrace$. Theorem~\ref{Thm. Ergodicity of nonlinear continuous model} holds with $G=\bar{\beta}_1W+\norm{\beta_2}I$  and $\bar{\beta}_1+\norm{\beta_2}<1$.
		\end{excont}
		\begin{excont} 
			In the integer-valued case, 
			Theorem~\ref{Thm. Ergodicity of nonlinear model} applies to model  \eqref{stnar}  with $G=(\beta_1+\alpha)W+\beta_2I$ and $\beta_1+\alpha+\beta_2<1$, which coincides with the stationarity condition developed for the standard STAR model \cite{terasvirta1994}. 
			By considering all the possible combinations of signs for $\beta_1$ and $\alpha$, it is not difficult to show that  Theorem~\ref{Thm. Ergodicity of nonlinear continuous model} is verified, for model \eqref{stnar}, under the similar sufficient condition $\beta^\star_1+\norm{\beta_2}<1$, where $\beta^\star_1=\max\left\lbrace \norm{\beta_1}, \norm{\beta_1+\alpha} \right\rbrace $. 
		\end{excont}

		\begin{excont}
			The threshold model \eqref{tnar} does not satisfy the contraction conditions \eqref{contraction}-\eqref{contraction_2}. For the case of count data and $N$ fixed, we  develop a different proof   to show that $\{ Y_t \}$ is stationary and ergodic provided that  it has a positive conditional probability mass function and $\vertiii{G}_1<1$, where $G=(\beta_1+\alpha_1)W+(\beta_2+\alpha_2)I$. Analogous result holds also for continuous data; see Supp. Mat.~\ref{threshold models}.
		\end{excont}
		
		\subsection{Stability conditions for increasing network size} \label{stability conditions div N}
		
		In this section, following the works by \cite{zhu2017} and \cite{armillotta_fokianos_2021}, we  investigate the stability conditions of the process $\left\lbrace Y_t \in E^N \right\rbrace$, with $E=\R$ or $E=\N$, respectively, when the network size diverges ($N\to\infty$).  We use  a working definition of stationarity for increasing dimensional processes following \cite[Def.~1]{zhu2017}; see Supp. Mat.~\ref{Supp:Stationarity}. 
		
		\begin{theorem} 
			\label{Thm. Ergodicity of nonlinear model div N}
			Consider model \eqref{nonlinear_pnar} and $N\to\infty$. Define $G=\mu_1 W + \mu_2 I$, where $\mu_1$, $\mu_2\geq 0$ are constants such that $\mu_1+\mu_2<1$ and the contraction condition \eqref{contraction} holds, with $\max_{i\geq 1}f_i(0,0)<\infty$.
			Then, there exists a unique strictly stationary solution $\{ Y_t\in\N^N,~ t \in \mathbb{Z} \}$ to the nonlinear PNAR model, with 
			$\max_{i\geq 1}\E\norm{Y_{i,t}}^a\leq C_a<\infty$, for any $a\geq1$.
		\end{theorem}
		
		\begin{theorem} 
			\label{Thm. Ergodicity of nonlinear continuous model div N}
			Consider model \eqref{nonlinear_nar} and $N\to\infty$. Define $G=\norm{\mu_1} W + \norm{\mu_2} I$, where $\mu_1$, $\mu_2$ are real constants such that $\norm{\mu_1}+\norm{\mu_2}<1$ and the contraction condition \eqref{contraction} holds, with $\max_{i\geq 1}\norm{f_i(0,0)}<\infty$. Then, there exists a unique strictly stationary solution $\{ Y_t\in\R^N,~ t \in \mathbb{Z} \}$ to the nonlinear NAR model. In addition, if  $\max_{i\geq 1}\E\norm{\xi_{i,t}}^a\leq C_{\xi,a}<\infty$ for some $a\geq 8$, then $\max_{i\geq 1}\E\norm{Y_{i,t}}^a\leq C_a<\infty$.
		\end{theorem}
		
		Theorems~\ref{Thm. Ergodicity of nonlinear model div N}-\ref{Thm. Ergodicity of nonlinear continuous model div N} (whose proof is given in  Appendices \ref{Proof of Thm. Ergodicity of nonlinear model div N} and \ref{Proof of Thm. Ergodicity of nonlinear model div N cont})   extend the increasing network-type results of \cite[Thm.~1]{armillotta_fokianos_2021} and \cite[Thm.~2]{zhu2017} to nonlinear versions of the PNAR and NAR models, respectively. For models \eqref{nar_1}-\eqref{tnar}, $\max_{i\geq 1}\norm{f_i(0,0)}=\beta_0$. Moreover,  the contraction condition \eqref{contraction_2}, with $\mu_1+\mu_2<1$ ($\norm{\mu_1}+\norm{\mu_2}<1$), fulfils the conditions of   Theorem~\ref{Thm. Ergodicity of nonlinear model div N} (Theorem~\ref{Thm. Ergodicity of nonlinear continuous model div N}), i.e.  we obtain identical  sufficient conditions which guarantee  stationarity with fixed  and diverging $N$. We emphasize again   that   \eqref{tnar} does not satisfy  \eqref{contraction}-\eqref{contraction_2}. This  fact makes difficult to  show  stationarity,  when $N$ increases. More importantly, all  stability results do not depend on the network structure, as specified by the matrix $W$, and on the data generating process describing the joint dependence.
		
		\section{Quasi maximum likelihood inference} \label{SEC: inference}
		
		Consider model \eqref{general model}. Estimation for the unknown parameter vector $\theta$ is developed by means of QMLE. 
		Define the quasi-log-likelihood function for $\theta$ by
		\begin{equation}
			l_{NT}(\theta)=\sum_{t=1}^{T}\sum_{i=1}^{N} l_{i,t}(\theta)\,, 
			\label{log-lik}
		\end{equation}
		where $l_{i,t}(\theta)$ is the log-likelihood contribution of a single network node  whose form depends on the type of  data (discrete or continuous).
		Observe that \eqref{log-lik} is not necessarily the \emph{true} log-likelihood. The  QMLE is denoted by  $\hat{\theta}$ and maximizes  \eqref{log-lik}. It is obtained by  solving the system of equations $S_{NT}(\theta)=0$,  where
		\begin{equation}
			S_{NT}(\theta)=\frac{\partial l_{NT}(\theta)}{\partial\theta}=\sum_{t=1}^{T}s_{Nt}(\theta)
			\label{score}
		\end{equation}
		is the quasi-score function. Moreover define the following matrices
		\begin{equation}
			H_{NT}(\theta)=-\frac{\partial^2 l_{NT}(\theta)}{\partial\theta\partial\theta^\prime},\quad
			B_{NT}(\theta)= \sum_{t=1}^{T}\E\left(s_{Nt}(\theta)s_{Nt}^\prime(\theta) \mid 
			\Fb_{t-1}\right) \,,
			\label{hessians}
		\end{equation}
		as the sample Hessian matrix and the conditional information matrix, respectively. Henceforth we drop the dependence on $\theta$ when a quantity is evaluated at the true value $\theta_0$.

		\subsection{Inference for  PNAR models}
		\label{ex pnar}
		
		Consider model \eqref{nonlinear_pnar}. In this case the QMLE estimator, $\hat{\theta}$, maximizes
		\begin{equation}
			l_{NT}(\theta)=\sum_{t=1}^{T}\sum_{i=1}^{N}\Big(  Y_{i,t}\log\lambda_{i,t}(\theta)-\lambda_{i,t}(\theta)\Big) \,,
			\label{pois-log-lik}
		\end{equation}
		which is the log-likelihood obtained if all time series were contemporaneously independent. This simplifies computations allowing to establish consistency and asymptotic normality of the resulting estimator. It is worth noting that the joint copula structure, say $C(\cdot, \rho)$, with set of parameters $\rho$, do not enter into the maximization problem of the working log-likelihood \eqref{pois-log-lik}. However, this does not imply that inference does not take into account dependence among observations. 
		The corresponding score function is given by
		\begin{equation}
			S_{NT}(\theta)=\sum_{t=1}^{T}\sum_{i=1}^{N}\left(  \frac{Y_{i,t}}{\lambda_{i,t}(\theta)}-1\right) \frac{\partial\lambda_{i,t}(\theta)}{\partial\theta}=\sum_{t=1}^{T}s_{Nt}(\theta)\,.
			\label{score_poisson}
		\end{equation}
		Define $\partial\lambda_t(\theta)/\partial\theta^\prime$ the $N\times m$ matrix of derivatives, $D_t(\theta)$ the $N\times N$ diagonal matrix with elements equal to $\lambda_{i,t}(\theta)$, for $i=1,\dots,N$ and $\xi_t(\theta)=Y_t-\lambda_t(\theta)$ is a Martingale Difference Sequence ($MDS$) at $\theta=\theta_{0}.$Then, the empirical Hessian and conditional information matrices are given, respectively, by
		\begin{align}
			H_{NT}(\theta)&=\sum_{t=1}^{T}\sum_{i=1}^{N}\frac{Y_{i,t}}{\lambda_{i,t}^2(\theta)}\frac{\partial\lambda_{i,t}(\theta)}{\partial\theta}\frac{\partial\lambda_{i,t}(\theta)}{\partial\theta^\prime}-\sum_{t=1}^{T}\sum_{i=1}^{N}\left( \frac{Y_{i,t}}{\lambda_{i,t}(\theta)}-1\right)\frac{\partial^2\lambda_{i,t}(\theta)}{\partial\theta\partial\theta^\prime} \label{H_T}\,, \\
			B_{NT}(\theta)&=\sum_{t=1}^{T}\frac{\partial\lambda^\prime_t(\theta)}{\partial\theta}D^{-1}_t(\theta)\Sigma_t(\theta) D^{-1}_t(\theta)\frac{\partial\lambda_t(\theta)}{\partial\theta^\prime}\,,
			\label{B_T}
		\end{align}
		where $\Sigma_t(\theta)=\E\left( \xi_t(\theta)\xi_t^\prime(\theta)\left| \right. \mathcal{F}_{t-1} \right) $ is the conditional covariance matrix evaluated at $\theta$. We impose the following standard assumptions:
		
		\begin{enumerate}[label=\textbf{\Alph*}]
				\item	The parameter space $\Theta$ is compact and the true value $\theta_0$ belongs to its interior.
				\label{interior}
				\item For $i=1,\dots,N$	the function $f_i(\cdot)$ is three times differentiable with respect to $\theta$ and satisfies, for $x_i,x^*_i\in\R_+$ and $y_i,y_i^*\in\N$
				\begin{align}
					\norm{\frac{\partial f_i(x_i, y_i, \theta)}{\partial\theta_g}-\frac{\partial f_i(x^*_i, y^*_i, \theta)}{\partial\theta_g}}& \leq c_{1g}\norm{x_i-x^*_i}+c_{2g}\norm{y_i-y^*_i}\,,\quad g=1,\dots,m , \nonumber\\[0.3cm]
					\norm{\frac{\partial^2 f_i(x_i, y_i, \theta)}{\partial\theta_g\partial\theta_l}-\frac{\partial^2 f_i(x^*_i, y^*_i, \theta)}{\partial\theta_g\partial\theta_l}}& \leq c_{1gl}\norm{x_i-x^*_i}+c_{2gl}\norm{y_i-y^*_i}\,,\quad g,l=1,\dots,m , \nonumber\\[0.3cm]
					\norm{\frac{\partial^3 f_i(x_i, y_i, \theta)}{\partial\theta_g\partial\theta_l\partial\theta_s}-\frac{\partial^3 f_i(x^*_i, y^*_i, \theta)}{\partial\theta_g\partial\theta_l\partial\theta_s}}& \leq c_{1gls}\norm{x_i-x^*_i}+c_{2gls}\norm{y_i-y^*_i}\,,\,\, g,l,s=1,\dots,m .
					\nonumber
				\end{align}
				Furthermore, $\forall g,l,s$, $\max_{i\geq 1}\norm{\partial f_i(0, 0, \theta)/\partial\theta_g}<\infty$, $\max_{i\geq 1}\norm{\partial^2 f_i(0, 0, \theta)/\partial\theta_g\partial\theta_l}<\infty$, \break $\max_{i\geq 1}\norm{\partial^3 f_i(0, 0, \theta)/\partial\theta_g\partial\theta_l\partial\theta_s}<\infty$,  $\sum_g(c_{1g}+c_{2g})<\infty$, $\sum_{g,l}(c_{1gl}+c_{2gl})<\infty$, $\sum_{g,l,s}(c_{1gls}+c_{2gls})<\infty$. In addition, the components of $\partial f_i / \partial \theta$ are linearly independent.
				\label{smoothness}
				\item	For $i=1,\dots,N$, $f_i(x_i,y_i,\theta)\geq C > 0$, where $C$ is a generic constant.
				\label{bounded away from 0}
		\end{enumerate}
		
		Such regularity conditions have been employed in the literature to guarantee consistency and asymptotic normality of the QMLE in the context of nonlinear time series models; see \cite[Ch.~3]{tani2000}, among others.
		We now give additional assumptions employed for developing inference when \break $\left\lbrace N, T_N \right\rbrace \to\infty$ and the necessary network properties.
		Define
		\begin{equation}
		H_N(\theta)=\E\Bigg[\frac{\partial\lambda^\prime _{t}(\theta)}{\partial\theta}D_t^{-1}(\theta)\frac{\partial\lambda_{t}(\theta)}{\partial\theta^\prime }\Bigg]\,,
		\label{H}
		\end{equation}
		\begin{equation}
		B_N(\theta)=\E\Bigg[\frac{\partial\lambda^\prime _{t}(\theta)}{\partial\theta}D_t^{-1}(\theta)\Sigma_t(\theta) D_t^{-1}(\theta)\frac{\partial\lambda_{t}(\theta)}{\partial\theta^\prime }\Bigg]\,,
		\label{B}
		\end{equation}
		as, respectively, (minus) the expected Hessian matrix and the information matrix. Consider the following assumptions.
		\begin{enumerate}[label=\textbf{H\arabic*}]
		\item  The process $\left\lbrace \xi_t,\,\mathcal{F}_{t}:\,N\in\N, t\in\Z\right\rbrace $
		is $\alpha$-mixing with mixing coefficients $\{ \alpha(J)\}$.
		\label{Ass alpha mixing}
		\item  Define the standardized random process $\dot{Y}_t=D_t^{-1/2}(Y_t-\lambda_t)$. There exists a  non negative, non increasing sequence $\left\lbrace \varphi_h \right\rbrace_{h=1,\dots,\infty}$ such that $\sum_{h=1}^{\infty} h\varphi_h = \Phi<\infty$ and, for $i<j<k<l$, a.s. \label{Ass weak dependence}
		\begin{equation}
			\norm{\textrm{Cov}(\dot{Y}_{i,t}, \dot{Y}_{j,t} \dot{Y}_{k,t}\dot{Y}_{l,t}\left| \right. \Fb_{t-1} )}\leq \varphi_{j-i}\,, \quad \norm{\textrm{Cov}(\dot{Y}_{i,t} \dot{Y}_{j,t} \dot{Y}_{k,t}, \dot{Y}_{l,t}\left| \right. \Fb_{t-1} )}\leq \varphi_{l-k} \,,
			\nonumber 
		\end{equation}
		\begin{equation}
			\norm{\textrm{Cov}(\dot{Y}_{i,t} \dot{Y}_{j,t}, \dot{Y}_{k,t} \dot{Y}_{l,t}\left| \right. \Fb_{t-1} )}\leq \varphi_{k-j}\,, \quad \norm{\textrm{Cov}(\dot{Y}_{i,t}, \dot{Y}_{j,t}\left| \right. \Fb_{t-1} )}\leq \varphi_{j-i} \,. \nonumber
		\end{equation}
		\item  For model \eqref{nonlinear_pnar} with network $W$, the following limits exist, at $\theta=\theta_0$: \label{Ass limits existence}
		\begin{enumerate} [label*=\textbf{.\arabic*}]
			\item $\lim_{N\to\infty}N^{-1}H_{N}=H$, with $H$ a $m\times m$ positive definite matrix.
			\item $\lim_{N\to\infty}N^{-1}B_{N}=B$. 
			\item The third derivative of the quasi-log-likelihood \eqref{pois-log-lik} is bounded by functions $m_{i,t}$ which satisfy
			$\lim_{N\to\infty}N^{-1}\sum_{i=1}^{N}\E(m_{i,t})=M\,,$ where $M$ is a finite constant.
		\end{enumerate}
		\end{enumerate}
		
		Assumption \ref{Ass alpha mixing} is useful for studying  processes with dependent errors \cite{douk1994}. 
		When $N$ is fixed, a combination of Theorem~1-2 in \cite{doukhan_2012} and Remark~2.1 in \cite{Doukhanetal(2012)} shows  that the process $\left\lbrace \xi_t: t \in \mathbb{Z}\right\rbrace$, is $\alpha$-mixing, with exponentially decaying coefficients, provided that  $\vertiii{G}_1<1$. Analogous conclusion follow by  \cite[Prop.~3.1-3.4]{fok2020}. Condition \ref{Ass weak dependence} represents a contemporaneous weak dependence assumption.  Indeed, 
		even in the simple case of the independence model, i.e. $\lambda_{i,t}=\beta_0$, for all $i=1,\dots,N$, the reader can easily verify that, without any further constraints, $N^{-1}B_N=\mathcal{O}(N)$, so the limiting variance of the QMLE diverges. 
		Note that \ref{Ass weak dependence} does not guarantee  finiteness of the Hessian and information matrices, as $N\to\infty$. Such requirement is imposed by   Assumption \ref{Ass limits existence}. Obviously, such properties  depend on the structure of  $W$ and on the functional form of  $f(\cdot)$ in \eqref{nonlinear_pnar}; without the knowledge of these components it cannot be simplified any further.
		We present a detailed example involving the nonlinear PNAR model \eqref{nonlinear} in Section~\ref{SEC: detailed example} to offer further insight about \ref{Ass limits existence}.  Proofs for all the following results are given in  Supp. Mat. \ref{Supp:Sec:Proofs for Section 3}.


		\begin{lemma}
		Consider model \eqref{nonlinear_pnar} with $S_{NT}$, $H_{NT}$ and $B_{NT}$ defined as in \eqref{score_poisson},\eqref{H_T} and \eqref{B_T}, respectively. Let $\theta\in\Theta\subset\R^m_{+}$. Suppose the conditions of Theorem~\ref{Thm. Ergodicity of nonlinear model div N}, Assumption~\ref{smoothness}-\ref{bounded away from 0} and \ref{Ass alpha mixing}-\ref{Ass limits existence} hold. Then, as $ \left\lbrace N, T_N \right\rbrace \to\infty$
		\begin{enumerate} 
			\item $(NT_N)^{-1}H_{NT_N}\xrightarrow{p}H\,,$
			\item $(NT_N)^{-1}B_{NT_N}\xrightarrow{p}B\,,$
			\item $(NT_N)^{-\frac{1}{2}}S_{NT_N}\xrightarrow{d}N(0,B)\,,$
			\item  $\displaystyle \max_{g,l,s}\sup_{\theta\in\mathcal{O}(\theta_0)}\left|\frac{1}{NT_N}\sum_{t=1}^{T_N}\sum_{i=1}^{N}\frac{\partial^3l_{i,t}(\theta)}{\partial\theta_g\partial\theta_l\partial\theta_s}\right|\leq M_{NT_N}\xrightarrow{p}M\,,$
		\end{enumerate}
		where $M_{NT_N}\coloneqq (NT_N)^{-1}\sum_{t=1}^{T_N}\sum_{i=1}^{N}m_{i,t}$ and $\mathcal{O}(\theta_0)=\left\lbrace \theta:|\theta-\theta_0|_2<\delta\right\rbrace$ is a neighbourhood  of $\theta_0$.  
		\label{Lem. Consistency and Asymptotic Normality nonlinaer PNAR}
		\end{lemma}
		\begin{theorem} \label{Thm. Consistency and Asymptotic Normality nonlinaer PNAR}
		For model \eqref{nonlinear_pnar}, suppose that Assumption~\ref{interior} and conditions of Lemma~\ref{Lem. Consistency and Asymptotic Normality nonlinaer PNAR} hold. Then, there exists a fixed open neighbourhood $\mathcal{O}(\theta_0) 
		=\left\lbrace \theta:|\theta-\theta_0|_2<\delta\right\rbrace$
		of $\theta_0$ such that  with probability tending to 1, as $ \left\lbrace N, T_N \right\rbrace \to\infty$, the equation $S_{NT_N}(\theta)=0$ has a unique solution, denoted by $\hat{\theta}$, such that $\hat{\theta}\xrightarrow{p}\theta_0$
		and $\sqrt{NT_N}(\hat{\theta}-\theta_0)\xrightarrow{d}N(0,H^{-1}BH^{-1})$.
		\end{theorem}
		Thm~\ref{Thm. Consistency and Asymptotic Normality nonlinaer PNAR} follows by Lemma~\ref{Lem. Consistency and Asymptotic Normality nonlinaer PNAR} as proved by \cite[Sec.~S-3.3]{armillotta_fokianos_2021}. In addition, it extends  the results of \cite[Thm.~3]{armillotta_fokianos_2021} to nonlinear Poisson NAR models. The novelty   of Theorem~\ref{Thm. Consistency and Asymptotic Normality nonlinaer PNAR} 
		is that  both $N$ and $T$ tend to infinity as opposed to  standard case (when $N$ is fixed). Additional   conditions  guarantee strong consistency of the estimators, i.e. 
		
		\begin{theorem} \label{Thm. Consistency and Asymptotic Normality nonlinaer PNAR as}
		If $T_N=\lambda N$, for some $\lambda>0$ and Assumption~\ref{Ass alpha mixing} is such that the mixing coefficients satisfy $\alpha(J)^{1-1/r}=\mathcal{O}(J^{-3-\epsilon})$, for some $r>2$  and some $\epsilon>0$, then,  as $ \left\lbrace N, T_N \right\rbrace \to\infty$, all the convergences $``\xrightarrow{p}"$ in Lemma~\ref{Lem. Consistency and Asymptotic Normality nonlinaer PNAR} are replaced by  $``\xrightarrow{a.s.}"$ and Theorem~\ref{Thm. Consistency and Asymptotic Normality nonlinaer PNAR} holds with $\hat{\theta}\xrightarrow{a.s.}\theta_0$.
		\end{theorem}
		
		For instance,  exponential decay of the mixing coefficients  $\alpha(J)$ satisfies the assumption for any $r$ and $\epsilon$. Theorem \ref{Thm. Consistency and Asymptotic Normality nonlinaer PNAR as} is a new result, to the best of our knowledge, as 
		strong laws of large numbers for generally dependent double-indexed processes are scarce
		in the literature   (for an exception,  see \cite{dejong_1996strong}); see the discussion in \cite[Com.~6]{andrews_1987} and \cite[p.~256]{andrews_1992generic}. It is pointed out again  that the proof of Theorem~\ref{Thm. Consistency and Asymptotic Normality nonlinaer PNAR} 
		does not depend on the specification of the data generating process for the joint dependence of $\left\lbrace Y_t\right\rbrace $. 

		\subsection{A detailed example} \label{SEC: detailed example}
		We give a detailed discussion for proving Theorem~\ref{Thm. Consistency and Asymptotic Normality nonlinaer PNAR} for the nonlinear PNAR model \eqref{nonlinear} case.
		Let 
		$\Sigma_{\xi}=\E\norm{\xi_t\xi_t^\prime }_{vec}$ and 
		$\lambda_{\max}(X)$ the largest absolute eigenvalue of an arbitrary symmetric matrix $X$. Consider the vector form of model \eqref{nonlinear}: $\lambda_t= \beta_0 C_{t-1}+ GY_{t-1}$, where $G=\beta_1W+\beta_2I$ and $C_{t-1}=(1+X_{t-1})^{-\gamma}$. Under the conditions of Theorem~\ref{Thm. Ergodicity of nonlinear model div N}, and by using a infinite backward substitution argument on $Y_{t-1}$ we can rewrite the model as $Y_t=\mu^{\infty}_{t-1}+\tilde{Y}_{t}$, where $\mu^{\infty}_{t-1}=\beta_0\sum_{j=0}^{\infty}G^jC_{t-1-j}$ and $\tilde{Y}_{t}=\sum_{j=0}^{\infty}G^j\xi_{t-j}$. The proof of such representation is given in Supp. Mat. \ref{Proof of Prop condition limits nonlinear pnar}.  Define the following quantities: $L_{t-1}=\log(1+X_{t-1})$, $E_{t-1}= C_{t-1} \odot L_{t-1}$, $F_{t-1}= \log^2(1+X_{t-1}) \odot C_{t-1}$, $J_{t-1}= \log^3(1+X_{t-1}) \odot C_{t-1}$, where $\odot$ is the Hadamard product \cite[Sec.~11.7]{seber2008}; $I_{1,t-1}=C_{t-1}$, $I_{2,t-1}=W \mu^{\infty}_{t-1}$, $I_{3,t-1}=\mu^{\infty}_{t-1}$, $\Lambda_t=\Sigma_t^{1/2}D^{-1}_t$, $\Gamma(0)=\E[\Lambda_t(Y_{t-1}-\mu^{\infty}_{t-1})(Y_{t-1}-\mu^{\infty}_{t-1})^\prime \Lambda^\prime _t]$ and $\Delta(0)=\E[\Lambda_t W (Y_{t-1}-\mu^{\infty}_{t-1})(Y_{t-1}-\mu^{\infty}_{t-1})^\prime W^\prime \Lambda^\prime _t]$, $\Gamma_{1,t-1}=I_{1,t-1}$, $\Gamma_{2,t-1}=X_{t-1}$, $\Gamma_{3,t-1}=Y_{t-1}$, $\Gamma_{4,t-1}=E_{t-1}$; moreover let 
		$(j^*,l^*,k^*)=\argmax_{j,l,k}\norm{{N}^{-1}\sum_{i=1}^{N}\partial^3l_{i,t}(\theta)/\partial\theta_j\partial\theta_l\partial\theta_k}$, $\Pi_{jlk}=N^{-1}\sum_{i=1}^{N}\E(\Gamma_{i,j,t-1}\Gamma_{i,l,t-1}\Gamma_{i,k,t-1}/\lambda_{i,t})$, $\Pi_{F,k}=N^{-1}\E(F_{t-1}^\prime D_t^{-1}\Gamma_{k,t-1})$, $\Pi_{J}=N^{-1}\E(J_{t-1}^\prime D_t^{-1}\norm{\xi_t}_{vec})$, for $j,l,k=1,\dots,4$. Consider the following assumptions.
		
		\begin{enumerate}[label=\textbf{Q\arabic*}]
		\item Let $W$ be a sequence of matrices with non-stochastic entries indexed by $N$. \label{Ass network}
		\begin{enumerate}[label*=\textbf{.\arabic*}]
			\item Consider $W$ as a transition probability matrix of a Markov chain, whose state space is defined as the set of all the nodes in the network (i.e., $\left\lbrace 1,\dots, N\right\rbrace$). The Markov chain is assumed to be irreducible and aperiodic. Further, define $\pi = (\pi_1,\dots, \pi_N)^\prime \in\R^N$ as the stationary distribution of the Markov chain, where $\pi\geq 0$, $\sum_{i=1}^{N}\pi_i=1$ and $\pi=W^\prime \pi$. Furthermore, assume that $\lambda_{\max}(\Sigma_\xi)\sum_{i=1}^{N}\pi_i^2\to0$ as $N\to\infty$.
			\item Define $W^*=W+W^\prime $  and assume $\lambda_{\max}(W^*)=\mathcal{O}(\log N)$ and $\lambda_{\max}(\Sigma_{\xi})=\mathcal{O}((\log N)^\delta)$, for some $\delta\geq 1$.
		\end{enumerate}
		\item  Assume that the following limits exist: 
		$l^B_{kj}=\lim_{N\to\infty}N^{-1}\E( I_{k,t-1}^\prime \Lambda_t^\prime \Lambda_t I_{j,t-1})$, for $k,j=1,\dots,4$, $u^B_1=\lim_{N\to\infty}N^{-1}\textrm{tr}[\Delta(0)] $, $u^B_2=\lim_{N\to\infty}N^{-1}\textrm{tr}[W\Gamma(0)] $, $u^B_3=\lim_{N\to\infty}N^{-1}\textrm{tr}[\Gamma(0)] $, $v^B_{k4}=\lim_{N\to\infty}N^{-1}\E( \Gamma_{k,t-1}^\prime \Lambda_t^\prime \Lambda_t E_{t-1})$, $d^*=\lim_{N\to\infty}\Pi_{j^*,l^*,k^*}$.
		If at least two indices  among $(j^*,l^*,k^*)$ equal 4, $d_F^*=\lim_{N\to\infty}\Pi_{F,s^*}$, where $s^*=j^*,l^*,k^*$. Moreover, if  all three  $(j^*,l^*,k^*)$ equal 4, $d_J^*=\lim_{N\to\infty}\Pi_{J}$. \label{Ass final limits}
		
		\end{enumerate}
		
		\begin{theorem}   \label{Prop condition limits nonlinear pnar}
		Consider   \eqref{nonlinear} and suppose the conditions of Theorem~\ref{Thm. Ergodicity of nonlinear model div N}, Assumption~\ref{interior}-\ref{bounded away from 0},\ref{Ass alpha mixing}-\ref{Ass weak dependence} and \ref{Ass network}-\ref{Ass final limits} hold. Then, the conclusions of Theorem~\ref{Thm. Consistency and Asymptotic Normality nonlinaer PNAR} hold true for model \eqref{nonlinear}, with corresponding  limiting matrices: 
		\begin{equation}
			H=\begin{pmatrix}
				l^H_{11} & l^H_{12} & l^H_{13} & -\beta_0v^H_{14}  \\
				&  l^H_{22}+u^H_1 & l^H_{23}+u^H_2 & -\beta_0v^H_{24}\\
				&  & l^H_{33}+u^H_3 & -\beta_0v^H_{34} \\
				& & & \beta^2_0v^H_{44}
			\end{pmatrix}
			\,,\,\,B=\begin{pmatrix}
				l^B_{11} & l^B_{12} & l^B_{13} & -\beta_0v^B_{14}  \\
				&  l^B_{22}+u^B_1 & l^B_{23}+u^B_2 & -\beta_0v^B_{24}\\
				&  & l^B_{33}+u^B_3 & -\beta_0v^B_{34} \\
				& & & \beta^2_0v^B_{44}
			\end{pmatrix} \,, \label{H,B div N log}
		\end{equation}
		where the elements of the Hessian matrix are obtained by the elements  of  the information matrix  with $\Sigma_t=D_t$.
		\end{theorem}
		
		\begin{rem} \label{Rem: network impact} \rm
		Clearly, the network structure influences    the  results of Theorem~\ref{Prop condition limits nonlinear pnar}. Indeed, Assumption~\ref{Ass network} requires a well-behaved  underlying network: i) there should exists a non-zero probability to connect each pair of nodes; this allows the network to converge to its stationary distribution, i.e. $\lim_{N\to\infty}W^N=1\pi^\prime$. ii) The growth of the network should be such that  certain regularity properties hold. For instance, the covariances of the errors do not diverge fast, as $N\to\infty$.  The proof in Supp. Mat. \ref{Proof of Prop condition limits nonlinear pnar} shows that the leading  terms of Hessian and information matrices depend on the error component $\xi_{t-j}$ and the pseudo covariance matrix $\Sigma_\xi$  and are asymptotically negligible  (compare also with \cite[Lem.~S-1]{armillotta_fokianos_2021}). In this way,  the remaining terms appearing in  Assumption~\ref{Ass final limits} show existence of the limiting Hessian matrix $H$ and (together with Assumption~\ref{Ass weak dependence}) of the limiting information $B$. 
		Without any assumptions for the network,  the structure of matrices $H$ and $B$ is unknown and  conditions of finiteness of the limiting matrices could not be specified explicitly. 
		\end{rem}

		\subsection{Inference for NAR models}
		\label{ex nar}

		In this case, define $\hat{\theta}$, as the maximizer of the least squares criterion
		\begin{equation} 
		l_{NT}(\theta)=-\sum_{t=1}^{T}\left( Y_t-\lambda_t(\theta)\right)^\prime\left( Y_t-\lambda_t(\theta)\right)\,.
		\label{quasi log-lik nar}
		\end{equation}
		It follows that 
		\begin{equation}
		S_{NT}(\theta)=\sum_{t=1}^{T}\frac{\partial\lambda^\prime_t(\theta)}{\partial\theta}\left( Y_t-\lambda_t(\theta)\right)=\sum_{t=1}^{T}s_{Nt}(\theta).
		\label{score_cont}
		\end{equation}
		The empirical Hessian and information matrices are respectively
		\begin{align}
		H_{NT}(\theta)&=\sum_{t=1}^{T}\sum_{i=1}^{N}\frac{\partial\lambda_{i,t}(\theta)}{\partial\theta}\frac{\partial\lambda_{i,t}(\theta)}{\partial\theta^\prime}-\sum_{t=1}^{T}\sum_{i=1}^{N}\left( Y_{i,t}-\lambda_{i,t}(\theta)\right)\frac{\partial^2\lambda_{i,t}(\theta)}{\partial\theta\partial\theta^\prime} \label{H_T_cont}\,, \\
		B_{NT}(\theta)&=\sum_{t=1}^{T}\frac{\partial\lambda^\prime_t(\theta)}{\partial\theta}\Sigma_t(\theta) \frac{\partial\lambda_t(\theta)}{\partial\theta^\prime}\,,
		\label{B_T_cont}
		\end{align}
		where notation is as in  Section~\ref{ex pnar}. In addition  
		\begin{equation}
		H_N(\theta)=\E\left( \frac{\partial\lambda^\prime_{t}(\theta)}{\partial\theta}\frac{\partial\lambda_{t}(\theta)}{\partial\theta^\prime}\right) \,,\label{H_cont}
		\end{equation}
		\begin{equation}
		B_N=\E\left( \frac{\partial\lambda^\prime_t(\theta)}{\partial\theta}\xi_t(\theta)\xi^\prime_t(\theta) \frac{\partial\lambda_t(\theta)}{\partial\theta^\prime}\right) \,, \label{B_cont}
		\end{equation}
		and the latter equals $\sigma^2H_N$, when  $\theta=\theta_0$, because  $\xi_t$ an is $IID(0,\sigma^2)$ process. For the same reasons Assumption~\ref{Ass alpha mixing}-\ref{Ass weak dependence} hold trivially. Assumption \ref{Ass limits existence} is modified as 
		\begin{enumerate} [label=\textbf{H3$^\prime$}]
		\item For model \eqref{nonlinear_nar} with network $W$, the following limits exist, at $\theta=\theta_0$: \label{Ass limits existence cont}
		\begin{enumerate} [label=\textbf{H3$^\prime$.\arabic*}]
			\item $\lim_{N\to\infty}N^{-1}H_{N}=H$, with $H$ a $m\times m$ positive definite matrix.
			\item The third derivative of the quasi-log-likelihood \eqref{quasi log-lik nar} is bounded by functions $m_{i,t}$ which satisfy
			$\lim_{N\to\infty}N^{-1}\sum_{i=1}^{N}\E(m_{i,t})=M\,,$ where $M$ is a finite constant.
		\end{enumerate}
		\end{enumerate}
		
		\begin{theorem} 
		\label{Thm. Consistency and Asymptotic Normality nonlinaer NAR}
		Consider model \eqref{nonlinear_nar} with $S_{NT}$, $H_{NT}$ and $B_{NT}$ defined as in \eqref{score_cont},\eqref{H_T_cont} and \eqref{B_T_cont}, respectively. Let $\theta\in\Theta\subset\R^m$. Suppose that the conditions of  Theorem~\ref{Thm. Ergodicity of nonlinear continuous model div N}, Assumption~\ref{interior}-\ref{smoothness} and \ref{Ass limits existence cont} hold. Then, there exists a fixed open neighbourhood $\mathcal{O}(\theta_0)=\left\lbrace \theta:|\theta-\theta_0|_2<\delta\right\rbrace$ of $\theta_0$ such that with probability tending to 1 as $ \left\lbrace N, T_N \right\rbrace \to\infty$, the equation $S_{NT_N}(\theta)=0$ has a unique solution, denoted by $\hat{\theta}$, such that $\hat{\theta}\xrightarrow{p}\theta_0$ and $\sqrt{NT_N}(\hat{\theta}-\theta_0)\xrightarrow{d}N(0,B^{-1})$,
		where $B=\sigma^2H$ and $H$ is defined as in \eqref{H,B div N log}, with $\Sigma_t=D_t=I$. 
		\end{theorem}
		The proof is omitted since it is analogous to the proof of Theorem~\ref{Thm. Consistency and Asymptotic Normality nonlinaer PNAR}. Theorem~\ref{Thm. Consistency and Asymptotic Normality nonlinaer NAR} generalises the results of \cite[Thm.~3]{zhu2017} to nonlinear NAR models, and it can be proved to entail results analogous to Proposition~\ref{Prop condition limits nonlinear pnar}, by considering \cite[Assumption~C2]{zhu2017} instead of \ref{Ass network}, and \ref{Ass final limits} holding, with $D_t=\Lambda_t=I$ and $\Pi_{jlk}=0$, for $j,k,l=1,\dots,4$. See \cite[Thm.~3]{zhu2017} for a detailed proof concerning the case of model \eqref{nar_1}.
		A result similar to Theorem~\ref{Thm. Consistency and Asymptotic Normality nonlinaer PNAR as}
		for model \eqref{nonlinear_nar} is also established by setting $T_N=\lambda N$.
		
		\begin{rem} \rm
		Reiterating the discussion following 	Theorems~\ref{Thm. Ergodicity of nonlinear model div N}-\ref{Thm. Ergodicity of nonlinear continuous model div N}  and noting that Assumption~\ref{smoothness} does  not hold  for \eqref{tnar}, the double asymptotic based inference derived in this section and the associated testing theory (see Section~\ref{SEC: identifiability}) do not hold  for  threshold models  when $N$ is increasing.  However, Supp. Mat.~\ref{threshold models} provides all these results for the threshold model if $N$ is fixed.
		\label{Remark:TNAR}
		\end{rem}
		
		\begin{rem} \rm
		The asymptotic theory of  this section applies for  parameter values satisfying the  conditions of  Theorems~\ref{Thm. Ergodicity of nonlinear model div N}-\ref{Thm. Ergodicity of nonlinear continuous model div N}.
		In practical applications, the QMLE is obtained  using  constrained optimization where the constraints satisfy such  conditions. In the integer-valued case,  additional constraints should be introduced so that the mean process is positive.
		\end{rem}
		
		\section{Hypothesis testing on network autoregressive models} \label{SEC: linearity test}
		
		With the same notation as in Sections~\ref{SEC: model} and \ref{SEC: inference}, recall  \eqref{general model} 	and consider the following testing problems 
		\begin{equation}
		H_0: \theta^{(2)}=\theta_0^{(2)}\quad\text{vs.}\quad H_1: \theta^{(2)}\neq \theta_0^{(2)},\quad \text{componentwise}\, ,
		\label{test}
		\end{equation}
		against the Pitman's local alternatives
		\begin{equation}
		H_0: \theta^{(2)}=\theta_0^{(2)}\quad\text{vs.}\quad H_{1}: \theta^{(2)}= \theta_0^{(2)}+\frac{\delta_2}{\sqrt{NT}}, ~~ \delta_2\in\R^{m_2} . 
		\label{test_local}
		\end{equation}
		
		To develop a test statistic for testing \eqref{test}-\eqref{test_local}, we  employ a quasi-score test based on \eqref{log-lik}. An appealing property of the score test  is that it is computed   under the null, which is computationally simpler.  Moreover, the asymptotic distribution of the test is not affected  when  $\theta^{(2)}$ belongs to the boundary of the parameter space.
		Define  $\tilde{\theta}=( \tilde{\theta}^{(1)\prime}, \tilde{\theta}^{(2)\prime})^\prime$ the constrained quasi-likelihood estimator of  $\theta=(\theta^{(1)\prime}, \theta^{(2)\prime})^\prime$, under the null hypothesis, and $S_{NT}(\theta)=( S^{(1)\prime}_{NT}(\theta), S_{NT}^{(2)\prime}(\theta)) ^\prime$ denote the corresponding partition of the quasi-score function. Because we study a quasi-score test, we correct the test statistic to obtain thoroughly its limiting distribution;  see 
		\cite{fran2019}, among others. Accordingly the test statistic is given by (Supp. Mat. \ref{Supp:Sec:Proofs for Section 4})
		\begin{equation}
		LM_{NT}= S^{(2)\prime}_{NT}(\tilde{\theta})\Sigma_{NT}^{-1}(\tilde{\theta})  S^{(2)}_{NT}(\tilde{\theta}).
		\label{score test}
		\end{equation}
		Here  
		$(NT)^{-1}\Sigma_{NT}(\tilde{\theta})$ is a suitable estimator for the covariance matrix defined as \break 
		$\Sigma=\mathrm{Var}[(NT)^{-1/2}S_{NT}^{(2)}(\tilde{\theta})]$, where 
		\begin{equation}
		\Sigma=B_{22}-H_{21}H_{11}^{-1}B_{12}-B_{21}H_{11}^{-1}H_{12}+H_{21}H_{11}^{-1}B_{11}H_{11}^{-1}H_{12}\,,
		\label{sigma}
		\end{equation}
		with  $B_{gl}, H_{gl}, g,l=1,2$ with dimension $m_g\times m_l$, are blocks of the matrices $H,B$ 
		such that 
		\begin{equation*}
		H=
		\begin{pmatrix}
			H_{11} & H_{12} \\
			H_{21} & H_{22}
		\end{pmatrix}\,,\quad
		B=
		\begin{pmatrix}
			B_{11} & B_{12} \\
			B_{21} & B_{22}
		\end{pmatrix}\,.
		\end{equation*}
		If \eqref{log-lik} is the \emph{true} likelihood, then $LM_{NT}$ reduces to the standard score test with $B\equiv H$ and $\Sigma=B_{22}-B_{21}B_{11}^{-1}B_{12}\eqqcolon \Sigma_B$. 
		\begin{rem}  \label{Rem. sample information matrix} \rm
		Following \cite{armillotta_fokianos_2021}, the estimator $\Sigma_{NT}(\tilde{\theta})$ of \eqref{score test} is computed as the sample counterpart of \eqref{sigma}, obtained by replacing the partitioned matrices $H$ and $B$ respectively by $H_{NT}(\tilde{\theta})$ and $B_{NT}(\tilde{\theta})$, where $H_{NT}(\theta)$ is defined in \eqref{hessians} and $B_{NT}(\theta)$ is the sample information matrix. 
		\end{rem}

		\begin{enumerate} [label=\textbf{A$^\prime$}] 
		\item The parameter space $\Theta$ is compact. Define the partition of the parameter space $\Theta^{(1)}$ such that $\theta^{(1)}\in\Theta^{(1)}$ and the true value $\theta^{(1)}_0$ belongs to its interior.\label{interior H0}
		\end{enumerate}
		\begin{theorem} 
		Suppose that model \eqref{general model} admits a stationary solution, for $N\to\infty$. Consider $l_{NT}$,  $S_{NT}$, $H_{NT}$ and $B_{NT}$ defined by \eqref{log-lik}, \eqref{score} and \eqref{hessians}, respectively. Assume that, under $H_0$, \ref{interior H0} is satisfied  such that, as  $ \left\lbrace N, T_N \right\rbrace \to\infty$, Lemma \ref{Lem. Consistency and Asymptotic Normality nonlinaer PNAR} and Theorem \ref{Thm. Consistency and Asymptotic Normality nonlinaer PNAR} holds.
		Recall the testing problem \eqref{test}. Then, as $ \left\lbrace N, T_N \right\rbrace \to\infty$, the quasi-score test statistic \eqref{score test} converges to a chi-square random variable,
		\begin{equation}
			LM_{NT_N}\xrightarrow{d}\chi^2_{m_2}\,,
			\nonumber
		\end{equation}
		under $H_0$. Moreover, under the alternative  \eqref{test_local},  \eqref{score test} converges to a non-central chi-square random variable,
		\begin{equation}
			LM_{NT_N}\xrightarrow{d}\chi^2_{m_2}(\delta_2^{\prime}\tilde{\Delta}\delta_2)\,,
			\nonumber
		\end{equation} 
		where $\tilde{\Delta}=\tilde{\Sigma}_H\tilde{\Sigma}^{-1}\tilde{\Sigma}_H$ and $ \Sigma_H\coloneqq H_{22}-H_{21}H_{11}^{-1}H_{12}$; $\tilde{\Sigma}$ and $\tilde{\Sigma}_H$ are sample counterparts of $\Sigma$ and $\Sigma_H$, respectively, evaluated at $\tilde{\theta}$. 
		\label{limits}
		\end{theorem}
		
		Theorem~\ref{limits} extends the results of \cite{christou_fokianos_2015} for the case of multivariate discrete and continuous network autoregressive models with infinite dimensional data. In addition, it  implies that even though  $\theta^{(2)}$ belongs to the boundary of the parameter space, the asymptotic $\chi^2$ distribution remains unaffected. Instead, the asymptotic distribution of the Wald and likelihood ratio tests depends on the  null hypothesis  and do not
		converge  to  $\chi^2$ distributed when $N$ is fixed; see  \cite[Sec.~8.3.2]{fran2019} and \cite{fra2016}. We illustrate some  applications  of Theorem~\ref{limits} to the network models \eqref{nonlinear_pnar}-\eqref{nonlinear_nar} but we emphasize  that its conclusion applies to more general settings. 

		\begin{proposition}  \label{pnar chi general}
		Assume $Y_t$ follows \eqref{nonlinear_pnar} and the process $\lambda_t$ is defined as in $\eqref{general model}$.
		Consider the test $H_0: \theta^{(2)}=\theta_0^{(2)}$ vs. $H_1: \theta^{(2)}\neq \theta_0^{(2)}$. Then, under $H_0$, \ref{interior H0} and the conditions of Lemma~\ref{Lem. Consistency and Asymptotic Normality nonlinaer PNAR}, Theorem \ref{limits} is true. 
		
		\end{proposition}
		Proposition \ref{pnar chi general} follows by Lemma~\ref{Lem. Consistency and Asymptotic Normality nonlinaer PNAR}, Theorems~\ref{Thm. Ergodicity of nonlinear model div N} 
		and \ref{limits}. 
		
		\subsection{A detailed example (continued)} \label{SEC: detailed example conti}
		For model \eqref{nonlinear},  the linearity test \eqref{test} is equivalent to testing  $H_0: \gamma=0$ versus $H_1: \gamma>0$. Convergence for all necessary asymptotic quantities is required only under the null. Recall the notation of Section~\ref{SEC: detailed example}. Then, under $H_0$, $C_{t-1}=1$, so the decomposition of the count process simplifies to $Y_t=\mu+\sum_{j=0}^{\infty}G^j\xi_{t-j}$, since $\mu_{t-1}=\mu=\beta_0/(1-\beta_1-\beta_2)^{-1}$. This entails that
		$\Gamma_{1,t-1}=1$, $\Gamma_{4,t-1}=L_{t-1}$, $F_{t-1}=\log^2(1+X_{t-1})$ and $J_{t-1}=\log^2(1+X_{t-1})$. Moreover, set  $\Lambda=\E(\Lambda^\prime _t\Lambda_t)$, $\Gamma(0)=\E[\Lambda_t(Y_{t-1}-\mu)(Y_{t-1}-\mu)^\prime \Lambda^\prime _t]$ and $\Delta(0)=\E[\Lambda_tW(Y_{t-1}-\mu)(Y_{t-1}-\mu)^\prime W^\prime \Lambda^\prime _t]$. 
		So condition \ref{Ass final limits} 
		simplifies as follows: 
		\begin{enumerate} [label=\textbf{Q2$^\prime$}]
		\item Assume that the following limits exist: $f_1=\lim_{N\to\infty}N^{-1}\left( 1^\prime \Lambda 1\right)$, \\ $f_2=\lim_{N\to\infty}N^{-1}\textrm{tr}\left[ \Gamma(0)\right] $, $f_3=\lim_{N\to\infty}N^{-1}\textrm{tr}\left[ W\Gamma(0)\right] $, $f_4=\lim_{N\to\infty}N^{-1}\textrm{tr}\left[ \Delta(0)\right] $, $\dot{v}^B_{k4}=\lim_{N\to\infty}N^{-1}\E( \Gamma_{k,t-1}^\prime \Lambda_t^\prime \Lambda_t L_{t-1})$, $d^*=\lim_{N\to\infty}\Pi_{j^*,l^*,k^*}$.
		If at least two indices among $(j^*,l^*,k^*)$ equal 4, $d_F^*=\lim_{N\to\infty}\Pi_{F,s^*}$, where $s^*=j^*,l^*,k^*$. Moreover, if all three indices $(j^*,l^*,k^*)$ equal 4, $d_J^*=\lim_{N\to\infty}\Pi_{J}$.  \label{Ass final limits H0}
		\end{enumerate}
		In this case, the limiting Hessian and information matrices in \eqref{H,B div N log} are  equal to the respective matrices obtained by  the  linear model fitting  \cite[Eq.~22]{armillotta_fokianos_2021} plus the addition of the fourth row and column whose elements are given by $(-\beta_0)^\nu\dot{v}^B_{k4}$, for $k=1,\dots,4$ and $\nu=2$, when $k=4$ and $\nu=1$, otherwise.
		\begin{proposition}  \label{pnar chi}
		Assume $Y_t$  ~follows \eqref{nonlinear_pnar} and the process $\lambda_t$ is defined as in $\eqref{nonlinear}$. Suppose the conditions of Theorem~\ref{Thm. Ergodicity of nonlinear model div N}, Assumptions~\ref{interior H0}, \ref{smoothness}-\ref{bounded away from 0},\ref{Ass alpha mixing}-\ref{Ass weak dependence} and \ref{Ass network}-\ref{Ass final limits H0} hold. Consider the test $H_0: \gamma=0$ vs. $H_1: \gamma>0$. Then, Theorem \ref{limits} holds true.
		
		\end{proposition}
		Denoting the constrained QMLE  by $\tilde{\theta}=(\tilde{\theta}^{(1)},0)^\prime$, where $\tilde{\theta}^{(1)}$ is the QMLE of the linear model \eqref{nar_1}, the
		partial quasi-score  \eqref{score test} is given by  $S^{(2)}_{NT}(\tilde{\theta})=\sum_{t=1}^{T}\sum_{i=1}^{N}(Y_{i,t}/\lambda_{i,t}(\tilde{\theta})-1)\partial\lambda_{i,t}(\tilde{\theta})/\partial\gamma$, with $\partial\lambda_{i,t}(\tilde{\theta})/\partial\gamma=-\tilde{\beta_0}\log(1+X_{i,t-1})$, where $\tilde{\beta_0}$ is the QMLE of the intercept $\beta_0$ in the linear model \eqref{nar_1}. Furthermore, the covariance estimator $\Sigma_{NT}(\tilde{\theta})$ for the test statistic \eqref{score test} is defined as in Remark~\ref{Rem. sample information matrix}.  Note that in Proposition~\ref{pnar chi}, the nonlinear perturbation is due to the network structure. Moreover,  since the asymptotic distribution  of the  score test \eqref{score test} depends on the convergence of sample Hessian and information matrices to \eqref{H,B div N log},  the  approximation to the chi-square distribution  depends  by the convergence of the network according to the regularity properties given by \ref{Ass network}-\ref{Ass final limits H0} (see Remark~\ref{Rem: network impact}). Analogous result and conclusions are  obtained for \eqref{nonlinear_nar}, by using  Thm. ~\ref{Thm. Ergodicity of nonlinear continuous model div N}, \ref{Thm. Consistency and Asymptotic Normality nonlinaer NAR} and \ref{limits}, and therefore it is omitted. Consider the following condition:
		\begin{enumerate} [label=\textbf{Q2$^{\prime\prime}$}]
		\item For $Y_t$ defined as in \eqref{nonlinear_nar} and $\lambda_t$ following \eqref{nonlinear_cont}, Assumption~\ref{Ass final limits H0} holds, with $D_t=\Lambda_t=I$ and $\Pi_{jlk}=0$, for $j,k,l=1,\dots,4$. \label{Ass final limits H0 cont}
		\end{enumerate}
		\begin{proposition}  \label{nar chi}
		Assume $Y_t$ ~ follows \eqref{nonlinear_nar} and the process $\lambda_t$ is defined as in $\eqref{nonlinear_cont}$. Suppose the conditions of  Theorem~\ref{Thm. Ergodicity of nonlinear continuous model div N}, Assumptions~\ref{interior H0}-\ref{smoothness},\ref{Ass alpha mixing}-\ref{Ass weak dependence}, \cite[Assumption~C2]{zhu2017} and \ref{Ass final limits H0 cont} hold.  Consider the test $H_0: \gamma=0$ vs $H_1: \gamma>0$. 
		Then, Theorem \ref{limits} holds true.
		\end{proposition}

		\section{Testing under non identifiable parameters}
		\label{SEC: identifiability}
		We develop testing theory when the parameters are not identifiable under the
		linearity  hypothesis. A case in point is model \eqref{stnar}, with $\theta^{(1)}=(\beta_0, \beta_1, \beta_2)^\prime$ and $\theta^{(2)}=(\alpha, \gamma)^\prime$. Then testing  $H_0: \alpha=0$, makes $\gamma$ non-identifiable but 
		the score partition \eqref{score_poisson}--and consequently the test statistic-- still depends on the value of  $\gamma$. Hence the theory of Sec. \ref{SEC: linearity test} does not apply any more. Similar remarks hold for the threshold parameter $\gamma$ of model \eqref{tnar}, when testing $H_0: \alpha_0=\alpha_1=\alpha_2=0$. Assigning a fixed arbitrary value for $\gamma$ resolves such issues but this approach might lack  power as the test is sensitive to the choice of $\gamma$, especially when $\gamma$ is far from its true value. 
		It is well known (see  \cite[Sec.~5.1,5.5]{Terasvirtaetal(2010)} ) that testing linearity is an important issue because non-identifiable parameters have tremendous impact on properties of estimators. 
		Usually  a sup-type test, say  $g_{NT}=\sup_{\gamma\in \Gamma}LM_{NT}(\gamma)$, is employed in applications, where  $\Gamma=[\gamma_L,\gamma_U]$ is a compact domain for $\gamma$; 
		e.g. \cite{davies_1987, francq_etal_2010} and \cite[Par.~3.2]{christou_fokianos_2015}.
		
		Define $Z$ a random variable, and suppose that the function $g(\cdot):\Gamma\to\R$ is continuous with respect to the uniform metric, monotonic for each $\gamma$, and such that, as $Z\to\infty$, then $g(Z)\to\infty$ in a subset of $\Gamma$ with a non-zero probability $\mathbb{P}$.  For the standard asymptotics, i.e. $T\to\infty$, such functions have been employed in applications. Examples include  $g_{T}=g(LM_{T})$ \cite{hansen_1996} and  \cite{andrews_1994} who considered $g(LM_T)=\int_\Gamma LM_{T}(\gamma)d\mathbb{P}(\gamma)$ and $g(LM_{T})=\log(\int_\Gamma\exp(1/2 LM_{T}(\gamma))d\mathbb{P}(\gamma)) $.
		We extend this theory to the case of both $T,N \rightarrow \infty$. 
		
		
		\subsection{Specification} \label{Subsec: result non ident}
		In this section, we use a more convenient  notation.  Accordingly, consider 
		the nonlinear PNAR model defined in \eqref{nonlinear_pnar} as
		\begin{equation}
		Y_t=N_t(\lambda_t(\gamma))\,,\quad \quad	\lambda_t(\gamma)=Z_{1t}(W)\beta+h(Y_{t-1}, W, \gamma)\alpha,
		\label{general nuisance pnar}
		\end{equation}
		where $\beta$ is a $k_1\times 1$ vector of identifiable  parameters associated with the linear component of the model,  $\alpha$ is a $k_2\times 1$ vector of  identifiable  non-linear parameters and $\gamma$ denote nuisance parameters. We set 
		$\theta=(\phi^\prime, \gamma^\prime)^\prime$, $\phi=(\beta^\prime, \alpha^\prime)^\prime$. With this notation, the dimension of $\theta$ is $m=k+m^*$, where $m^*$ is the dimension of $\gamma$
		and $k=k_1+k_2$. In addition,   $Z_{1t}(W)=(1, WY_{t-1}, Y_{t-1})$ is a $N\times k_1$ matrix associated to the linear part of the network autoregressive model (for the order 1 model  $k_1=3$), and $h(Y_{t-1}, W, \gamma)$ is a $N\times k_2$ matrix describing the nonlinear part of the model.  Set $Z_{1t}=Z_{1t}(W)$, $h_t(\gamma)\equiv h(Y_{t-1}, W, \gamma)$ and $h_t(\gamma)=(h^{1}_t(\gamma)\dots h^{b}_t(\gamma) \dots h^{k_2}_t(\gamma))$, where each column indicates a nonlinear regressor $ h^{b}_t(\gamma)$, for $b=1,\dots,k_2$, being a $N\times 1$ vector whose elements are $h^{b}_{i,t}(\gamma)$, where $i=1,\dots,N$.  
		Then, the  conditional expectation of \eqref{general nuisance pnar} is  $\lambda_t(\gamma)=Z_{t}(\gamma)\phi$ where $Z_t(\gamma)=(Z_{1t}, h_t(\gamma))$ is the $N\times k$ matrix of regressors.
		Analogously, for continuous-valued time series and  $\xi_t\sim IID(0,\sigma^2)$, equation \eqref{nonlinear_nar} becomes
		\begin{equation}
		Y_t=\lambda_t(\gamma)+\xi_t\,,\quad \quad	\lambda_t(\gamma)=Z_{1t}(W)\beta+h(Y_{t-1}, W,  \gamma)\alpha.
		\label{general nuisance nar}
		\end{equation}
		
		Many nonlinear models  are included in this general frameworks provided by \eqref{general nuisance pnar}-\eqref{general nuisance nar}; for example, the STNAR model \eqref{stnar}, where $k_2=1$ and $h_{i,t}(\gamma)=\exp(-\gamma X_{i,t-1}^2)X_{i,t-1}$, for $i=1,\dots,N$, and the TNAR model \eqref{tnar}, where $k_2=3$ and $h^1_{i,t}(\gamma)=I(X_{i,t-1}\leq \gamma)$, $h^2_{i,t}(\gamma)=X_{i,t-1} I(X_{i,t-1}\leq \gamma)$ and $h^3_{i,t}(\gamma)=Y_{i,t-1} I(X_{i,t-1}\leq \gamma)$ ;  see  \cite[p.~414]{hansen_1996}. 
		
		\subsection{Testing linearity} \label{SUBSEC test non ident}
		For models \eqref{general nuisance pnar}-\eqref{general nuisance nar}, consider 
		testing  linearity  in the presence of  a non identifiable parameters $\gamma$
		\begin{equation}
		H_0: \alpha=0\,, \quad \text{vs.}\quad H_1: \alpha\neq0\,,\quad \text{elementwise}. \label{test nonident}
		\end{equation}
		
		Consider first  the case of count time series, i.e. eq. \eqref{general nuisance pnar}. In this case, the score \eqref{score_poisson},  Hessian \eqref{H_T}, and the sample information matrix \eqref{B_T}, for the quasi-log-likelihood \eqref{pois-log-lik} are
		$S_{NT}(\gamma)=\sum_{t=1}^{T} s_{Nt}(\gamma)
		\,,
		\, H_{NT}(\gamma_1, \gamma_2)=\sum_{t=1}^{T}\sum_{i=1}^{N}Y_{i,t}\bar{Z}_{i,t}(\gamma_1)\bar{Z}^\prime_{i,t}(\gamma_2),
  B_{NT}(\gamma_1, \gamma_2)=\sum_{t=1}^{T} \E\left[ s_t(\gamma_1)s_t^\prime(\gamma_2) | \Fb_{t-1} \right]$ 
where $s_{Nt}(\gamma)=Z_t^\prime(\gamma) D^{-1}_t(\gamma)\left( Y_t-Z_t(\gamma)\phi\right)$ and $\bar{Z}_{i,t}(\gamma)=Z_{i,t}(\gamma)/\lambda_{i,t}(\gamma)$. The theoretical counterpart of such quantities are then denoted by $H_N(\gamma_1, \gamma_2)=\sum_{i=1}^{N}\E\left( Y_{i,t}\bar{Z}_{i,t}(\gamma_1)\bar{Z}^\prime_{i,t}(\gamma_2)\right)$, $H(\gamma_1, \gamma_2) = \lim_{N\to\infty} N^{-1}H_N(\gamma_1, \gamma_2)$, and $B_N(\gamma_1,\gamma_2)=\E\left( s_t(\gamma_1)s^\prime_t(\gamma_2)\right)$, $B(\gamma_1, \gamma_2) = \lim_{N\to\infty} N^{-1}B_N(\gamma_1, \gamma_2)$. 
Following the discussion of Section~\ref{SEC: linearity test}, 
the quasi-score function is  partitioned again in two components:  the part concerning linear parameters and the component associated with the non-linear part of the model. 
We denote this by  $S_{NT}(\gamma)=(  S^{(1)\prime}_{NT}, S_{NT}^{(2)\prime}(\gamma)) ^\prime$.
Moreover, consider $S(\gamma)=\left( S^{(1)\prime}, S^{(2)\prime}(\gamma)\right) ^\prime$ a mean zero Gaussian process with covariance kernel $B(\gamma_1, \gamma_2)$. Define the matrix $\Sigma(\gamma_1,\gamma_2)$ as in \eqref{sigma}, with partitioned matrices $B_{gl}$, $H_{gl}$, for $g,l=1,2$, of dimension $k_g\times k_l$, being blocks of the matrices $B(\gamma_1,\gamma_2)$, $H(\gamma_1, \gamma_2)$, with obvious rearrangement of the notation. Then, $S^{(2)}(\gamma)$ is a Gaussian process with covariance kernel $\Sigma(\gamma_1, \gamma_2)$.  Define $\tilde{\phi}=(\tilde{\beta}^{\prime},0^{\prime})^{\prime}$ the constrained estimator under the null hypothesis and use the tilde notation for all quantities which correspond  to constrained QMLE. 
Then, for  testing \eqref{test nonident} we consider the test statistic 
\begin{equation}
LM_{NT}(\gamma)= \tilde{S}^{(2)\prime}_{NT}(\gamma)\tilde{\Sigma}_{NT}^{-1}(\gamma,\gamma) \tilde{S}^{(2)}_{NT}(\gamma)\,,
\label{score test nuisance}
\end{equation}
where, according to Remark~\ref{Rem. sample information matrix}, $\tilde{\Sigma}_{NT}(\gamma,\gamma)$ is the estimator of $\Sigma(\gamma,\gamma)$, obtained by substituting $H(\gamma,\gamma)$, $B(\gamma,\gamma)$ with $\tilde{H}_{NT}(\gamma,\gamma)$, $\tilde{B}_{NT}(\gamma,\gamma)$, respectively.
Define $Z_{1,i,t}=(1, X_{i,t-1}, Y_{i,t-1})^\prime$ and $\eta_{Nt}=N^{-1/2}\sum_{i=1}^{N}Y_{i,t}(Z_{1,i,t}^\prime Z_{1,i,t}-1)+X_{1,i,t}+Y_{1,i,t}$.
An extra condition is required:
%

\begin{enumerate} [label=\textbf{B$^{\prime}$}] 
	\item  Assumption \ref{smoothness} holds with all constants not depending on $\gamma\in\Gamma$, where $\Gamma$ is compact, and 
	$\lnorm{\eta_{Nt}}_q<\infty$, for some $q > \max\left\lbrace 1+ \delta, m^* \right\rbrace $, with $ 0 < \delta < 1$.
	\label{extra equicont}
\end{enumerate}
Assumption \ref{extra equicont} is  similar to assumption  \ref{smoothness} for the particular case we consider.
An extra moment assumption is required to guarantee  stochastic  equicontinuity of the score. It can be easily shown that a sufficient condition for obtaining $\lnorm{\eta_{Nt}}_q<\infty$ would be, for example, the  weak dependence condition  $|\E(Y^r_{i,t}Y^r_{j,t}|\Fb_{t-1})|\leq \phi_{j-i}$, such that $\sum_{h=1}^{\infty} \phi^{1/r}_h<\infty$, where $r=q/2$, if $q$ is even, and $r=(q+1)/2$, if $q$ is odd. For instance, in the STNAR model \eqref{stnar}, $m^*=1$, $q=2$ and $r=1$, so the condition simplifies to a special case of Assumption~\ref{Ass weak dependence}.
From \ref{extra equicont},  Assumption~\ref{bounded away from 0}  holds trivially for \eqref{general nuisance pnar}, because a.s. $\lambda_{i,t}(\gamma)\geq\beta_0+h^\prime_i(0,\gamma)\alpha=C>0$, for $i=1,\dots,N$.
Define $\delta_2\in\R^{k_2}_+$ and $J_2=(O_{k_2\times k_1}, I_{k_2})$, where $I_s$ is a $s\times s$ identity matrix and $O_{a\times b}$ is a $a\times b$ matrix of zeros.

\begin{theorem}  \label{Thm: Test non standard}
	Assume $Y_t$ is integer-valued, following \eqref{general nuisance pnar} and suppose the conditions of Theorem~\ref{Thm. Ergodicity of nonlinear model div N}, Assumption~\ref{interior H0}-\ref{extra equicont} and \ref{Ass alpha mixing}-\ref{Ass limits existence} hold. 
	Consider the test $H_0: \alpha=0$ vs. $H_1: \alpha>0$, componentwise. Then, under $H_0$, as $ \left\lbrace  N, T_N \right\rbrace \to\infty$, $S_{N T_N}(\gamma)\Rightarrow S(\gamma)$, $LM_{N T_N}(\gamma)\Rightarrow LM(\gamma)$ and $g_{N T_N}\Rightarrow g=g(LM(\gamma))$ where
	\begin{align}
		LM(\gamma)=S^{(2)\prime}(\gamma)\Sigma^{-1}(\gamma,\gamma)S^{(2)}(\gamma)\,.
		\nonumber
	\end{align}
	Moreover, the same result holds under local alternatives $H_1: \alpha= (N T_N)^{-1/2}\delta_2$, with $S^{(2)}(\gamma)$ having mean $J_2H^{-1}(\gamma,\gamma)J_2^\prime\delta_2$.
\end{theorem}
Theorem~\ref{Thm: Test non standard}--the proof is given in Appendix \ref{Proof Test non standard}--extends  \cite[Thm.~1]{hansen_1996} in three  directions: i) develops testing for NAR  models; ii) proves  convergence to asymptotic process, where both time and network dimension diverge together; 
iii) the results holds for both  continuous-valued data (see below) and  integer-valued multivariate random variables. 
In line with Section~\ref{SEC: detailed example conti}, for each single model encompassed in \eqref{general nuisance pnar} one can substitute \ref{Ass limits existence} with network conditions \ref{Ass network} and suitable limits existence as in \ref{Ass final limits H0}.
An analogous result holds  for continuous valued time series, as  in \eqref{general nuisance nar}. Its proof is omitted.  Consider $s_t(\gamma)=Z_t^\prime(\gamma)\xi_t$ and $H_T(\gamma_1, \gamma_2)=\sum_{t=1}^{T}\sum_{i=1}^{N}Z_{i,t}(\gamma_1)Z_{i,t}^\prime(\gamma_2)$. In this case no additional weak dependence assumption is required since the error sequence is independent. 
\begin{enumerate} [label=\textbf{B$^{\prime\prime}$}] 
	\item  Assumption \ref{smoothness} holds with all constants not depending on $\gamma\in\Gamma$, where $\Gamma$ is compact.
	\label{extra equicont cont}
\end{enumerate}

\begin{theorem}  \label{Thm: Test non standard cont}
	Assume $Y_t$ is continuous-valued, following \eqref{general nuisance nar} and suppose the conditions of Theorem~\ref{Thm. Ergodicity of nonlinear continuous model div N}, Assumptions~\ref{interior H0}-\ref{extra equicont cont} and \ref{Ass limits existence cont} hold. Consider the test $H_0: \alpha=0$ vs. $H_1: \alpha\neq 0$, componentwise. Then, the results of Theorem~\ref{Thm: Test non standard} hold true.
\end{theorem}

\begin{rem}  \label{Rem: p lag} \rm
	The results of this paper extend straightforwardly to the case of model order $p>1$, i.e.  $\lambda_{t}=f(Y_{t-1},\dots,Y_{t-p}, W, \theta)$. Indeed, the proof of stability conditions of  Theorem~\ref{Thm. Ergodicity of nonlinear model}-\ref{Thm. Ergodicity of nonlinear continuous model}  are based  on the fact that the process  $\left\lbrace Y_t: t\in\Z \right\rbrace $ is  a first order Markov chain. All proofs adapt directly to a Markov chain of generic order $p$, by suitable adjustment of the contraction property \eqref{contraction}.
	Similar remark holds  for asymptotic properties of the QMLE and Theorems \ref{Thm. Consistency and Asymptotic Normality nonlinaer PNAR}-\ref{Thm: Test non standard cont}) by a suitable extension. 
	
\end{rem}

\subsection{Computations of $p$-values} \label{SEC: computation}
The null distribution of the process $g(\cdot)$ cannot be tabulated in general, apart from special cases; \cite{andrews_1993}. To overcome this obstacle we consider two different approaches.  Consider the sup-type test, $g=\sup_{\gamma\in\Gamma}(LM(\gamma))$. By Theorems~\ref{Thm: Test non standard}-\ref{Thm: Test non standard cont}, under $H_0$, $LM(\gamma)$ is a chi-square process with $k_2$ degrees of freedom. 
If the nuisance parameter $\gamma$ is scalar, \cite{davies_1987} proves that the  $p$-value of the sup-test is approximately bounded by 
\begin{equation}
	\mathrm{P}\left[ \sup_{\gamma\in\Gamma_F}(LM(\gamma))\geq M\right] \leq \mathrm{P}(\chi^2_{k_2}\geq M)+VM^{\frac{1}{2}(k_2-1)}\frac{\exp(-\frac{M}{2})2^{-\frac{k_2}{2}}}{\Gamma(\frac{k_2}{2})}\,, \label{Davies bound}
\end{equation}
where $M$ is the maximum of the test statistic $LM_{NT}(\gamma)$,  with $\gamma\in\Gamma_F$ and $\Gamma_F=(\gamma_L,\gamma_1,\dots,\gamma_l,\gamma_U)$ is a grid of values for $\Gamma$. The quantity $V$ is the approximated total variation, defined by 
\begin{equation}
	V=
\norm{LM_{NT}^{\frac{1}{2}}(\gamma_1)-LM_{NT}^{\frac{1}{2}}(\gamma_L)}+
\dots + \norm{LM_{NT}^{\frac{1}{2}}(\gamma_U)-LM_{NT}^{\frac{1}{2}}(\gamma_l)}. \nonumber
\end{equation}
Such method is attractive  because of its simplicity and  its  computational speed. This last point is of great importance in network models, especially when the dimension $N$ is large. However, the method suffers from three main drawbacks. First,  \eqref{Davies bound} leads to a conservative test, because usually the $p$-values are smaller than their bound. Second, the results of \cite{davies_1987} hold only for scalar nuisance parameters. Though this 
observation applies to several  models discussed so far, like the STNAR model \eqref{stnar},  more complex models may require inclusion of  more than one nuisance parameter. Finally,  \eqref{Davies bound} cannot be applied to the TNAR model \eqref{tnar}, because $LM(\gamma)$, under the null hypothesis, has to be differentiable \cite[p.~36]{davies_1987}, 
\cite[Sec.~4]{hansen_1996}. Following \cite{hansen_1996}, we develop a bootstrap method based on stochastic permutations. 

Define $F(\cdot)$ the distribution function of the process $g$ with $p_{NT}=1-F(g_{NT})$. From Theorem~\ref{Thm: Test non standard} and the Continuous Mapping Theorem (CMT), $p_{NT}\Rightarrow p$, where $p=1-F(g)$ and $p\sim U(0,1)$, under the null.  Hence,  the test rejects  $H_0$ if $ p_{NT}\leq a_{H_0}$, where $a_{H_0}$ is the asymptotic size of the test.  
Define $\left\lbrace \nu_t: t=1,\dots,T \right\rbrace\sim IIDN(0,1)$, such that $\tilde{S}^\nu_{NT}(\gamma)=\sum_{t=1}^{T}\tilde{s}^\nu_{Nt}(\gamma)$, with $\tilde{s}^\nu_{Nt}(\gamma)=\tilde{s}_{Nt}(\gamma)\nu_t$, is the version of the estimated score perturbed by a Gaussian noise. 
Similarly, the perturbed score test is defined by  $LM^\nu_{NT}=\tilde{S}^{\nu (2)\prime}_{NT}(\gamma)\tilde{\Sigma}_{NT}^{-1}(\gamma,\gamma)\tilde{S}^{\nu (2)}_{NT}(\gamma)$ and $\tilde{g}_{NT}=g(LM^\nu_{NT})$. Finally, $\tilde{p}_{NT}=1-\tilde{F}_{NT}(g_{NT})$ is the approximation of $p$-values obtained by stochastic permutations, where $\tilde{F}_{NT}(\cdot)$ denotes the distribution function of $\tilde{g}_{NT}$, conditional to the available sample. 
The following result shows that such a bootstrap approximation provides adequate approximation to the null distribution:
\begin{theorem} \label{Thm: Bootstrap} 
Assume the conditions of Theorems~\ref{Thm. Consistency and Asymptotic Normality nonlinaer PNAR as}, \ref{Thm: Test non standard} hold. Then, $\tilde{p}_{NT_N}-p_{NT_N}=o_p(1)$ and $\tilde{p}_{NT_N}\Rightarrow p$. Moreover, under $H_0$, $\tilde{p}_{NT_N}\xrightarrow{d}U(0,1)$.
\end{theorem}
The proof of this theorem is given in Appendix \ref{Proof Thm. Bootstrap}.
An analogous result is obtained for continuous-valued network models and it is omitted.
Although $\tilde{p}_{NT}$ is close to $p_{NT}$ asymptotically, the conditional distribution $\tilde{F}_{NT}(\cdot)$ is not observed. 
We can approximate this by Monte Carlo simulations following  (i)-(iv) of \cite[p.~419]{hansen_1996}. A Gaussian sequence $\left\lbrace \nu_{t,j} : t=1,\dots,T\right\rbrace\sim IIDN(0,1)$ is generated  and at each iteration, compute  the quantities $\tilde{S}^{\nu_j}_{NT}(\gamma)$, 
$LM^{\nu_j}_{NT}(\gamma)$ and $\tilde{g}^j_{NT}=g(LM^{\nu_j}_{NT}(\gamma))$, for $j=1,\dots,J$. 
Hence, an approximation of the $p$-values is obtained 
by $\tilde{p}^J_{NT}=J^{-1}\sum_{j=1}^{J}I(\tilde{g}^j_{NT}\geq g_{NT})$. The Glivenko-Cantelli Theorem implies that   $\tilde{p}^J_{NT}\xrightarrow{p}\tilde{p}_{NT}$, as $J\to\infty$, and choosing $J$ large enough allows to make $ \tilde{p}^J_{NT}$ arbitrary close to $\tilde{p}_{NT}$. 

%

The proposed  bootstrap methodology provides a direct approximation of the $p$-values instead of an approximate bound, given by  \eqref{Davies bound}. Furthermore,  it is suitable even when testing linearity in the 
presence of more than one nuisance parameter. 
As a final remark, the stochastic permutation bootstrap method has been preferred instead of parametric bootstrap as it requires only the generation of standard univariate normal sequences at each step. This reduces  considerably the computational burden of generating a $N\times 1$ vector of observations at each step of the procedure. This is especially relevant in the case of count data, since the simulation of copula   
can  be time consuming. 

\begin{rem} 		\rm
Following up Remark \ref{Remark:TNAR}, note that the previous results do not apply to TNAR model \eqref{tnar}, if $N \rightarrow \infty$. 
The stochastic equicontinuity and uniform convergence assumptions  require $h_{t}(\gamma)$ to be continuous with respect to $\gamma$ which is  not satisfied for \eqref{tnar}.  For instance, when trying to establish stochastic equicontinuity of the score, it can be proved that for \eqref{tnar}, the Lipschitz property \eqref{lipschitz score} can be obtained in  expectation but with magnitude $\lambda=1/(2q)$. However, to establish the result of Theorem~\ref{Thm: Test non standard} we need $\lambda>m^*/q$ \cite[p.~357]{hansen1996stochastic}.  This can happen only when $m^*<1/2$ but for the  TNAR model $m^*=1$, so the condition is not satisfied. Supp. Mat.~\ref{threshold models}  provides properties, estimation and testing  for  TNAR models when  $N$ is fixed.
\end{rem}

\section{Simulations}	\label{SEC: simulations}

We provide two different cases for the network generating mechanism to verify empirically the above results. Additional results are reported in Supp. Mat. \ref{SuppSec:Additionalsimulations}.

\begin{exnet} \label{sbm}(Stochastic Block Model (SBM)). First consider  the stochastic block model, see \cite{wang1987} and \cite{nowicki_2001}, among others. A block label $(k = 1,\dots, K)$ is assigned for each node with equal probability and $K$ is the total number of blocks. Then, set $\mathrm{P}(a_{ij}=1) = N^{-0.3}$ if $i$ and $j$ belong to the same block, and $\mathrm{P}(a_{ij}=1)= N^{-1}$ otherwise. Practically, the model assumes that nodes within the same block are more likely to be connected with respect to nodes from different blocks. We assume $K\in \left\lbrace 2,5 \right\rbrace $.
\end{exnet}

\begin{exnet}  \label{erdos-renyi} (Erd\H{o}s-R\'{e}nyi  (ER) Model). Introduced by \cite{erdos_1959} and \cite{gilbert_1959}, this graph model is simple. The network is constructed by connecting $N$ nodes randomly. Each edge is included in the graph with probability $p$, independently from every other edge. In this example we set  $p=\mathrm{P}(a_{ij}=1)=N^{-0.3}$. 
\end{exnet} 

Consider testing  $H_0: \gamma=0$ versus $H_1: \gamma>0$, for models \eqref{nonlinear} and \eqref{nonlinear_cont}. Under $H_0$, the model reduces to \eqref{nar_1}. For the continuous-valued case, we test linearity of the NAR against the nonlinear version in \eqref{nonlinear_cont}. The random errors $\xi_{i,t}$ are simulated from $N(0,1)$. For the data generating process of the vector $Y_t$, the initial value $Y_0$ is randomly simulated according to its stationary distribution \cite[Prop.~1]{zhu2017}, which is Gaussian  with mean $\mu=\beta_0(1-\beta_1-\beta_2)^{-1}1$ and covariance matrix $\mathrm{vec}[\mathrm{Var}(Y_t)]=(I_{N^2}-G \otimes G)^{-1}\mathrm{vec}(I)$, where $\mathrm{vec}(.)$ denotes the vec operator  and $\otimes$ denotes the Kronecker product. 
We set $\theta^{(1)}=(\beta_0,\beta_1,\beta_2)^\prime=(1.5,0.4,0.5)^\prime$. This procedure is replicated $S=1000$ times. Then,  $\tilde{\theta}^{(1)}$ is computed for each replication. 
By Proposition~\ref{nar chi}, the quasi-score statistic \eqref{score test} is evaluated and compared with the critical values of a $\chi^2_1$ distribution. 
Results of this simulation study 
are reported in Table~\ref{sim_gauss}. The empirical size of the test does not exceed the nominal significance level in all cases considered . When $N$ is small and $T$ is large enough, the power of the test statistics tends to 1. In the case of small temporal size $T$ and large  network dimension $N$, the approximation suffers. This is  expected and is explained  by  i) the double asymptotic results of Section~\ref{SEC: linearity test} hold when  $T_N\to\infty$ as $N\to\infty$--see  the proof of Lemma  \ref{Lem. Consistency and Asymptotic Normality nonlinaer PNAR}; ii) the temporal dependence induced by  the error term requires a sufficiently large $T$ for successful model identification;  iii) the quasi-likelihood might not approximate the true likelihood, see also  \cite[Sec.~4.1]{armillotta_fokianos_2021}. When both $N,T$ are large enough, the test approximates adequately its asymptotic distribution. 
As expected, when $\gamma=1$, the test statistic's power improves, because $\gamma$ 
is far from 0.  Improved performance of the test statistic is observed when either  $K=5$ or  when the Erd\H{o}s-R\'{e}nyi model is employed. 
Histograms and  Q-Q plots of the simulated score test against the $\chi^2_1$ distribution are plotted in Supp. Mat. Fig.~\ref{qq}. For all the network models, the histogram is positively skewed and approximates satisfactory  the $\chi^2_1$ distribution. The Q-Q plots lie into the confidence bands quite satisfactory and the  empirical mean and variance of the simulated score tests are close to 1 and 2, respectively. Further simulations results for the integer-valued case can be found in Supp. Mat.~\ref{SuppSec:Additionalsimulations} together with simulation results related to the non-identifiable case.

\begin{table}[h]
\centering
\caption{Empirical size and power   of the test statistics \eqref{score test} for testing $H_0:\gamma=0$ versus $H_1:\gamma>0$, in model \eqref{nonlinear_cont}, with $S=1000$ simulations, for various values of $N$ and $T$. 
	Data are continuous-valued and generated from the linear model \eqref{nar_1}.}

\hspace*{-1cm}
	\begin{tabular}{c|c|c|c|ccc|ccc|ccc}\hline\hline
		Model &\multicolumn{3}{c|}{} &\multicolumn{3}{c|}{Size}  & \multicolumn{3}{c|}{Power $(\gamma=0.5)$} & \multicolumn{3}{c}{Power $(\gamma=1)$} \\\hline
		& $K$ & $N$ & $T$ & 10\% & 5\% & 1\% & 10\% & 5\% & 1\% & 10\% & 5\% & 1\% \\\hline
		\multirow{5}{*}{SBM} & \multirow{5}{*}{2} & 4 & 500 & 0.093 & 0.043 & 0.009 & 1.000 &  1.000 & 0.999 & 1.000 & 1.000 & 1.000 \\
		&& 500 & 10  & 0.019 & 0.004 & 0.000 & 0.158 & 0.063 & 0.002 & 0.164 & 0.067 & 0.001 \\
		&& 200 & 300 & 0.110 & 0.044 & 0.006 & 0.495 &  0.337 & 0.125 & 0.994 & 0.990 & 0.933 \\
		&& 500 & 300 & 0.101 & 0.048 & 0.009 & 0.716 &  0.583 & 0.288 & 1.000 & 1.000 & 0.995 \\
		&& 500 & 400 & 0.105 & 0.050 & 0.006 & 0.751 & 0.619 & 0.311 & 1.000 & 1.000 & 0.999  \\\hline
		\multirow{5}{*}{SBM} &\multirow{5}{*}{5} & 10 & 500 & 0.119 & 0.062  & 0.015 & 1.000 & 1.000 & 1.000 & 1.000 & 1.000 & 1.000 \\
		&& 200 & 300 & 0.091 & 0.051 & 0.006 & 0.667 &  0.542 & 0.268 & 1.000 & 1.000 & 1.000 \\
		&& 500 & 300 & 0.098 & 0.047 & 0.006 & 0.847 &  0.748 & 0.448 & 1.000 & 1.000 & 1.000  \\
		&& 500 & 400 & 0.086 & 0.039 & 0.006 & 0.885 & 0.807 & 0.541 & 1.000 & 1.000 & 1.000  \\
		\hline
		\hline
		\multirow{5}{*}{ER}& \multirow{5}{*}{-} & 30 & 500 & 0.066 & 0.029 & 0.004 & 0.272 & 0.156 & 0.048 & 0.888 & 0.802 & 0.565 \\
		&& 500 & 30  & 0.026 & 0.005 & 0.000 & 0.392 & 0.235 & 0.044 & 0.935 & 0.847 & 0.523 \\
		&& 200 & 300 & 0.085 & 0.031 & 0.004 & 0.411 &  0.272 & 0.080 & 0.974 & 0.949 & 0.798 \\
		&& 500 & 300 & 0.082 & 0.042 & 0.004 & 0.649 &  0.476 & 0.192 & 0.999 & 0.998 & 0.974 \\
		&& 500 & 400 & 0.089 & 0.051 & 0.008 & 0.666 & 0.519 & 0.206 & 1.000 & 1.000 & 0.992  \\
		\hline
		\hline
	\end{tabular}
\label{sim_gauss}
\end{table}

\section{Empirical example} \label{SEC: applications}

We discuss  an example  of the testing methods for integer data. 
For an example  concerning continuous data see Supp. Mat.~\ref{wind}.
The dataset consists of monthly number of burglaries on the south side of Chicago from 2010-2015, i.e. $T=72$ and  $N=552$ census block groups of Chicago; see \cite{clark_2021}, \href{https://github.com/nick3703/Chicago-Data}{\url{https://github.com/nick3703/Chicago-Data}}. 
To predict future number of burglaries, the ordinary 
Vector Autoregressive (VAR) model can be applied  but we should take into account that data are counts and dimensionality issues  because  the  number of VAR  parameters is large compared to the sample size. A simple  method, like fitting AR(1) models separately to each individual census blocks, is applicable but still requires  $2N$ parameters  to be fitted. More crucially,  the relationship across different time series is not taken into account. To overcome such issues, we  appeal to  geographic network information between blocks to fit a  PNAR model which takes into account dependence among count valued data. An undirected network structure is defined  by geographical proximity: two blocks are connected if they share at least a border. The density of this network is 1.74\%. The median number of connections is 5. The QMLE is employed for fitting  linear PNAR model \eqref{nar_1}. The results are summarized in Table \ref{chicago_results}. 
The magnitude of the network effect $\beta_1$ shows that  an increasing number of burglaries in a block can lead to a growth in the same type of crime committed in a neighborhood  area. The effect of the lagged variable has a upwards impact on the number of burglaries, as well.
We evaluate the out-of-sample forecasting performance of the linear PNAR(1) model versus a baseline AR(1) model fitted separately to each individual census block. We evaluate the one step ahead forecast by computing its Root Mean Square Error (RMSE). The RMSE for the PNAR model is 0.038. This  is considerably smaller than the RMSE obtained by  the AR(1) models (which is 0.167). In conclusion,  the PNAR model gives  significant accuracy improvement of the one-step prediction  and at the same time  achieves   parsimony. We apply now the proposed linearity tests.
A quasi-score linearity test is computed according to \eqref{score test}, by using the asymptotic chi-square test, for the nonlinear model \eqref{nonlinear}, testing $H_0: \gamma=0$ vs. $ H_1: \gamma>0$. We also test linearity against the presence of smooth transition effects, as in \eqref{stnar}, with $H_0: \alpha=0$ vs. $ H_1: \alpha>0$.  
A grid of 100 equidistant values in an interval of values $\Gamma_F=[\gamma_L,\gamma_U]$  
is selected for the nuisance parameter $\gamma$, where the extremes are defined as in \cite[p.~9]{armillotta_et_al_rpackage_2022}. 
According to the results of Theorem~\ref{Thm: Test non standard}, the $p$-values are computed with the Davies bound approximation \eqref{Davies bound} for the test $\textrm{sup}LM_T=\sup_{\gamma\in \Gamma_F} LM_T(\gamma)$ as well as through  bootstrap approximation procedure.
The number of bootstrap replications is  $J=299$.
Finally, a linearity test against threshold effects, as in \eqref{tnar}, is also performed, which leads to the test $H_0:\alpha_0=\alpha_1=\alpha_2=0$ vs. $ H_1: \alpha_l>0$, for some $l=0,1,2$. A feasible range values for the non identifiable threshold parameter has been considered as in \cite[p.~11]{armillotta_et_al_rpackage_2022}. 
From Table \ref{chicago_results}, the linearity test against \eqref{nonlinear} is rejected at standard levels.
This gives an intuition for possible nonlinear drifts in the intercept. Davies bound gives evidence in favour of STNAR effects  at 5\% level. Conversely, bootstrap sup tests reject nonlinearity coming from both smooth \eqref{stnar} and abrupt transitions \eqref{tnar} models. We conclude that 
there is no clear evidence of regime switching effect. 


\begin{table}[h]
\centering
\caption{QMLE estimates of the linear model \eqref{nar_1} for Chicago burglaries counts. Standard errors in brackets.
	Linearity is tested against the nonlinear model \eqref{nonlinear}, with $\chi^2_1$ asymptotic test \eqref{score test}; against the STNAR model \eqref{stnar}, with $p$-values computed by (DV) Davies bound \eqref{Davies bound}, bootstrap $p$-values of sup-type test; 
	and versus TNAR model \eqref{tnar}. }
	\begin{tabular}{cccc}\hline \hline
		
		Models & $\beta_0$ & $\beta_1$ & $\beta_2$   \\\hline
		\eqref{nar_1} & 0.455  & 0.322 & 0.284  \\\hline
		SE & (0.022) & (0.013) & (0.008) \\
		\hline
		\hline
		Models & Chi-sq. & DV & Bootstrap \\
		\eqref{nonlinear} & 8.999 & - & -  \\
		\eqref{stnar} & - & 0.038 & 0.515 \\
		\eqref{tnar} & - & - & 0.498 \\
		\hline
		\hline
	\end{tabular}
\label{chicago_results}
\end{table}

\appendix
\section{Appendix} \label{SEC appendix}
\renewcommand{\thefigure}{A-\arabic{figure}}
\renewcommand{\thetable}{A-\arabic{table}}
\setcounter{figure}{0}
\setcounter{table}{0}
\renewcommand{\theproposition}{\Alph{section}.\arabic{proposition}}
\renewcommand{\thelemma}{\Alph{section}.\arabic{lemma}}

\subsection{Proof of Theorem \ref{Thm. Ergodicity of nonlinear model}}
\label{Proof Ergodicity of nonlinear model}
Consider the $N\times 1$ Markov chain
$Y_t=F(Y_{t-1},N_t)$ where $\left\lbrace N_t, t\in\Z\right\rbrace $ defined in \eqref{nonlinear_pnar} is a sequence of $IID$ $N$-dimensional count processes such that $N_{i,t}$, for $i=1\dots,N$, are Poisson processes with intensity 1. $F(\cdot)$ is a measurable function such that $F(y,N_t)=(N_t[f(y)])$ 
and $f(\cdot)$ is defined in \eqref{nonlinear_pnar} for $y\in \N^N$. 
By \eqref{contraction}, $f(y)\preceq C+Gy$, where $C=f(0)$, we have 
\begin{equation}
	\E\norm{F(y,N_1)}_1=1^\prime f(y) \leq 1^\prime\left[ C+Gy\right] <\infty
	\nonumber
\end{equation}
since the expectation of the Poisson process is $\E N_1(\lambda)=\lambda$. 
Moreover, for  $y,y^*\in \N^N$ 
\begin{equation}
	\E\norm{F(y,N_1)-F(y^*,N_1)}_{vec}\preceq G\norm{y-y^*}_{vec}
	\nonumber
\end{equation}
as $\E\norm{N_1(\lambda_1)-N_1(\lambda_2)}_{vec}=\norm{\lambda_1-\lambda_2}_{vec}$. Note that 
$\rho(G)<1$, 
Therefore, by \cite[Thm.~1]{tru2021}, $\{Y_t,~ t \in \mathbb{Z} \}$ is a stationary  and ergodic process with $\mbox{E}\norm{Y_t}_1  < \infty$.
Now set $\delta>0$ such that $\rho(G_{\delta})<1$, if $G_{\delta}=(1+\delta)G$. 
From \cite[Lemma~2]{tru2021} we have that
\begin{equation}
	\lnorm{N_t[f(y)]}_{a,vec}\preceq(1+\delta)\norm{f(y)}_{vec}+b1\preceq C_{\delta b}+G_\delta\norm{y}_{vec}\,  \nonumber
\end{equation}
by recalling that $\norm{f(y)}_{vec}\preceq C+G\norm{y}_{vec}$ and $\rho(G_{\delta})<1$, where $b>0$ and $C_{\delta b}=(1+\delta)C+b$. 
Then, by \cite[Thm.~1]{tru2021} 
we get $\mbox{E}\norm{Y_t}_a^a  < \infty$, $\forall a\geq1$. Theorem~\ref{Thm. Ergodicity of nonlinear continuous model} follows analogously.
\qed 

\subsection{Proof of Theorem~\ref{Thm. Ergodicity of nonlinear model div N}}
\label{Proof of Thm. Ergodicity of nonlinear model div N}
For any arbitrary $N$, $\E(Y_t)=\E(\lambda_t)=\E[f(Y_{t-1}) ]\preceq c1 + G\E(Y_{t-1})$, by \eqref{contraction}, where $\max_{i\geq 1}f_i(0,0)=c>0$. 
Define $\mu=c(1-\mu_1-\mu_2)^{-1}$. Note that $ \rho(G) \leq \vertiii{G}_{\infty} \leq \mu_{1} \vertiii{W}_\infty+\mu_{2}
\leq \mu_{1}+\mu_{2}$. This is so because $\vertiii{W}_\infty=1$, by construction.
Since $\mu_1+\mu_2 <1$ we have $\vertiii{G}_\infty <1$ 
and by \cite[19.16(a)]{seber2008}  $(I-G)^{-1}$ exists. Moreover, $(I-G)^{-1}1=(1-\mu_1-\mu_2)^{-1}1$, implying that $\E(Y_t)\preceq \mu1$ and $\max_{i\geq 1}E(Y_{i,t})\leq \mu$. It holds that $\xi_t=Y_t-\lambda_t$, $\E\norm{\xi_{i,t}}\leq 2 \E(Y_{i,t}) \leq 2\mu <\infty$. 
Furthermore, by using backward substitution and  \eqref{contraction}, we have $Y_t\preceq\mu1+\sum_{j=0}^{\infty}G^j\xi_{t-j}=\sum_{j=0}^{\infty}G^j(c1+\xi_{t-j})$.

From the definition in Supp. Mat.~\ref{Supp:Stationarity} we have
that $\mathcal{W} = \left\lbrace \omega\in\R^\infty:\omega_\infty=\sum\norm{\omega_i}<\infty\right\rbrace $, where $\omega = (\omega_i\in\R: 1\leq i \leq \infty)^\prime\in\R^\infty$. For each $\omega\in\mathcal{W}$, let $\omega_N = (\omega_1,\dots,\omega_N)^\prime \in\R^N$ be its truncated $N$-dimensional version. For any $\omega\in\mathcal{W}$, $\E\norm{c1+ \xi_t}_{vec}\preceq (c+ 2\mu)1 = C1<\infty$, 
$G^j1=(\mu_1+\mu_2)^j1$ and $\E\norm{\omega_N^\prime Y_t}\leq\E(\norm{\omega_N}_{vec}^\prime \sum_{j=0}^{\infty}G^j(c1+\xi_{t-j}))\leq C \omega_\infty\sum_{j=0}^{\infty}(\mu_1+\mu_2)^j=C_*$, since $\mu_1+\mu_2<1$, which implies that $Y_t^\omega=\lim_{N\to\infty}\omega_N^\prime Y_t<\infty$ with probability 1. Moreover, $Y_t^\omega$ is strictly stationary and therefore $\left\lbrace Y_t \right\rbrace $ is strictly stationary, following Supp. Mat.~\ref{Supp:Stationarity}. To verify the uniqueness of the solution 
take another stationary solution $Y^*_t$ 
to the PNAR model. 
Then $\E\norm{\omega_N^\prime Y_t-\omega_N^\prime Y^*_t}\leq \norm{\omega_N}_{vec}^\prime \E\norm{N_t(\lambda_t)-N_t(\lambda^*_t)}_{vec}\leq\norm{\omega_N}_{vec}^\prime G\E\norm{Y_{t-1}-Y^*_{t-1}}_{vec}= 0$, by infinite backward substitution, for any $N$ and weight $\omega$. So $Y_t^\omega=Y_t^{*,\omega}$ with probability one. In addition, $\mu_1w_i^\prime+\mu_2e_i^\prime=e_i^\prime G$ and condition \eqref{contraction} is equivalent to require, for $i=1,\dots,N$,  a.s.
\begin{align}
	\norm{\lambda_{i,t}-\lambda^*_{i,t}}&=e_i^\prime\norm{f(Y_{t-1})-f(Y^*_{t-1})}_{vec}\leq \left( \mu_1w_i^\prime+\mu_2e_i^\prime\right) \norm{Y_{t-1}-Y_{t-1}^*}_{vec} \nonumber \\
	&=\mu_1\sum_{j=1}^{N}w_{ij}\norm{Y_{j,t-1}-Y_{j,t-1}^*}+\mu_2\norm{Y_{i,t-1}-Y_{i,t-1}^*} \nonumber
\end{align}
which leads a.s. to $\lambda_{i,t}=f_i(X_{i,t-1},Y_{i,t-1})\leq c+ \mu_1X_{i,t-1} + \mu_2 Y_{i,t-1}$. 
Then, when $N$ is increasing, \cite[Prop.~2]{armillotta_fokianos_2021} applies directly by a recursion argument \cite[Sec.~S-1.1]{armillotta_fokianos_2021} and all moments of the process $\left\lbrace Y_t \right\rbrace $ are uniformly bounded.
\qed

\subsection{Proof of Theorem~\ref{Thm. Ergodicity of nonlinear continuous model div N}}
\label{Proof of Thm. Ergodicity of nonlinear model div N cont}
Similar to \ref{Proof Ergodicity of nonlinear model}, by assuming that $\max_{i\geq 1}\E\norm{\xi_{i,t}}^a\leq C_{\xi,a}<\infty$, the first $a^{\text{th}}$ moments of $Y_t$ are uniformly bounded.  By \eqref{contraction}, $\norm{Y_t}_{vec}\preceq\sum_{j=0}^{\infty}G^j(c1+\norm{\xi_{t-j}}_{vec})$, where $c=\max_{i\geq 1}\norm{f_i(0,0)}$. Analogously to \ref{Proof of Thm. Ergodicity of nonlinear model div N}, since $G^j1=(\norm{\mu_1}+\norm{\mu_2})^j1$ and $\norm{\mu_1}+\norm{\mu_2}<1$, $\left\lbrace Y_t \right\rbrace $ defined as in \eqref{nonlinear_nar}, is strictly stationary, following \cite[Def.~1]{zhu2017}. The uniqueness of the solution follows by $\norm{Y_t-Y^*_t}_{vec}=\norm{\lambda_t-\lambda^*_t}_{vec}$ and the infinite backward substitution argument.
\qed

\subsection{Proof of Theorem~\ref{Thm: Test non standard}}
\label{Proof Test non standard}
First we show the weak convergence of $(NT)^{-1/2}S_{NT}(\gamma)$ to a Gaussian process with kernel $B(\gamma_1, \gamma_2)$.  For all non-null $\eta\in\R^k$, consider the triangular array 
$s^*_{Nt}(\gamma)=\eta^\prime(N^{-1/2}\sum_{i=1}^{N} s_{i,t}(\gamma))$. By Assumption~\ref{extra equicont}, $\Gamma$ is compact, and by the continuity of the score $s^*_{Nt}(\gamma)$ is compact.
Note that $s^*_{Nt}(\gamma) $ is a martingale difference array. So by the results of Lemma~\ref{Lem. Consistency and Asymptotic Normality nonlinaer PNAR}, the multivariate pointwise central limit theorem and $(NT_N)^{-1}B_{NT_N}(\gamma_1,\gamma_2)\xrightarrow{p}B(\gamma_1,\gamma_2)$ establish the finite dimensional convergence. It remains to show the stochastic equicontinuity, i.e. (\cite[Thm.~2]{hansen1996stochastic}) 
\begin{equation}
	\norm{s^*_{Nt}(\gamma)-s^*_{Nt}(\gamma^*)}\leq \delta_{Nt}\norm{\gamma-\gamma^*}^\lambda_1\,,
	\label{lipschitz score}
\end{equation}
a.s. with $\lnorm{\delta_{Nt}}_q<\infty$ and $\lnorm{s^*_{Nt}(\gamma)}_q<\infty$, where $q \geq 2$ and $\lambda$ such that $ q > m^*/\lambda$. By \cite[Sec.~S-6]{armillotta_fokianos_2021} and Assumption \ref{Ass alpha mixing}-\ref{Ass limits existence}, $\lnorm{s^*_{Nt}(\gamma)}_4<\infty$. For $q > 4$ a similar result can be obtained following the arguments of  \cite[Rem.~2.3]{yaskov_2015} by requiring higher order covariances in Assumption~\ref{Ass weak dependence}.
%
To prove \eqref{lipschitz score} we recall the following uniform bounds.  
By Assumption~\ref{extra equicont}, 	for $i=1,\dots,N$, 
$\norm{\partial f_i(x_i,y_i,\theta)/\partial\alpha_b}=h^{b}_{i,t}(\gamma)\leq c_b + c_{1b} x_i + c_{2b} y_i $ for $b=1,\dots,k_2$ and $l=1,\dots,m^*$, where $c_b=h^{b}_{i}(0, \gamma)$ $\forall\gamma\in\Gamma$. 
Let $C, C_0, C_1, C_2 >0$ be generic constants varying from place to place, which do not depend on $\gamma$.
Then, a.s. $\norm{h_{i,t}(\gamma)}_1\leq C_0+ C_1X_{i,t-1}+C_2Y_{i,t-1}$. 
Similar bounds hold for 
$\lambda_{i,t}(\gamma)$ and $\norm{Z_{i,t}(\gamma)}_1$. By Theorem~\ref{Thm. Ergodicity of nonlinear model div N} all the moments of the Poisson process $Y_t$ exist as well as those associated to the error $\xi_t(\gamma)=Y_t-\lambda_t(\gamma)$. This fact and the multinomial theorem imply that every moment of all the previously defined random variables is uniformly bounded. 
Define $h_{i,t}(0, \gamma)=c$, $\forall\gamma\in\Gamma$, and $h^*_{i,t}(\gamma)=h_{i,t}(\gamma)-h_{i,t}(0, \gamma)$. For $i=1,\dots,N$, and MVT $\norm{h_{i,t}(\gamma)-h_{i,t}(\gamma^*)}_1=|h^*_{i,t}(\gamma)- h^*_{i,t}(\gamma^*)|_1=|\partial h^*_{i,t}(\tilde{\gamma})/\partial\gamma|_1\norm{\gamma-\gamma^*}_1\leq A_{i,t-1}\norm{\gamma-\gamma^*}_1 $ a.s. where $A_{i,t-1}=C_1 X_{i,t-1}+C_2 Y_{i,t-1} $, 
$\tilde{\gamma}_l$ are intermediate points between $\gamma_l$ and $\gamma^*_l$, for $l=1,\dots,m^*$, and the last inequality holds by Assumption \ref{extra equicont} since $\partial h^*_{i,t}(\tilde{\gamma})/\partial\gamma=\partial^2 f_i(x_i,y_i,\theta)/\partial\alpha\partial\gamma-\partial^2 f_i(0,0,\theta)/\partial\alpha\partial\gamma$.
For all $\gamma,\gamma^*\in\Gamma$, standard algebra and previous bounds show that a.s. $\norm{s^*_{Nt}(\gamma)-s^*_{Nt}(\gamma^*)}\leq \delta_{Nt} \norm{\gamma-\gamma^*}_1$ 
where $\delta_{Nt}= C/\sqrt{N}\sum_{i=1}^{N}A_{i,t-1}\left(1+C_0Y_{i,t}+C_1Y_{i,t} X_{i,t-1}+C_2Y_{i,t} Y_{i,t-1} \right) $ proving \eqref{lipschitz score} with $\lambda=1$. By Assumption~\ref{extra equicont} $\lnorm{\delta_{Nt}}_q<\infty$ since $\delta_{Nt}\leq C \eta_{Nt}$, then $(NT_N)^{-1/2}S_{NT_N}(\gamma)=T_N^{-1/2}\sum_{t=1}^{T_N}N^{-1/2}s_{Nt}(\gamma)$ is stochastically equicontinuous and, as $\left\lbrace N, T_N \right\rbrace \to \infty$, $(NT_N)^{-1/2}S_{NT_N}(\gamma) \Rightarrow S(\gamma)$.

We now prove uniform convergence of $\tilde{\Sigma}_{NT}(\gamma_1,\gamma_2)$ by showing stochastic equicontinuity for Hessian and information matrices.  For all  $\eta\in\R^k$, $\eta \neq 0$, consider the triangular array
$b_{Nt}(\gamma_1,\gamma_2)=\eta^\prime(N^{-1} B_{Nt}(\gamma_1,\gamma_2))\eta$ where $B_{Nt}(\gamma_1,\gamma_2)$ is the single summand of $B_{NT}(\gamma_1,\gamma_2)$. Define  $\rho_{ijt}(\gamma_1, \gamma_2)=\E[\xi_{i,t}(\gamma_1) \xi_{j,t}(\gamma_2) |  \Fb_{t-1}]/\sqrt{\lambda_{i,t}(\gamma_1)\lambda_{j,t}(\gamma_2)}$ the conditional correlation. Then, a.s. 
\begin{align}
	&\norm{b_{Nt}(\gamma_1,\gamma_2)-b_{Nt}(\gamma^*_1, \gamma^*_2)} \nonumber \\
	&\leq\eta^\prime\left( \frac{1}{N}\sum_{i,j=1}^{N}\frac{Z_{i,t}(\gamma_1)\rho_{ijt}(\gamma_1, \gamma_2)Z^\prime_{j,t}(\gamma_2)}{\sqrt{\lambda_{i,t}(\gamma_1)}\sqrt{\lambda_{j,t}(\gamma_2)}}-\frac{Z_{i,t}(\gamma^*_1)\rho_{ijt}(\gamma^*_1, \gamma^*_2)Z^\prime_{j,t}(\gamma^*_2)}{\sqrt{\lambda_{i,t}(\gamma^*_1)}\sqrt{\lambda_{j,t}(\gamma^*_2)}} \right) \eta \nonumber \\
	&\leq C \sum_{r=1}^{5}D_{r}\,, \nonumber
\end{align}
and the inequality follows since for a matrix $M$, $\eta^\prime M \eta \leq \norm{\eta\eta^\prime}_1\norm{M}_1$, and by $\lambda_{i,t}(\gamma) \geq C$ $\forall \gamma \in \Gamma$. The elements $D_r$ are obtained by 
consecutive addition and subtraction. 
We focus on one  element (say $D_1$), the other terms are treated analogously.  Some tedious algebra shows that, a.s. 
\begin{align}
	&\norm{\rho_{ijt}(\gamma_1, \gamma_2) - \rho_{ijt}(\gamma^*_1, \gamma^*_2)} \nonumber \\
	&\leq \norm{\lambda^{\frac{1}{2}}_{i,t}(\gamma_1)-\lambda^{\frac{1}{2}}_{i,t}(\gamma^*_1)} \norm{\rho_{ijt}(\gamma^*_1, \gamma^*_2)}\lambda^{\frac{1}{2}}_{j,t}(\gamma_2)+\norm{\lambda^{\frac{1}{2}}_{j,t}(\gamma_2)-\lambda^{\frac{1}{2}}_{j,t}(\gamma^*_2)}\norm{\rho_{ijt}(\gamma^*_1, \gamma^*_2)}\lambda^{\frac{1}{2}}_{i,t}(\gamma^*_1) \nonumber \\
	& \leq C_1 \norm{\lambda_{i,t}(\gamma_1)-\lambda_{i,t}(\gamma^*_1)} \varphi_{j-i} A^*_{j,t-1} + C_2 \norm{\lambda_{j,t}(\gamma_2)-\lambda_{j,t}(\gamma^*_2)}  \varphi_{j-i} A^*_{i,t-1}  \nonumber \\
	&\leq C^*_1 \varphi_{j-i} \tilde{A}_{ij,t-1} \norm{\gamma_1-\gamma^*_1}_1 +  C^*_2 \varphi_{j-i} \tilde{A}_{ji,t-1}\norm{\gamma_2-\gamma^*_2}_1 \nonumber 
\end{align}
where $A^*_{i,t-1}=A_{i,t-1} +C_0$ and $\tilde{A}_{ij,t-1}=A_{i,t-1}A^*_{j,t-1}$. The first inequality follows by addition and subtraction. The second inequality is a consequence of Assumption~\ref{Ass weak dependence} and $\norm{\sqrt{x}-\sqrt{y}}=\norm{x-y}/(\sqrt{x}+\sqrt{y})$; the third is due to Lipschitz continuity of $h_{i,t}(\gamma)$. Let us define $\pi_{ijt}(\gamma,\gamma^*) =   \norm{\sqrt{\lambda_{i,t}(\gamma^*_1)\lambda_{j,t}(\gamma^*_2)} Z_{i,t}(\gamma_1) Z^\prime_{j,t}(\gamma_2)}_1 \leq \pi_{ijt}$ a.s. with the inequality coming
from previous uniform bounds where $\pi_{ijt}$ is a linear combination of $X_{i,t-1}$ and $Y_{i,t-1}$ not depending on $\gamma$. Then,
\begin{align}
	D_1 &= \frac{1}{N}\sum_{i,j=1}^{N} \norm{\rho_{ijt}(\gamma_1, \gamma_2) - \rho_{ijt}(\gamma^*_1, \gamma^*_2)}\pi_{ijt}(\gamma,\gamma^*) \nonumber \\
	&\leq \frac{ C^*_1}{N}\sum_{i,j=1}^{N}  \varphi_{j-i} \tilde{A}_{ij,t-1}\pi_{ijt}  \norm{\gamma_1-\gamma^*_1}_1 +  \frac{ C^*_2}{N}\sum_{i,j=1}^{N}\varphi_{j-i} \tilde{A}_{ji,t-1}\pi_{ijt} \norm{\gamma_2-\gamma^*_2}_1 \,. \nonumber 
\end{align}		
This shows that $\norm{b_{Nt}(\gamma_1,\gamma_2)-b_{Nt}(\gamma^*_1, \gamma^*_2)}\leq b^*_{1,Nt} \norm{\gamma_1-\gamma^*_1}_1 + b^*_{2,Nt} \norm{\gamma_2-\gamma^*_2}_1$ a.s. with $b^*_{s,Nt}$  defined by obvious notation, not depending on $\gamma$ and such that $\E(b^*_{s,Nt}) < \infty$, for $s=\left\lbrace 1, 2\right\rbrace $. Rewriting in matrix form we have $b^*_{s,Nt}=\eta^\prime(N^{-1} B^*_{s,Nt}(\gamma_1,\gamma_2))\eta$.
According to \cite[Lem.~1]{andrews_1992generic} this is a sufficient condition for the information matrix to be stochastic equicontinuous and by \cite[Thm.~1]{andrews_1992generic} $(NT_N)^{-1}B_{NT_N}(\gamma_1,\gamma_2)\xrightarrow{p}B(\gamma_1,\gamma_2)$ uniformly over $\gamma_1, \gamma_2 \in \Gamma$, as $\left\lbrace N, T_N \right\rbrace \to \infty$.
An analogous result for the Hessian follows by MVT with respect to $\gamma_1, \gamma_2$ and the uniform boundedness of the third derivative (Lemma~\ref{Lem. Consistency and Asymptotic Normality nonlinaer PNAR}). 
By  standard
Taylor expansion arguments, and CMT,   $(NT_N)^{-1}\tilde{\Sigma}_{NT_N}(\gamma_1,\gamma_2)\xrightarrow{p}\Sigma(\gamma_1,\gamma_2)$ uniformly over $\gamma_1, \gamma_2 \in \Gamma$.  
Following analogous steps of Supp. Mat.~\ref{Proof general nonlinear score} for the identifiable parameters $\phi$, equation \eqref{normality} leads to
\begin{equation}
	\frac{\tilde{S}^{(2)}_{NT_N}(\gamma)}{\sqrt{NT_N}}\doteq P(\gamma,\gamma)\frac{S_{NT_N}(\gamma)}{\sqrt{NT_N}}\Rightarrow P(\gamma,\gamma)S(\gamma)\coloneqq S^{(2)}(\gamma)\equiv N(0, \Sigma(\gamma,\gamma))\,, \label{S tilde}
\end{equation}
with $P(\gamma,\gamma)=[-J_2H(\gamma,\gamma)J_1^\prime(J_1H(\gamma,\gamma)J_1^\prime)^{-1}, I_{k_2}]$, $\Sigma(\gamma,\gamma)=P(\gamma,\gamma)B(\gamma,\gamma)P(\gamma,\gamma)^\prime$. Finally, the  CMT shows that   $LM_{NT_N}(\gamma)\Rightarrow LM(\gamma)$ and $g_{NT_N} \Rightarrow g$. A similar conclusion is obtained for the local alternatives $\alpha= (NT)^{-1/2}\delta_2$, where $\delta_2\in\R^{k_2}$, by \eqref{normality local}, with $S^{(2)}(\gamma)\equiv N(J_2H^{-1}(\gamma,\gamma)J_2^\prime\delta_2, \Sigma(\gamma,\gamma))$ in \eqref{S tilde}. This ends the proof.
\qed

\subsection{Proof of Theorem~\ref{Thm: Bootstrap}}  \label{Proof Thm. Bootstrap} 
Following the results of Section~\ref{Proof Test non standard} the information matrix $B_{Nt}(\gamma_1,\gamma_2)$ is Lipschitz for $\gamma_1$, $\gamma_2$ with constants $B^*_{1,Nt}$, $B^*_{2,Nt}$ having finite absolute moments. 
Moreover, by Supp. Mat.  \ref{Proof of Thm. Consistency and Asymptotic Normality nonlinaer PNAR as}, $B_{NT_N}(\gamma_1,\gamma_2)\xrightarrow{a.s.}B(\gamma_1,\gamma_2)$ $\forall\gamma_1,\gamma_2\in\Gamma$. Define $B^*_{s,NT}=T^{-1}\sum_{t=1}^{T}B^*_{s,Nt}$, for $s=\left\lbrace 1,2\right\rbrace$.  Following the same arguments of Supp. Mat.~\ref{Proof Lem. Consistency and Asymptotic Normality nonlinaer PNAR}-\ref{Proof of Thm. Consistency and Asymptotic Normality nonlinaer PNAR as} it can be proved that $B^*_{s,NT_N}-B^*_{s,N}\xrightarrow{a.s.}0$ where $B^*_{s,N}=\E(B^*_{s,Nt})$. Assumptions~\ref{Ass alpha mixing}-\ref{Ass limits existence} imply that $B^*_s=\lim_{N\to\infty}B^*_N$ is finite. Then, $B^*_{s,NT_N}\xrightarrow{a.s.}B^*_s$. This is a sufficient condition for $B_{NT}$ to be strongly stochastically equicontinuous \cite[Lem.~1]{andrews_1992generic} and, together with pointwise almost sure convergence, \cite[Thm.~2]{andrews_1992generic} shows that $B_{NT_N}(\gamma_1,\gamma_2)\xrightarrow{a.s.}B(\gamma_1,\gamma_2)$ uniformly over $\gamma_1, \gamma_2$.

Consider $\omega\in \Omega$, where $\Omega$ denotes a set of samples. We operate conditionally on the sample $\omega$, so randomness is through  the IID standard normal process $\nu_t$. Set $S^\nu_{NT}(\gamma)=\sum_{t=1}^{T}s^\nu_{Nt}(\gamma)$, with $s^\nu_{Nt}(\gamma)=s_{Nt}(\gamma)\nu_t$. Then, $\tilde{S}^\nu_{NT}(\gamma)=S^\nu_{NT}(\gamma)+\bar{S}_{NT}(\gamma)$, $\bar{S}_{NT}(\gamma)=\sum_{t=1}^{T}(\tilde{s}^\nu_{Nt}(\gamma)-s^\nu_{Nt}(\gamma))=\sum_{t=1}^{T}\sum_{i=1}^{N}(\tilde{s}^\nu_{i,t}(\gamma)-s^\nu_{i,t}(\gamma))$, and a.s.
\begin{align}
	\bar{S}_{NT}(\gamma)&=\sum_{t=1}^{T}\sum_{i=1}^{N}\left(\frac{Z_{i,t}(\gamma)\tilde{\xi}_{i,t}\nu_t}{\tilde{\lambda}_{i,t}}-\frac{Z_{i,t}(\gamma)\xi_{i,t}(\gamma)\nu_t}{\lambda_{i,t}(\gamma)} \right)\nonumber \\
	&=\sum_{t=1}^{T}\sum_{i=1}^{N}Z_{i,t}(\gamma)\left(\frac{\lambda_{i,t}(\gamma)\tilde{\xi}_{i,t}-\tilde{\lambda}_{i,t}\xi_{i,t}(\gamma)}{\tilde{\lambda}_{i,t}\lambda_{i,t}(\gamma)} \right)\nu_t \nonumber \\
	& \leq \beta_0^{-2}\sum_{t=1}^{T}\sum_{i=1}^{N}Z_{i,t}(\gamma)Z_{i,t}^\prime(\gamma)\left( \phi-\tilde{\phi} \right) Y_{i,t}\nu_t \,. \nonumber
\end{align}
Set $G^\nu_{NT}(\gamma)\coloneqq\beta_0^{-2}\sum_{t=1}^{T}\sum_{i=1}^{N}Y_{i,t}Z_{i,t}(\gamma)Z_{i,t}^\prime(\gamma)\nu_t$, so
\begin{equation}
	\sup_{\gamma\in \Gamma}\norm{\frac{\bar{S}_{NT}(\gamma)}{\sqrt{NT}}}_1\leq \sup_{\gamma\in \Gamma}\vertiii{\frac{G^\nu_{NT}(\gamma)}{NT}}_1\norm{\sqrt{NT}(\tilde{\phi}-\phi)}_1\,. \nonumber
\end{equation}
By Section~\ref{Proof Test non standard}, $s_t(\gamma)$ is $L^2$ integrable. 
Then, from the assumptions of Theorems~\ref{Thm. Consistency and Asymptotic Normality nonlinaer PNAR as}, \ref{Thm: Test non standard} and Pollard's central limit theorem for triangular empirical processes \cite[Thm.~10.6]{pollard_1990empirical}, the arguments in \cite[pp.~426-427]{hansen_1996} prove
that  
$(NT_N)^{-1/2}S^\nu_{NT_N}(\gamma)\Rightarrow_p S(\gamma)$, where $\Rightarrow_p$ denotes the weak convergence in probability, as defined in \cite{gine_zinn_1990}. Furthermore, we have $(NT_N)^{-1}G^\nu_{NT_N}(\gamma)\xrightarrow{a.s.} O_{k\times k}$. 
Then, $(NT_N)^{-1/2}\bar{S}_{NT_N}(\gamma)\Rightarrow_p 0$, $(NT_N)^{-1/2}\tilde{S}^\nu_{NT_N}(\gamma)\Rightarrow_p S(\gamma)$,
$LM^\nu_{NT_N}(\gamma)\Rightarrow_p LM(\gamma)$, $\tilde{g}_{NT_N}\Rightarrow_p g$, $\tilde{F}_{NT_N}(x)\xrightarrow{p}F(x)$, uniformly over $x$, and $\tilde{p}_{NT_N}=1-\tilde{F}_{NT_N}(g_{NT_N})=1-F(g_{NT_N})+o_p(1)=p_{NT_N}+o_p(1)$. 
\qed


\section*{Acknowledgments}
This work was completed when M. Armillotta was with the Department of Mathematics \& Statistics at the University of Cyprus. We greatly appreciate comments made by the Editor, AE and two reviewers  on an earlier version of the manuscript.

\vspace{-0.3cm}
\section*{Funding}
The research was supported by the European Regional Development Fund and the Republic of Cyprus through the Research and Innovation Foundation, under the project INFRASTRUCTURES/\break1216/0017 (IRIDA).

\setcounter{section}{18}
\begin{center}
	\section{Supplementary Material} \label{SEC supplementary material}
\end{center}

The supplement  contains  proofs for Sections \ref{SEC: inference} and  \ref{SEC: linearity test} and  detailed study of the  TNAR model \eqref{tnar} with fixed network dimension $N$. It also includes  key concepts used  in the main paper,  simulations and further data analysis,  as described in Sections~\ref{SEC: simulations}-\ref{SEC: applications}.

\renewcommand{\thesection}{S-\arabic{section}}
\renewcommand{\theequation}{S-\arabic{equation}}
\renewcommand{\thelemma}{S-\arabic{lemma}}
\renewcommand{\thefigure}{S-\arabic{figure}}
\renewcommand{\thetable}{S-\arabic{table}}
\setcounter{equation}{0}
\setcounter{figure}{0}
\setcounter{table}{0}
\setcounter{section}{0}
\renewcommand{\theproposition}{S-\arabic{proposition}}
\setcounter{example}{0}
\renewcommand{\theexample}{S\arabic{example}}

\vspace{0.5cm}

\section{Key concepts}
\label{Supp:Concepts}
\subsection{Joint copula-Poisson process}
\label{Supp:copula-Poisson}
We describe the joint copula-Poisson DGP for the the nonlinear PNAR \eqref{nonlinear_pnar}.
Consider a network $W$, a set of values for the parameters $\theta$ of \eqref{nonlinear_pnar} and a starting vector at time $t=0$, say $\lambda_0=(\lambda_{1,0},\dots,\lambda_{N,0})^\prime$,
\begin{enumerate}	
	\item Let $U_{l}=(U_{1,l},\dots,U_{N,l})$, for $l=1,\dots,K$ a sample from a $N$-dimensional copula $C(u_1,\dots, u_N)$, where $U_{i,l}$ follows a $Uniform(0,1)$ distribution, for $i=1,\dots, N$.
	\item The transformation $E_{i,l}=-\log{U_{i,l}}/\lambda_{i,0}$  follows the exponential distribution  with parameter $\lambda_{i,0}$, for $i=1,\dots, N$.
	\item If $E_{i,1}>1$, then $Y_{i,0}=0$, otherwise 
	$Y_{i,0}=\max\left\lbrace k\in[1,K]:  \sum_{l=1}^{k}E_{i,l}\leq 1\right\rbrace$, by taking $K$ large enough.
	Then, $Y_{i,0}|\lambda_0 \sim Poisson(\lambda_{i,0})$, for $i=1,\dots, N$. So, $Y_{0}=(Y_{1,0},\dots, Y_{N,0})^\prime$ is a set of (conditionally) marginal Poisson processes with mean $\lambda_0$. 
	\item By updating model \eqref{nonlinear_pnar}, $\lambda_1$ is obtained.
	\item Return back to step 1 to obtain $Y_1$, and so on.
\end{enumerate}
In  applications, the sample size $K$ should be taken large, e.g. $K=100$. Its value clearly depends, in general, on the magnitude of observed data. Moreover, the copula function  $C(\dots)$ depends on one or more unknown parameters, say $\rho$, which express  the contemporaneous correlation among the variables.

The proposed DGP ensures that all marginal distributions of $Y_{i,t}$ are univariate Poisson, conditionally to the past, as described in \eqref{nonlinear_pnar}, while it introduces  arbitrary dependence among them, in a flexible and general way, using  the copula construction. 
We  point out that the  assumed marginal Poisson distribution can be replaced by other assumptions, like Negative Binomial and more generally mixed Poisson distribution, by modifying suitable the above algorithm.
See \cite{armillotta_fokianos_2021} and \cite{fok2020} for further details on copula-Poisson joint distributions.

\subsection{Stationarity for increasing dimensional processes}
\label{Supp:Stationarity}
We follow \cite{zhu2017}.
Define $\{ X_{t}\in\R^N \}$ be an $N$-dimensional time series with $N\to\infty$ and $\mathcal{W} = \left\lbrace \omega\in\R^\infty:\omega_\infty=\sum\norm{\omega_i}<\infty\right\rbrace $, where $\omega = (\omega_i\in\R: 1\leq i \leq \infty)^\prime\in\R^\infty$. For each $\omega\in\mathcal{W}$, let $\omega_N = (\omega_1,\dots,\omega_N)^\prime\in\R^N$ be  its truncated $N$-dimensional version. Then  $\left\lbrace  X_{t} \right\rbrace $ is said to be strictly stationary if the following conditions hold: 
\begin{itemize}
	\item $X^\omega_t=\lim_{N\to\infty}\omega_N ^T X_t<\infty$,  a.s. (almost surely)
	\item $\left\lbrace X^\omega_t, ~ t\in\Z \right\rbrace $ is strictly stationary.
\end{itemize}

\section{Proofs  for Section  \ref{SEC: inference}}
\label{Supp:Sec:Proofs for Section 3}

\subsection{Proof of Lemma \ref{Lem. Consistency and Asymptotic Normality nonlinaer PNAR}}
\label{Proof Lem. Consistency and Asymptotic Normality nonlinaer PNAR}
We provide  a sketch of the proof of Lemma~\ref{Lem. Consistency and Asymptotic Normality nonlinaer PNAR}  because the proof is analogous  the corresponding proof of the linear PNAR model, as described in \cite[Lemma~1-2]{armillotta_fokianos_2021}. 
Set $\mu_1+\mu_2=d$. Recall model \eqref{nonlinear_pnar}. 
For $J>0$, define $\bar{Y}_t=f(\bar{Y}_{t-1},\theta)$, if $ t>0 $, and $\bar{Y}_t=Y_0$, if $ t\leq 0$. Moreover, $\hat{Y}^s_{t-J}=
f(\hat{Y}^{s-1}_{t-J},\theta)+\xi_s$, if $\max\left\lbrace t-J,0\right\rbrace < s\leq t $, and $\hat{Y}^s_{t-J}=\bar{Y}_s$, if $ s\leq\max\left\lbrace t-J,0\right\rbrace $ 
where $f(\hat{Y}^{t-1}_{t-J},\theta)=\hat{\lambda}^{t}_{t-J}$. Let $\tilde{Y}^{*}_{t}=\delta Y_t+(1-\delta)\bar{Y}_{t}$ and $\tilde{Y}_{t}=\delta Y_t+(1-\delta)\hat{Y}^{t}_{t-J}$ with $0\leq \delta\leq 1$. Then, a.s
\begin{equation}
	\norm{Y_t-\hat{Y}^t_{t-J}}_\infty\leq d^J\sum_{j=0}^{t-J-1}d^j\norm{\xi_{t-J-j}}_\infty\,,\label{recursion}
\end{equation}
where $\norm{\xi_t}_\infty=\max_{1\leq i\leq N}|\xi_{i,t}|$. For \eqref{contraction}, \cite[Lemma~ S-2]{armillotta_fokianos_2021} applies to \eqref{nonlinear_pnar}  and \eqref{recursion} holds true. 

Define $\hat{Y}_{i,t}$, $\hat{\lambda}_{i,t}$ the $i^{\text{th}}$ elements of $\hat{Y}^t_{t-J}$ and $\hat{\lambda}^t_{t-J}$. Set $l_{i,t}=Y_{i,t}\lambda_{i,t}d_{i,g}d_{i,gl}$, where $d_{i,g}=\partial f_i(X_{i,t-1}, Y_{i,t-1}, \theta)/\partial\theta_g$ and $d_{i,gl}=\partial^2 f_i(X_{i,t-1}, Y_{i,t-1}, \theta)/\partial\theta_g\partial\theta_l$. Set $a\geq 4$, a generic integer which can take different values. Recall that the notations $C_a, l_a$ denote constants which depend on $a$ and $c$ a generic constant.  By Theorem~\ref{Thm. Ergodicity of nonlinear model div N} and the a.s. inequality of Appendix~\ref{Proof of Thm. Ergodicity of nonlinear model div N} $\lambda_{i,t}=f_i(X_{i,t-1},Y_{i,t-1})\leq c+ \mu_1 X_{i,t-1} + \mu_2 Y_{i,t-1}$, for every $i=1,\dots,N$, we have that $\max_{i\geq 1}\lnorm{\lambda_{i,t}}_a\leq\max_{i\geq 1}\lnorm{Y_{i,t}}_a\leq C_a$, by the conditional Jensen's inequality. Similarly, $\max_{i\geq 1}\lnorm{\hat{\lambda}_{i,t}}_a\leq\max_{i\geq 1}\lnorm{\hat{Y}_{i,t}}_a$. By \eqref{recursion}, $\max_{i\geq 1}\lnorm{Y_{i,t}-\hat{Y}_{i,t}}_a 
\leq d^J2C_a^{1/a}/(1-d)$. An analogous recursion shows that $\max_{i\geq 1}\lnorm{\hat{Y}_{i,t}}_a 
\leq(2c+2C_a^{1/a})/(1-d)\coloneqq\Delta<\infty$. By Assumption \ref{smoothness}, $\norm{\partial f_i(x_i, y_i, \theta)/\partial\theta_g}\leq \max_{i\geq 1}\norm{\partial f_i(0, 0,\theta)/\partial\theta_g}+c_{1,g}x_i+c_{2,g}y_i$; so $\max_{i\geq 1}\lnorm{\partial f_i(\hat{x}_i, \hat{y}_i, \theta)/\partial\theta_g}_a\leq \hat{\Delta}_g<\infty$, $\max_{i\geq 1}\lnorm{\partial f_i(x_i, y_i, \theta)/\partial\theta_g}_a\leq \Delta_g<\infty$ for $i=1,\dots,N$ and $r=1,\dots,m$, where $\hat{\Delta}_g, \Delta_g$ are constants depending on index $g$. Similar arguments apply to second and third derivatives. Then, by H\"{o}lder's inequality, $\lnorm{l_{i,t}}_a\leq l_a<\infty$. Set $W_t=(Y_t, Y_{t-1})^\prime $, $\hat{W}^t_{t-J}=(\hat{Y}^{t}_{t-J}, \hat{Y}^{t-1}_{t-J})^\prime $. Consider the following triangular array $\left\lbrace g_{Nt}(W_t): 1\leq t\leq T_N, N\geq1\right\rbrace $, where $T_N\to\infty$ as $N\to\infty$. For any $\eta\in\R^m$, $g_{Nt}(W_t)=\sum_{g=1}^{m}\sum_{l=1}^{m}\eta_g\eta_lh_{glt}$, $N^{-1}h_{Nt}=(h_{glt})_{1\leq g,l\leq m}$ and $h_{Nt}=\sum_{i=1}^{N}-\partial^2 l_{i,t}(\theta)/\partial\theta\partial\theta^\prime$, where the second derivative is defined as in \eqref{H_T}. We take a single element, $h_{glt}$, the result is similarly proved for the other elements. By recalling \eqref{contraction} and Assumptions ~\ref{smoothness}-\ref{bounded away from 0}, following \cite[Sec.~A.2]{armillotta_fokianos_2021}, it can be proved that
$\lnorm{h_{glt}-h_{gl,t-J}^{t}}_2\leq 2l_aC_a^{1/a}/(1-d)d^{J-1}\coloneqq c_{gl}\nu_J$,
with $\nu_J=d^{J-1}$. This holds true for any $g,l=1,\dots,m$. Therefore, the triangular array process $\left\lbrace \bar{W}_{Nt}=g_{Nt}(W_t)-\E\left[ g_{Nt}(W_t)\right] \right\rbrace$ is $L^p$-near epoch dependence ($L^p$-NED), with $p\in[1,2]$. Moreover, using \ref{Ass alpha mixing} and the argument in \cite[p.~464]{and1988}, we have that $\left\lbrace \bar{W}_{Nt}\right\rbrace $ is a uniformly integrable $L^1$-mixingale and, by Assumption~\ref{Ass limits existence}, the law of large numbers of \cite[Thm.~2]{and1988} shows that $(NT_N)^{-1}\eta^\prime H_{NT_N}\eta\xrightarrow{p}\eta^\prime  H\eta$ as $ \left\lbrace N, T_N \right\rbrace \to\infty$. By the last inequality of \ref{Ass weak dependence}, an analogous argument, along the lines of \cite[Sec.~A.2-A.3]{armillotta_fokianos_2021},  applies to  the convergence of the information matrix $(NT_N)^{-1}\eta^\prime B_{NT_N}\eta\xrightarrow{p}\eta^\prime  B\eta$ and to the third log-likelihood  derivative defined by  Lemma~\ref{Lem. Consistency and Asymptotic Normality nonlinaer PNAR}. Assumption  \ref{Ass weak dependence}, the arguments in \cite[Sec.~A.3, Supp. Mat.~S-6]{armillotta_fokianos_2021} and  the central limit theorem for martingale arrays in \cite[Cor.~3.1]{hall1980} yields  $(NT_N)^{-1/2}S_{NT_N}\xrightarrow{d}N(0,B)$, as $N\to\infty$, leading to the desired result. \qed

\subsection{Proof of Theorem~\ref{Thm. Consistency and Asymptotic Normality nonlinaer PNAR as}} \label{Proof of Thm. Consistency and Asymptotic Normality nonlinaer PNAR as}
Recall that $c$ is a constant and $\mu_1+\mu_2=d$. By the assumptions of Lemma~\ref{Lem. Consistency and Asymptotic Normality nonlinaer PNAR}, $S_{Nt}/N$, $H_{Nt}/N$, $B_{Nt}/N$ and $M_{Nt}/N$ have  finite absolute fourth moments. Moreover, following the results of Section~\ref{Proof Lem. Consistency and Asymptotic Normality nonlinaer PNAR}, $H_{Nt}/N$, $B_{Nt}/N$ and $M_{Nt}$ are $L^p$-NED. Arguing as in  \cite[Def.~1]{dejong_1996strong} and \cite[p.~464]{and1988}, they are $L^2$-mixingales with coefficients $c_{Nt}= c$ and $\psi(J)=d^{[J/2]-1}+6\alpha([J/2])^{1-1/r}$, where $[\cdot]$ denotes the integer part operator. 
Note that $S_{Nt}/N$ is a martingale difference array so it is trivially $L^2$-mixingale with $\psi(J)=0$ \cite[Com.~(2), Example~1]{and1988}. Define $\zeta_N=[N^{1/3+\eta}]$ and $m_N=[N^{1/6-2\eta}]$, for some $\eta>0$ such that $m_N \geq 1$. Then, $\sum_{N=1}^{\infty}T_N^{-1}\sum_{t=1}^{T_N}\E|H_{Nt}/N|^4\zeta_N^{-3}\leq C_4\sum_{N=1}^{\infty}\zeta^{-3}_N \leq C_4\sum_{N=1}^{\infty}(N^{1/3+\eta})^{-3} < \infty$. 
Moreover, for all $\delta>0$ 
\begin{align}
	&\sum_{N=1}^{\infty}m_N \exp \left(-\delta^2 T^2_N m_N^{-2}\left( \sum_{t=1}^{T_N}\zeta^2_N\right) ^{-1}\right)  = \sum_{N=1}^{\infty} m_N\exp \left(-\delta^2 T_N (m_N\zeta_N)^{-2}\right) \nonumber \\
	& \leq \sum_{N=1}^{\infty} N^{1/6-2\eta} \exp \left(-\delta^2T_N (N^{1/6-2\eta} N^{1/3+\eta})^{-2}\right) \leq \sum_{N=1}^{\infty} N^{1/6-2\eta} \exp \left(-\delta^2\lambda N^{2 \eta}\right) < \infty\,.  \nonumber
\end{align} 
where the last inequality follows by $T_N=\lambda N$.
Finally, note that  $\sum_{N=1}^{\infty}\left( T_N^{-1}\sum_{t=1}^{T_N}c_{Nt}\psi(m_N)\right)^2\leq c^2 \sum_{N=1}^{\infty}\psi(m_N)^2$ is convergent if $\psi(J)=\mathcal{O}(J^{-3-\epsilon})$, by appropriate choice of $\eta>0$; since $d^{[J/2]-1}=\mathcal{O}(J^{-3-\epsilon})$, the assumption of Theorem~\ref{Thm. Consistency and Asymptotic Normality nonlinaer PNAR as} provides the result. Then, conditions 1.(c), 2-3 of \cite[Thm.~4]{dejong_1996strong} are satisfied and, by Assumption~\ref{Ass limits existence}, as $\left\lbrace N, T_N \right\rbrace \to\infty$, $T^{-1}_N\sum_{t=1}^{T_N}N^{-1}H_{Nt}=(NT_N)^{-1}H_{NT_N}\xrightarrow{a.s.} H$. Analogous results hold for the other mixingales. An application of \cite[Thm.~3.2.23]{tani2000} ends the proof. \qed

\subsection{Proof of Theorem~\ref{Prop condition limits nonlinear pnar}}
\label{Proof of Prop condition limits nonlinear pnar}
The process $Y_t=\lambda_t+\xi_t$, where $\lambda_t$ is defined as in \eqref{nonlinear}, can be rewritten as $Y_t=\beta_0 C_{t-1}+ GY_{t-1}+\xi_t$. By backward substitution, $Y_t=\beta_0\sum_{j=0}^{t-1}G^jC_{t-1-j}
+G^tY_0+\sum_{j=0}^{t-1}G^j\xi_{t-j}$. Under the conditions of Theorem~\ref{Thm. Ergodicity of nonlinear model div N}, as $t\to\infty$, the second term is negligible and the last term is well-defined as $\tilde{Y}_t\coloneqq\sum_{j=0}^{\infty}G^j\xi_{t-j}$. It remains to show the finiteness of the first summand $\mu_{t-1}\coloneqq\beta_0\sum_{j=0}^{t-1}G^jC_{t-1-j}$.  Note that $\mu_{t-1}$ is increasing sequence with respect to $t$, moreover $C_t\preceq 1$, so $\mu_{t-1}\preceq \beta_0\sum_{j=0}^{t-1}G^j1=\beta_0\sum_{j=0}^{t-1}(\beta_1+\beta_2)^j1\preceq\mu 1$, where $\mu\coloneqq\beta_0(1-\beta_1-\beta_2)^{-1}$, then $\mu^{\infty}_{t-1}\coloneqq\beta_0\sum_{j=0}^{\infty}G^jC_{t-1-j}<\infty$ a.s. and $Y_t=\mu^{\infty}_{t-1}+\tilde{Y}_t$. The proof of Theorem~\ref{Prop condition limits nonlinear pnar} follows by Lemma~\ref{Lem. Consistency and Asymptotic Normality nonlinaer PNAR}, we only need to show that assumptions \ref{Ass network}-\ref{Ass final limits} imply \ref{Ass limits existence} for model \eqref{nonlinear_pnar}. We consider one element of the Hessian matrix, since we can work analogously  for all other elements; see also \cite[Sec.~A.2]{armillotta_fokianos_2021}. Consider the element in position $(1,2)$ of matrix \eqref{H} and the decomposition $Y_t=\mu^{\infty}_{t-1}+\tilde{Y}_t$, so $H_{12}=H_{12a}+H_{12b}$, where $N^{-1}H_{12a}=N^{-1}\E(\Gamma_{1,t-1}^\prime D_t^{-1}\Gamma_{2,t-1})$ converges to $l^H_{12}$, by \ref{Ass final limits}; indeed
\begin{equation}
	\norm{\frac{H_{12b}}{N}}\leq \frac{1}{N}E\left(\sum_{i=1}^{N}\frac{w_i^\prime \norm{\tilde{Y}_{t-1}}_{vec}}{\lambda_{i,t}(1+X_{i,t-1})^\gamma}\right)\leq \frac{1}{C N}E\left( 1^\prime W\norm{\tilde{Y}_{t-1}}_{vec} \right) \xrightarrow{N\to\infty}0\,,
	\nonumber
\end{equation}
where the last inequality follows by Assumption~\ref{bounded away from 0}; the convergence holds by \ref{Ass network}, \cite[Eq.~A-2]{armillotta_fokianos_2021} and an application of \cite[Lemma~S-1]{armillotta_fokianos_2021}. For the information matrix \eqref{B} similar results hold true, by considering Assumption~\ref{Ass weak dependence}. The third derivative of the likelihood \eqref{pois-log-lik} has the form derived in \cite[Supp. Mat. S-4]{fok2020}. Moreover, $\partial \lambda_{i,t}/\partial\theta_k=\Gamma_{i,k,t-1}$, for $k=1,2,3$ and $\lambda_{i,t}/\partial\theta_4=-\beta_0\Gamma_{i,4,t-1}$; the second derivatives are all zeros apart from $\partial^2 \lambda_{i,t}/\partial\beta_0\partial\gamma=-\Gamma_{i,4,t-1}$ and $\partial^2 \lambda_{i,t}/\partial\gamma^2=\beta_0F_{i,t-1}$; the third derivatives are all zeros apart from $\partial^3 \lambda_{i,t}/\partial\beta_0\partial\gamma^2=F_{i,t-1}$ and $\partial^3 \lambda_{i,t}/\partial\gamma^3=-\beta_0J_{i,t-1}$.
Condition of \ref{Ass limits existence}  is satisfied by recalling the $d^*$ conditions in Assumption~\ref{Ass final limits}. We omit the details.
\qed

\section{Proofs for Section   \ref{SEC: linearity test}}
\label{Supp:Sec:Proofs for Section 4}

\subsection{Proof of Theorem~\ref{limits}} \label{Proof general nonlinear score}
Under the condition of Lemma \ref{Lem. Consistency and Asymptotic Normality nonlinaer PNAR}, Theorem \ref{Thm. Consistency and Asymptotic Normality nonlinaer PNAR} and  \ref{interior H0}, the constrained estimator $\tilde{\theta}$ is consistent for $\theta_0$, when $H_0$ holds. See, for example, \cite[Thm.~3.2.23]{tani2000}. So $\tilde{\theta}^{(1)}$ is consistent for $\theta^{(1)}_0$ and, for $\left\lbrace N, T_N \right\rbrace $ large enough, we have $S_{NT}^{(1)}(\tilde{\theta})=0$ with $\tilde{\theta}^{(1)}\neq0$. Let $J_1=(I_{m_1}, O_{m_1\times m_2})$, $J_2=(O_{m_2\times m_1}, I_{m_2})$, where $I_s$ is a $s\times s$ identity matrix and $O_{a\times b}$ is a $a\times b$ matrix of zeros. Therefore, $0=S_{NT}^{(1)}(\tilde{\theta})=J_1S_{NT}(\tilde{\theta})$ and $S_{NT}(\tilde{\theta})=J_2^\prime S_{NT}^{(2)}(\tilde{\theta})$.
Lemma \ref{Lem. Consistency and Asymptotic Normality nonlinaer PNAR} and Taylor's theorem provides 
$0=(NT)^{-1/2}S^{(1)}_{NT}(\tilde{\theta})\doteq (NT)^{-1/2} S^{(1)}_{NT}(\theta_0)-J_1H(NT)^{1/2}(\tilde{\theta}-\theta_0)$, and $
(NT)^{-1/2} S^{(2)}_{NT}(\tilde{\theta}) \doteq (NT)^{-1/2} S^{(2)}_{NT}(\theta_0)-J_2H(NT)^{1/2}(\tilde{\theta}-\theta_0)$, where $\doteq$ means equality up to a $o_p(1)$ term.
Since $\tilde{\theta}-\theta_0=J^\prime_1(\tilde{\theta}^{(1)}-\theta^{(1)}_0)$ we have $0\doteq(NT)^{-1/2}S^{(1)}_{NT}(\theta_0)-J_1HJ^\prime_1(NT)^{1/2}(\tilde{\theta}^{(1)}-\theta^{(1)}_0)$ and $(NT)^{1/2}(\tilde{\theta}^{(1)}-\theta^{(1)}_0)\doteq H_{11}^{-1} (NT)^{-1/2} S^{(1)}_{NT}(\theta_0)$,
where $H_{11}=J_1HJ^\prime_1$. Combining the above results we obtain 
$			(NT)^{-1/2} S^{(2)}_{NT}(\tilde{\theta}) \doteq V_2-H_{21}H_{11}^{-1}V_1=PV$, where $J_2HJ_1^\prime=H_{21}$, $V=(NT)^{-1/2}S_{NT}(\theta_0)$ and $P=\left( -H_{21}H_{11}^{-1}  I_{m_2}\right) $. Then, $V\xrightarrow{d}N(0, B)$, as $N,T_N\to\infty$
leading to 
\begin{equation}
	\frac{S^{(2)}_{NT}(\tilde{\theta})}{\sqrt{NT}}\doteq PV\xrightarrow[N,T_N\to\infty]{d}N(0, \Sigma)\,,
	\label{normality}
\end{equation}
where $\Sigma=PBP^\prime$ has the expression \eqref{sigma}. Following \cite{engle_1984}, \cite{gal1987},  and \cite{fran2019}, among others, the general formulation of the score test statistic for the nonlinear parameters of the model is
\begin{align}
	LM_{NT}&=
	\frac{S^{(2)\prime}_{NT}(\tilde{\theta})}{\sqrt{NT}}\Sigma^{-1}\frac{S^{(2)}_{NT}(\tilde{\theta})}{\sqrt{NT}}\,\label{test_1}
\end{align}
Then, by replacing $\Sigma$ in \eqref{test_1} with
the result follows by Lemma~\ref{Lem. Consistency and Asymptotic Normality nonlinaer PNAR},  \ref{interior H0} 
and \eqref{normality}, under $H_0$,
\begin{equation}
	LM_{NT}= S^{(2)\prime}_{NT}(\tilde{\theta})\Sigma_{NT}^{-1}(\tilde{\theta}) S^{(2)}_{NT}(\tilde{\theta})\xrightarrow[N,T_N\to\infty]{d}\chi^2_{m_2}\,.
	\nonumber
\end{equation} 
To determine the asymptotic distribution of the test statistic $LM_{NT}$ under $H_1$ in \eqref{test_local} note that $\theta= \theta_0+(NT)^{-1/2}\delta$, where $\delta=(\delta_1, \delta_2)^\prime$. A Taylor expansion around $\theta$ shows that  $(NT)^{-1/2}S_{NT}(\theta_0)\doteq(NT)^{-1/2}S_{NT}(\theta)-H_{NT}(\theta)(NT)^{-1/2}(\theta_0-\theta)= (NT)^{-1/2} S_{NT}(\theta) \break + (NT)^{-1} H_{NT}(\theta)\delta$.
Since, as $ \left\lbrace N, T_N \right\rbrace \to\infty$, $\theta\xrightarrow{p}\theta_0$, and, by mean value theorem and Lemma \ref{Lem. Consistency and Asymptotic Normality nonlinaer PNAR} $(NT)^{-1/2}S_{NT}(\theta)\xrightarrow{d}N(0, B)$ and $(NT)^{-1}H_{NT}(\theta)\xrightarrow{p}H$, under $H_1$, $V\xrightarrow{d}N(H\delta, B)$.
Hence,
\begin{equation}
	\frac{S^{(2)}_{NT}(\tilde{\theta})}{\sqrt{NT}}\doteq PV\xrightarrow[N,T_N\to\infty]{d}N(\Sigma_H\delta_2, \Sigma)\,,
	\label{normality local} 
\end{equation}
under $H_1$, where $\Sigma_H\equiv (H^{22})^{-1}$, with $H^{22}=J_2H^{-1}J_2^\prime$, and $\Sigma_H\delta_2=PH\delta$. By \eqref{test_1},
\begin{equation}
	LM_{NT}= S^{(2)\prime}_{NT}(\tilde{\theta})\Sigma_{NT}^{-1}(\tilde{\theta}) S^{(2)}_{NT}(\tilde{\theta})\xrightarrow[N,T_N\to\infty]{d}\chi^2_{m_2}(\delta_2^{\prime}\tilde{\Delta}\delta_2)\,.
	\nonumber
\end{equation}
under $H_1$, where $\chi^2_{m_2}(\cdot)$ is the noncentral chi-square distribution and $\tilde{\Delta}$ is the sample counterpart of $\Delta=\Sigma_H\Sigma^{-1}\Sigma_H$ computed according to Remark~\ref{Rem. sample information matrix}.\qed	

\section{Threshold models}
\label{threshold models}
In this section we give a detailed study   for the TNAR model \eqref{tnar}. In this case the function $f(\cdot)$ is not continuous with respect to the threshold variable $X_t$. Therefore, the contraction condition \eqref{contraction} does not hold. In addition, the Lipschitz conditions, stated in Assumption~\ref{smoothness}, are not satisfied due to the discontinuity of the process. However, we  establish stability conditions for fixed network dimension (stationarity and geometric ergodicity)  for model  \eqref{tnar}, by establishing absolute regularity of
the process.

\begin{proposition}  \label{Prop: ergodicity tar}
	Consider model \eqref{general nuisance pnar}, with process $\{\lambda_t\}$ defined by \eqref{tnar} and  $\vertiii{G}_1<1$, where $G=(\beta_1+\alpha_1)W+(\beta_2+\alpha_2)I$. 
	Suppose the probability mass function of $Y_t\left| \right. \mathcal{F}_{t-1}$ is positive. 
	Then, the TNAR model is stationary, ergodic and absolutely regular with geometrically decaying coefficients.
\end{proposition}

\textit{Proof of Proposition~\ref{Prop: ergodicity tar}} \label{Proof ergodicity tnar} 
Consider model \eqref{tnar}. for $i=1,\dots,N$, a.s.
\begin{align}
	\norm{f_i(X_{i,t-1},Y_{i,t-1})}&\leq (\beta_0 + \alpha_0) + (\beta_1+\alpha_1) X_{i,t-1} + (\beta_2+\alpha_2) Y_{i,t-1} \nonumber \\
	& = (\beta_0 + \alpha_0) + \left[(\beta_1+\alpha_1)e_i^\prime W+(\beta_2+\alpha_2)e_i^\prime \right]\norm{Y_{t-1}}_{vec} \,,\nonumber 
\end{align}
implying that $\norm{f(Y_{t-1})}_{vec}\preceq c1+ G\norm{Y_{t-1}}_{vec}$, with $G=(\beta_1+\alpha_1) W+(\beta_2+\alpha_2) I $ and $c=\beta_0+\alpha_0$. It is immediate to see that the process $\left\lbrace  Y_t: t\in\Z \right\rbrace $, defined by \eqref{general nuisance pnar} with $\lambda_t$ specified as in \eqref{tnar} is a Markov chain on countable state space $\N^N$. 
Since $Y_t$ given $Y_{t-1}$  has a positive probability mass function, the first order transition probabilities of the Markov chain are $p_{m,l}=\mathrm{P}(Y_t=m | Y_{t-1}=l)>0$, for all $m,l$. This implies irreducibility and aperiodicity of the chain. See, for example, \cite[Ch.~4-5]{MeynandTweedie(2009)}. 
A geometric drift condition is established, by the property of the Poisson process, $\forall y\in\N^N$, 
\begin{equation}
	\E\left( \norm{Y_t}_1 + 1 | Y_{t-1}=y \right) =  1 +\E\norm{N_t[f(y)]}_1 = 1 +\norm{f(y)}_1\leq c^* + d_1(\norm{y}_1+1) \nonumber\,, 
\end{equation}
where $c_0=Nc$, $c^*=c_0+1 < \infty$ and $d_1=\vertiii{G}_1 < 1$. Consider the non-empty compact sets 
$S^1_M=\left\lbrace y : \norm{y}_1 + 1 \leq M\right\rbrace$, with real constants $M$. Since the state space is countable and we have that $\epsilon =\sum_{m\in\N^N} \inf_{l\in S^1_M} p_{m,l} >0$, by \cite[eq.~8-9]{roberts_2004}, every non-null compact set $S^1_M$ is small. By considering the small set $S^1=\left\lbrace y : \norm{y}_1 + 1 \leq c^*/\delta_1\right\rbrace$, with $\delta_1=1/2(1-d_1)$, \cite[Thm~15.0.1]{MeynandTweedie(2009)} implies that the stochastic process $\left\lbrace Y_t: t\in\Z \right\rbrace $ is stationary and absolutely regular with geometrically decaying coefficients. 
\qed

An analogous result holds for the continuous case. Define $\dot{\beta}_1=\max\left\lbrace \norm{\beta_1}, \norm{\beta_1+\alpha_1} \right\rbrace $ and $\dot{\beta}_2=\max\left\lbrace \norm{\beta_2}, \norm{\beta_2+\alpha_2} \right\rbrace $.

\begin{proposition}  \label{Prop: ergodicity tar cont}
	Consider model \eqref{general nuisance nar}, with process $\{\lambda_t\}$ defined as in \eqref{tnar} and $\dot{\beta}_1+\dot{\beta}_2<1$. The marginal cumulative distribution of $\xi_t$ is absolutely continuous, with positive probability density function, and $\E\norm{\xi_t}_1<\infty$. 
	Then, the TNAR model is stationary, ergodic and absolutely regular with geometrically decaying coefficients.
\end{proposition}

\textit{Proof of Proposition~\ref{Prop: ergodicity tar cont}}
\label{Proof ergodicity tar cont}
Consider model \eqref{tnar}. Define $\dot{\beta}_0=\max\left\lbrace \norm{\beta_0}, \norm{\beta_0+\alpha_0}\right\rbrace $ and $\dot{d}=\dot{\beta}_1+\dot{\beta}_2$. For $i=1,\dots,N$, a.s.
\begin{equation}
	\norm{f_i(X_{i,t-1},Y_{i,t-1})}\leq \dot{\beta}_0 + \dot{\beta}_1 X_{i,t-1} + \dot{\beta}_2 Y_{i,t-1} = \dot{\beta}_0 + \left(\dot{\beta}_1e_i^\prime W+\dot{\beta}_2e_i^\prime \right)\norm{Y_{t-1}}_{vec} \,,\nonumber 
\end{equation}
implying that $\norm{f(Y_{t-1})}_{vec}\preceq \dot{\beta}_01+ G\norm{Y_{t-1}}_{vec}$, with $G=\dot{\beta}_1 W+\dot{\beta}_2 I $. 
By \cite{chan_tong_1985} the Markov chain $\left\lbrace  Y_t: t\in\Z \right\rbrace $ defined by \eqref{general nuisance nar} with $\lambda_t$ defined in \eqref{tnar}, is irreducible and aperiodic; moreover, since the p.d.f of $\xi_t$  is lower semi-continuous and the function $f(\cdot)$ is compact, all the non-empty compact sets $S_M=\left\lbrace y : \norm{y}_\infty + 1\leq M\right\rbrace$ with real constant $M$, are small.  $\forall y\in\R^N$, 
\begin{equation}
	\E\left( \norm{Y_t}_\infty + 1 | Y_{t-1}=y \right)\leq c+\dot{d}\norm{y}_\infty+1\leq c^* + \dot{d}(\norm{y}_\infty+1) \label{drift condition}\,,
\end{equation}
where $\E\norm{\xi_t}_\infty\leq \E\norm{\xi_t}_1<c_\xi$, $c=\dot{\beta}_0 + c_\xi$, $c^*=c+1 < \infty$ and $0 < \dot{d} < 1$.
Then, the process $\left\lbrace Y_t: t\in\Z \right\rbrace $ is stationary and absolutely regular with geometrically decaying coefficients. 
\qed

\subsection{Inference and testing for TNAR models}
We state an analogous result to Theorem~\ref{Thm. Consistency and Asymptotic Normality nonlinaer PNAR}, for  $N$ fixed and $T\to\infty$. In this case, we drop the dependence of $N$ to point out that the inference is for $T\to\infty$. The following result implies, in particular,  good large sample properties of the QMLE for model \eqref{tnar} when the threshold parameter $\gamma$ is known.

\begin{corollary} \label{Cor. Consistency and Asymptotic Normality nonlinaer PNAR}
	Consider model \eqref{nonlinear_pnar}. Let $\theta\in\Theta\subset\R^m_{+}$. Suppose the conditions of Theorem~\ref{Thm. Ergodicity of nonlinear model}, Assumption~\ref{interior}, \ref{smoothness} (with $x_i^*=y_i^*=0$) and \ref{bounded away from 0} hold. Then, there exists a fixed open neighbourhood $\mathcal{O}(\theta_0)=\left\lbrace \theta:|\theta-\theta_0|_2<\delta\right\rbrace$ of $\theta_0$ such that with probability tending to 1  as $T\to\infty$, for the score function \eqref{score_poisson}, the equation $S_{T}(\theta)=0$ has a unique solution, called $\hat{\theta}$, which is (strongly) consistent and asymptotically normal:
	\begin{equation}
		\sqrt{T}(\hat{\theta}-\theta_0)\xrightarrow{d}N(0,H^{-1}BH^{-1})\,,
		\nonumber
	\end{equation}
	where the matrices $H\equiv H_N$ and $B \equiv B_N$ are defined in \eqref{H}-\eqref{B}, respectively.
\end{corollary}
This follows by an application of the ergodic theorem and analogous arguments to \cite[Sec.~S-3.3]{armillotta_fokianos_2021}.  Assumption \ref{smoothness} is imposed  with $x_i^*=y_i^*=0$ to provide bounds of  all necessary moments.
Corollary~\ref{Cor. Consistency and Asymptotic Normality nonlinaer PNAR} shows weak  consistency of QMLE but  strong consistency is   proved under the same assumptions.
Recall from Section~\ref{SEC: identifiability} that $\theta=(\phi^\prime, \gamma)^\prime$, where $\gamma$ is the threshold parameter and $\phi$ is the vector of, linear and nonlinear, identifiable parameters.

Note that 
\ref{smoothness}  with $x_i^*=y_i^*=0$
is satisfied for $\theta=\phi$ under the threshold model \eqref{tnar}. Also condition~\ref{bounded away from 0} holds for TNAR models. Then, we have the following result.

\begin{corollary}  \label{Cor: inference tar}
	Consider model \eqref{general nuisance pnar}, with process $\{\lambda_t\}$ defined as in \eqref{tnar} where $\gamma$ is known.
	Suppose the conditions of Proposition~\ref{Prop: ergodicity tar} and Assumption~\ref{interior} hold. Then, Corollary~\ref{Cor. Consistency and Asymptotic Normality nonlinaer PNAR} holds true. 
\end{corollary}
Estimation of the  $\gamma$ can be tackled by profiling likelihood. This is a common practice; see for example 
\cite{Wangetal(2014)} among others.
A similar result is established for continuous-valued data.	

We now provide the related result for testing the linearity of model \eqref{nar_1} versus \eqref{tnar}, i.e. $H_0:\alpha_0=\alpha_1=\alpha_2=0$ vs. $ H_1: \alpha_l>0$, for some $l=0,1,2$.

\begin{proposition}  \label{Prop: Test non standard tar}
	Consider model \eqref{general nuisance pnar}, with process $\{\lambda_t\}$ defined as in \eqref{tnar}. Suppose conditions of Proposition~\ref{Prop: ergodicity tar} hold true. Moreover, assume that $X_t$ has marginal probability density function $f_{X_i}$, such that $\sup_{x_i\in\R}f_{X_i}(x_i)=\bar{f}_{X}<\infty$. Then, under \ref{interior H0}, Theorem~\ref{Thm: Test non standard} holds true for $T\to\infty$ and fixed $N$.
\end{proposition}

\textit{Proof of Proposition~\ref{Prop: Test non standard tar}}
Note that $I(X_{i,t-1}\leq \gamma)$, as a function of $\gamma$, is upper semi-continuous, then $\sup_{\gamma\in \Gamma}I(X_{i,t-1}\leq \gamma)$ exists and equals 1. Finally, the rows of $-Z_t(\gamma)$ are linearly independent for some $t$, uniformly over $\gamma\in\Gamma\equiv[\gamma_L,\gamma_U]$, by considering that $\inf_{\gamma\in\Gamma}-h_{i,t}(\gamma)=-\sup_{\gamma\in \Gamma}I(X_{i,t-1}\leq \gamma)Z_{1t}\neq0$ componentwise, for some $t$, by the ergodicity of $Y_t$. So we have that $\inf_{\gamma\in\Gamma}\text{det}(H(\gamma,\gamma))>0$. This fact, the finiteness of all the moments of $Y_t$ and $\xi_t(\gamma)$, and the results of Proposition~\ref{Prop: ergodicity tar} satisfy \cite[Asm.~1]{hansen_1996}.
Set $\gamma > \gamma^*$, for $a > 1$, 
\begin{align}
	\lnorm{h_{i,t}(\gamma)-h_{i,t}(\gamma^*)}_{a}&\leq\lnorm{\left( 1+X_{i,t-1}+ Y_{i,t-1}\right) I\left( \gamma^*< X_{i,t-1} \leq \gamma \right) }_{a} \label{expectation h}\\
	&\leq \left( 1 + \lnorm{X_{i,t-1}}_{2a}+ \lnorm{Y_{i,t-1}}_{2a} \right)\left( \int_{\gamma^*}^{\gamma} f_{X_i}(x_i)dx_i \right) ^\delta \nonumber \\
	&\leq B_1\norm{\gamma-\gamma^*}^\delta \nonumber
\end{align}
where $\delta=1/(2a)$ and $B_1=\left( 1+2C\right) \bar{f}_{X}^\delta$. The same holds for $\gamma < \gamma^*$. Then, \cite[Asm.~2]{hansen_1996} holds true. 

Now we show uniform convergence of all quantities involved in  inference. \cite[Asm.~3]{hansen_1996}. 
Standard arguments show that the pointwise convergence of $H_T(\gamma_1,\gamma_2)$ and $B_T(\gamma_1,\gamma_2)$ hold, by the measurability of all the functions involved and the existence of moments of any order for the count process $Y_t$. 
For all $\gamma_1,\gamma_2\in\Gamma$, define $\bar{\gamma}=(\gamma_1^\prime,\gamma_2^\prime)^\prime$, $B(\bar{\gamma},\rho)=\left\lbrace \bar{\gamma}^*\in\Gamma: \norm{\bar{\gamma}^*-\bar{\gamma}}_1< \rho \right\rbrace$  and consider the following random variables.
\begin{equation}
	H_t^*(\bar{\gamma},\rho)=\sup\left\lbrace H_t(\bar{\gamma}^*): \bar{\gamma}^*\in B(\bar{\gamma},\rho)\right\rbrace\,, \quad H_t^\star(\bar{\gamma},\rho)=\inf\left\lbrace H_t(\bar{\gamma}^*): \bar{\gamma}^*\in B(\bar{\gamma},\rho)\right\rbrace\,. \nonumber
\end{equation}
We need to prove, for all $\bar{\gamma}\in\Gamma$,
\begin{equation}
	\lim_{\rho\to 0} \sup_{t\geq 1}\norm{\frac{1}{T}\sum_{t=1}^{T}\left[ \E H_t^*(\bar{\gamma}, \rho)- \E H_t(\bar{\gamma})\right] }_1=0 \,, \label{equicont}
\end{equation}
and the same limit should hold for $H_t^\star(\bar{\gamma},\rho)$, as well. Set $\bar{\gamma}^*\prec \bar{\gamma}$. For clarity in the notation, we prove \eqref{equicont} for the single element $H^{(g,l)}_t(\bar{\gamma})$, which is bounded by
\begin{align}
	&\lim_{\rho\to 0}\sum_{i=1}^{N}\E \norm{ Y_{i,t}\left(\sup_{\bar{\gamma}^*\in B(\bar{\gamma}, \rho)} \frac{h^{g}_{i,t}(\gamma_1^*)h^{l}_{i,t}(\gamma_2^*)}{\lambda_{i,t}(\gamma_1^*)\lambda_{i,t}(\gamma_2^*)}-\frac{h^{g}_{i,t}(\gamma_1)h^{l}_{i,t}(\gamma_2)}{\lambda_{i,t}(\gamma_1)\lambda_{i,t}(\gamma_2)}\right)}\nonumber\\
	&\leq
	\lim_{\rho\to 0}\sum_{i=1}^{N}\E \left(\frac{Y_{i,t}}{\beta_0^4}\sum_{s=1}^{4}\sup_{\bar{\gamma}^*\in B(\bar{\gamma}, \rho)}\norm{p_s}\right) \,. \nonumber
\end{align}
Note that $\sup_{\bar{\gamma}^*\in B(\bar{\gamma}, \rho)}h^g_{i,t}(\gamma_1^*)\leq h^g_{i,t}(\gamma_1+\rho)\leq 1$ and $\inf_{\bar{\gamma}^*\in B(\bar{\gamma}, \rho)}h^g_{i,t}(\gamma_1^*)\geq  h^g_{i,t}(\gamma_1-\rho)\geq 0$ a.s.; the same holds for $\gamma^*_2$. So $\sup_{\bar{\gamma}^*\in B(\bar{\gamma}, \rho)}\norm{h^l_{i,t}(\gamma_2^*)-h^l_{i,t}(\gamma_2)}= \sup_{\bar{\gamma}^*\in B(\bar{\gamma}, \rho)}(h^l_{i,t}(\gamma_2)-h^l_{i,t}(\gamma_2^*))=h^l_{i,t}(\gamma_2)-\inf_{\bar{\gamma}^*\in B(\bar{\gamma}, \rho)}h^l_{i,t}(\gamma_2^*) \leq  h^l_{i,t}(\gamma_2)-h^l_{i,t}(\gamma_2-\rho)$. Then, 
\begin{align}
	&\lim_{\rho\to 0}\sum_{i=1}^{N}\E\left( \frac{Y_{i,t}}{\beta_0^4}\sup_{\bar{\gamma}^*\in B(\bar{\gamma}, \rho)}\norm{p_1}\right)  \nonumber \\
	&\leq \lim_{\rho\to 0}\sum_{i=1}^{N}\E\left( \frac{Y_{i,t}}{\beta_0^4}\lambda_{i,t}(\gamma_1)\lambda_{i,t}(\gamma_2)\sup_{\bar{\gamma}^*\in B(\bar{\gamma}, \rho)}h^{g}_{i,t}(\gamma_1^*)\norm{h^{l}_{i,t}(\gamma_2^*) - h^{l}_{i,t}(\gamma_2)}\right)  \nonumber\\
	& \leq C \lim_{\rho\to 0} \lnorm{h^l_{i,t}(\gamma_2)-h^l_{i,t}(\gamma_2-\rho)}_a \nonumber \\
	&\leq C B_1\lim_{\rho\to 0} \norm{\gamma_2-\gamma_2+\rho}^{1/(2a)}  \nonumber \\
	& =0 \,, \nonumber
\end{align}
where the second inequality holds by H\"{o}lder's inequality and the third inequality holds by \eqref{expectation h}. Analogous results hold for $p_s$, with $s=2,3,4$; the same limit is obtained for $\bar{\gamma}\prec \bar{\gamma}^*$ and the cross inequalities $\gamma_i^*< \gamma_i, \gamma_j<\gamma_j^*$, with $i,j=1,2$ and $i\neq j$. Similar arguments apply (inverted) to $H_t^\star(\bar{\gamma},\rho)$ and to the information matrix. Then, the uniform law of large number of \cite{andrews_1987} holds, leading to $T^{-1}H_T(\gamma_1,\gamma_2)\xrightarrow{a.s.}H(\gamma_1,\gamma_2)$, $T^{-1}B_T(\gamma_1,\gamma_2)\xrightarrow{a.s.}B(\gamma_1,\gamma_2)$, uniformly in $\gamma_1,\gamma_2\in\Gamma$. Then, \cite[Thm.~1]{hansen_1996} applies to model \eqref{tnar}.

\begin{proposition}  \label{Prop: Test non standard tar cont}
	Consider model \eqref{general nuisance nar}, with process $\{\lambda_t\}$ defined as in \eqref{tnar}. Suppose the conditions of Proposition~\ref{Prop: ergodicity tar cont} hold and $\E\norm{\xi_t}^{8}_{8}<\infty$. 
	Moreover, assume that $X_t$ has marginal probability density function $f_{X_i}$, such that $\sup_{X_i\in\R}f(x_i)=\bar{f}_{X}<\infty$. Then, under \ref{interior H0}, Theorem~\ref{Thm: Test non standard cont} holds true for $T\to\infty$ and fixed $N$.
\end{proposition}
The proof is omitted since is analogous to the integer-valued case. The convergence of bootstrap $p$-values as in Theorem~\ref{Thm: Bootstrap} (with fixed $N$) is obtained under the conditions of Propositions~\ref{Prop: Test non standard tar}-\ref{Prop: Test non standard tar cont} by \cite[Thm.~2]{hansen_1996}.

\section{Additional simulations and  empirical results}
\label{SuppSec:Additionalsimulations}

\subsection{Simulations} 


\label{simu standard testing}

For the integer-valued case, we test linearity of the PNAR model against the nonlinear version in \eqref{nonlinear}. The observed count time series $\left\lbrace Y_t \right\rbrace $ is generated as in \eqref{nonlinear_pnar}, using the copula-based data generating process of Section~\ref{Supp:copula-Poisson}. 
We consider a Gaussian copula  $C^{Ga}_R(\dots)$, with correlation matrix $R=(R_{ij})$, where $R_{ij}=\rho^{|i-j|}$, the so called Toeplitz correlation matrix; $\rho=0.5$ is the copula parameter. 
Set $\lambda_0=1$ and use a burnout sample of the 300 first  observations to initiate the process.  The true value of the linear parameters is set to $\theta^{(1)}=(\beta_0,\beta_1,\beta_2)^\prime=(1,0.3,0.2)^\prime$. Then, the QMLE estimation $\tilde{\theta}^{(1)}$ optimising \eqref{pois-log-lik} is computed for each replication. By considering the results of Proposition~\ref{pnar chi}, the quasi-score statistic \eqref{score test} is evaluated and compared with the critical values of a $\chi^2_1$ distribution. The same data generating process  is employed for simulating  observations $Y_t$ from  \eqref{nonlinear} to compute   the empirical power of the test. All  results of the simulation study are reported in Table~\ref{sim_pois}. Histograms and the Q-Q plots of the simulated score test for the integer-valued case are plotted in Fig.~\ref{qq_pois}. 
We discover similar findings as the ones discussed in the main text. The asymptotic approximation though is slower and therefore  the asymptotic $\chi^2$ distribution of the test requires larger temporal and network sample dimensions.  Hence, the empirical power of the test tends to 1 at  a slower rate. Overall, the empirical evidence shows that, when $N$ and $T$ are large enough, the test performs quite satisfactory in the case of integer-valued data. The adequacy of the $\chi^2_1$ test is also confirmed by histograms and Q-Q plots of Table~\ref{qq_pois}.

\begin{table}[h]
	\centering
	\caption{Empirical size of the test statistics \eqref{score test} for testing $H_0:\gamma=0$ versus $H_1:\gamma>0$, in model \eqref{nonlinear}, with $S=1000$ simulations, for various values of $N$ and $T$. 
		Data are count time series  generated from the linear model \eqref{nar_1}.
		The empirical power is also reported for data generated from model \eqref{nonlinear} with  $\gamma=\left\lbrace 0.5, 1 \right\rbrace $.}
			\begin{tabular}{c|c|c|c|ccc|ccc|ccc}\hline\hline
				Model &\multicolumn{3}{c|}{} &\multicolumn{3}{c|}{Size}  & \multicolumn{3}{c|}{Power $(\gamma=0.5)$} & \multicolumn{3}{c}{Power $(\gamma=1)$} \\\hline
				& $K$ & $N$ & $T$ & 10\% & 5\% & 1\% & 10\% & 5\% & 1\% & 10\% & 5\% & 1\% \\\hline
				\multirow{5}{*}{SBM} & \multirow{5}{*}{2} & 4 & 500 & 0.085 & 0.038 & 0.004 & 0.771 &  0.651 & 0.412 & 0.710 & 0.710 & 0.694 \\
				&& 500 & 10  & 0.100 & 0.031 & 0.001 & 0.096 & 0.031 & 0.003 & 0.096 & 0.045 & 0.010 \\
				&& 200 & 300 & 0.088 & 0.045 & 0.007 & 0.290 &  0.188 & 0.071 & 0.815 & 0.721 & 0.487 \\
				&& 500 & 300 & 0.108 & 0.054 & 0.011 & 0.252 &  0.156 & 0.053 & 0.690 & 0.570 & 0.331 \\
				&& 500 & 400 & 0.123 & 0.075 & 0.015 & 0.297 & 0.190 & 0.060 & 0.809 & 0.705 & 0.454  \\\hline
				\multirow{5}{*}{SBM} &\multirow{5}{*}{5} & 10 & 500 & 0.077 & 0.038  & 0.008 & 0.976 & 0.948 & 0.857 & 0.935 & 0.935 & 0.935 \\
				&& 500 & 10  & 0.089 & 0.027 & 0.002 & 0.114 & 0.049 & 0.004 & 0.196 & 0.099 & 0.009 \\
				&& 200 & 300 & 0.108 & 0.042 & 0.008 & 0.901 &  0.838 & 0.638 & 1.000 & 1.000 & 1.000 \\
				&& 500 & 300 & 0.094 & 0.039 & 0.008 & 0.795 &  0.693 & 0.427 & 1.000 & 0.999 & 0.999  \\
				&& 500 & 400 & 0.108 & 0.059 & 0.011 & 0.863 & 0.771 & 0.549 & 1.000 & 1.000 & 1.000  \\
				\hline
				\hline
				\multirow{5}{*}{ER}& \multirow{5}{*}{-} & 30 & 500 & 0.103 & 0.053 & 0.012 & 0.361 & 0.258 & 0.095 & 0.854 & 0.793 & 0.578 \\
				&& 500 & 30  & 0.105 & 0.045 & 0.003 & 0.084 & 0.040 & 0.004 & 0.086 & 0.050 & 0.016 \\
				&& 200 & 300 & 0.091 & 0.052 & 0.005 & 0.141 &  0.080 & 0.021 & 0.357 & 0.241 & 0.093 \\
				&& 500 & 300 & 0.094 & 0.044 & 0.010 & 0.142 &  0.065 & 0.012 & 0.273 & 0.166 & 0.048 \\
				&& 500 & 400 & 0.112 & 0.062 & 0.014 & 0.169 & 0.098 & 0.024 & 0.342 & 0.230 & 0.086  \\
				\hline
				\hline
			\end{tabular}
		\label{sim_pois}
	\end{table}
	
	\begin{figure}[H]
		\begin{center}
			\includegraphics[width=0.85\linewidth]{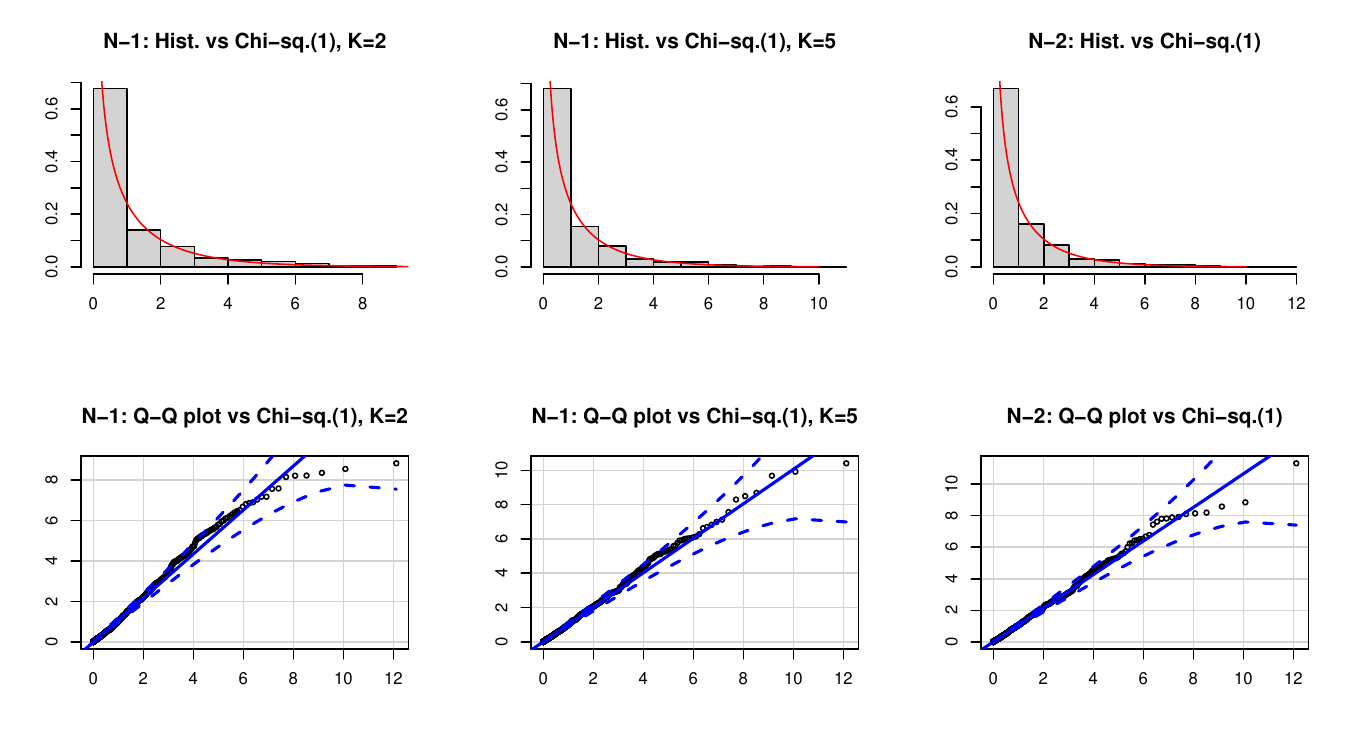}
			\caption{Histogram and Q-Q plots for the score test \eqref{score test}, defined in Tab.~\ref{sim_pois}, against the $\chi^2_1$ distribution; $N=500$, $T=400$. Left: network \ref{sbm}, with $K=2$. Center: network \ref{sbm}, with $K=5$. Right: \ref{erdos-renyi}. Red line: $\chi^2_1$ density. Dashed blue line: $5\%$ confidence bands.}%
			\label{qq_pois}
		\end{center}
	\end{figure}

	\begin{figure}[H]
		\begin{center}
			\includegraphics[width=0.85\linewidth]{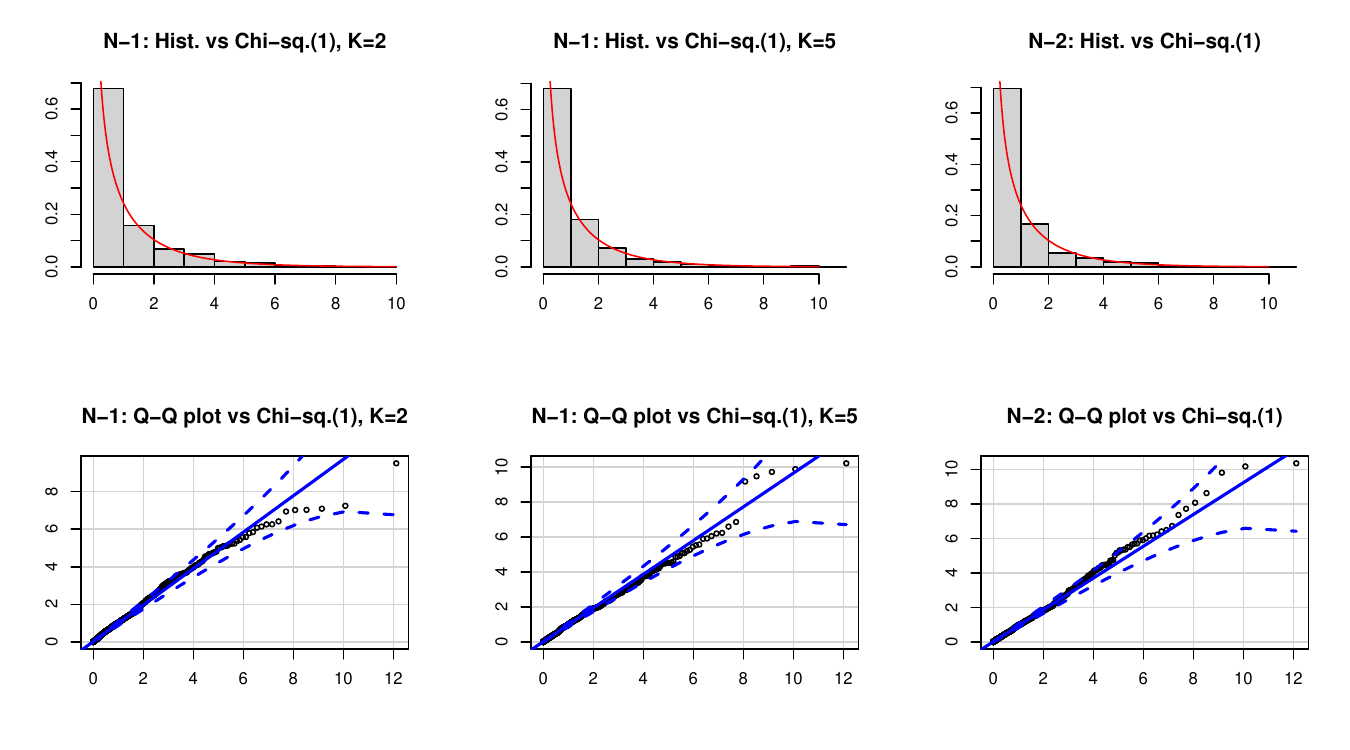}
			\caption{Histogram and Q-Q plots for the score test \eqref{score test}, defined in Tab.~\ref{sim_gauss}, against the $\chi^2_1$ distribution; $N=500$, $T=400$. Left: network \ref{sbm}, with $K=2$. Center: network \ref{sbm}, with $K=5$. Right: \ref{erdos-renyi}. Red line: $\chi^2_1$ density. Dashed blue line: $5\%$ confidence bands.}%
			\label{qq}
		\end{center}
	\end{figure}

	\subsection{Non-standard testing}
	
	\label{simu non standard test}
	Consider the hypothesis test $H_0: \alpha=0$ vs $H_1: \alpha>0\quad(\alpha\neq 0)$, for integer and continuous-valued random variables, respectively, on model \eqref{stnar}. This is an example of a a non-identifiable parameter; in this case it is   $\gamma$. Under the null,  the model reduces to the linear NAR model \eqref{nar_1}. For the continuous case, we simulate data using the data  generating process employed in Section~\ref{SEC: simulations}. Set $(\beta_0,\beta_1,\beta_2)^\prime=(1,0.3,0.2)^\prime$. According to the results of Theorem~\ref{Thm: Test non standard}-\ref{Thm: Test non standard cont}, for each of the $S$ replications,  we can approximates the $p$-values of the sup-type test, $\sup_{\gamma\in \Gamma_F} LM_T(\gamma)$, by the Davies bound \eqref{Davies bound}, with $k_2=1$, where $\Gamma_F$ is a grid of equidistant values computed as in Section~\ref{SEC: applications}. The fraction of cases in which the $p$-value approximation is smaller than the standard  nominal  levels defines the empirical size of the test. The empirical power of the test is obtained by computing the same fraction on data generated by the model \eqref{stnar} instead, with $\alpha=\left\lbrace 0.1,0.5 \right\rbrace $ and $\gamma=\left\lbrace 0.05, 1 \right\rbrace $.
	
	Results for the continuous case are summarized in Table~\ref{sim_gauss_non}. The empirical size is close to the nominal level, although slightly smaller; this is expected as the approximation using the Davies bound \eqref{Davies bound} is conservative.   For both SBM and Erd\H{o}s-R\'{e}nyi model, the empirical power is low, when the tested parameter is close to the tested value, $\alpha=0.1$, and tends to increase for values away from the null ($\alpha=0.5$). The power of the test obviously improves for small values of the nuisance parameter ($\gamma=0.05$). 
	When $N$ and $T$ are large, the test performs adequately.	
	
	The same analysis is replicated for the PNAR model.  The data generating mechanism for the multivariate count process $Y_t$ has been already discussed  in Section~\ref{simu standard testing}  but now the  correlation matrix  is given by  $R_{ij}=\rho$, for all $i,j$, where $\rho=0.5$ is the copula parameter. The results are presented in Table~\ref{sim_pois_non}.  We obtain analogous results as those in the continuous case.

	Finally, a linearity test concerning the threshold model \eqref{tnar}, is empirically  studied. In this case,  $H_0:\alpha_0=\alpha_1=\alpha_2=0$ vs. $ H_1: \alpha_l>0$, for some $l=0,1,2$. 
	The $p$-values are computed from the bootstrap approximation procedure for the test $\textrm{sup}LM_T=\sup_{\gamma\in \Gamma_F} LM_T(\gamma)$.
	The number of bootstrap replications is set to $J=499$. From Table~\ref{sim_gauss_boot} we note that the empirical size is smaller than the expected nominal levels but the empirical power of the test shows that the test works satisfactorily. These outcomes are in line with previous empirical results in the literature \cite[Tab.~5-6]{christou_fokianos_2015}.   Results are also analogous to the Gaussian case presented in Table~\ref{sim_gauss_boot} 
	and therefore are omitted.

	\begin{sidewaystable}[ph!]
		\centering
		\caption{Empirical size  of the test statistics \eqref{score test nuisance} for testing $H_0:\alpha=0$ in $S=1000$ simulations of model \eqref{stnar}, for $N,T=\left\lbrace 200, 500 \right\rbrace $. Data are continuous-valued and generated from the linear model \eqref{nar_1}. The empirical power is also reported for data generated from model \eqref{stnar} with $\alpha=\left\lbrace 0.1,0.5 \right\rbrace $ and $\gamma=\left\lbrace 0.05, 1 \right\rbrace $. The network is derived from Ex.~\ref{sbm}, for first two rows, with $K=2$; second two: Ex.~\ref{erdos-renyi}.} 	\label{sim_gauss_non}
			\begin{tabular}{c|c|c|ccc|cccccccccccc}\hline\hline
				Model & \multicolumn{2}{c|}{} &\multicolumn{3}{c|}{Size}  & \multicolumn{12}{c}{Power } \\\hline
				& \multicolumn{2}{c|}{} &\multicolumn{3}{c|}{}  & \multicolumn{3}{c}{$\alpha=0.1$, $\gamma=1$} & \multicolumn{3}{c}{$\alpha=0.1$, $\gamma=0.05$} &\multicolumn{3}{c}{$\alpha=0.5$, $\gamma=1$ } &\multicolumn{3}{c}{$\alpha=0.5$, $\gamma=0.05$ } \\\hline
				&	$N$ & $T$ & 10\% & 5\% & 1\% & 10\% & 5\% & 1\% & 10\% & 5\% & 1\% & 10\% & 5\% & 1\% & 10\% & 5\% & 1\%  \\\hline
				\multirow{2}{*}{SBM} & 200 & 200 & 0.078 & 0.045 & 0.012 & 0.092 &  0.044 & 0.007 & 0.099 & 0.055 & 0.014 & 0.322 & 0.194 & 0.061 & 0.921 & 0.746 & 0.214 \\
				& 500 & 500 & 0.084 & 0.036 & 0.006 & 0.088 &  0.0048 & 0.014 & 0.128 & 0.074 & 0.012 & 0.504 & 0.368 & 0.167 & 1.000 & 0.997 & 0.840 \\
				\hline
				\hline
				\multirow{2}{*}{ER} &200 & 200 & 0.078 & 0.034 & 0.003 & 0.088 &  0.039 & 0.009 & 0.097 & 0.046 & 0.006 & 0.173 & 0.099 & 0.019 & 0.930 & 0.772 & 0.321 \\
				& 500 & 500 & 0.082 & 0.037 & 0.011 & 0.102 &  0.053 & 0.010 & 0.105 & 0.049 & 0.012 & 0.220 & 0.129 & 0.038 & 0.999 & 0.998 & 0.911 \\
				\hline
				\hline
			\end{tabular}
		
		\bigskip\bigskip \bigskip\bigskip 
		\centering
		\caption{Empirical size   of the test statistics \eqref{score test nuisance} for testing $H_0:\alpha=0$ in $S=1000$ simulations of model \eqref{stnar}, for $N,T=\left\lbrace 200, 500 \right\rbrace $. Data are integer-valued and generated from the linear model \eqref{nar_1}. The empirical power is also reported for data generated from model \eqref{stnar} with $\alpha=\left\lbrace 0.1,0.5 \right\rbrace $ and $\gamma=\left\lbrace 0.05, 1 \right\rbrace $. The network is derived from Ex.~\ref{sbm}, for first two rows, with $K=2$; second two: Ex.~\ref{erdos-renyi}.} \label{sim_pois_non}
			\begin{tabular}{c|c|c|ccc|cccccccccccc}\hline\hline
				Model & \multicolumn{2}{c|}{} &\multicolumn{3}{c|}{Size}  & \multicolumn{12}{c}{Power } \\\hline
				& \multicolumn{2}{c|}{} &\multicolumn{3}{c|}{}  & \multicolumn{3}{c}{$\alpha=0.1$, $\gamma=1$} & \multicolumn{3}{c}{$\alpha=0.1$, $\gamma=0.05$} &\multicolumn{3}{c}{$\alpha=0.5$, $\gamma=1$ } &\multicolumn{3}{c}{$\alpha=0.5$, $\gamma=0.05$ } \\\hline
				&	$N$ & $T$ & 10\% & 5\% & 1\% & 10\% & 5\% & 1\% & 10\% & 5\% & 1\% & 10\% & 5\% & 1\% & 10\% & 5\% & 1\%  \\\hline
				\multirow{2}{*}{SBM} & 200 & 200 & 0.052 & 0.027 & 0.006 & 0.059 &  0.024 & 0.002 & 0.088 & 0.045 & 0.007 & 0.129 & 0.072 & 0.013 & 0.601 & 0.460 & 0.179 \\
				& 500 & 500 & 0.058 & 0.025 & 0.005 & 0.054 &  0.028 & 0.005 & 0.146 & 0.072 & 0.015 & 0.218 & 0.130 & 0.034 & 0.902 & 0.804 & 0.580 \\
				\hline
				\hline
				\multirow{2}{*}{ER} &200 & 200 & 0.068 & 0.025 & 0.001 & 0.055 &  0.019 & 0.002 & 0.095 & 0.043 & 0.005 & 0.136 & 0.051 & 0.008 & 0.538 & 0.393 & 0.143 \\
				& 500 & 500 & 0.068 & 0.026 & 0.002 & 0.071 &  0.028 & 0.004 & 0.123 & 0.063 & 0.015 & 0.200 & 0.108 & 0.023 & 0.880 & 0.794 & 0.557 \\
				\hline
				\hline
				%
			\end{tabular}
	\end{sidewaystable}

	\begin{table}[h]
		\centering
		\caption{Empirical size at  of the test statistics \eqref{score test nuisance} for testing $H_0:\alpha_0=\alpha_1=\alpha_2=0$ in $S=200$ simulations of model \eqref{tnar}, for $N=8$ and $T=1000$. Data are continuous-valued and generated from the linear model \eqref{nar_1}. The empirical power is also reported for data generated from model \eqref{tnar} with $\alpha=(\alpha_0, \alpha_1, \alpha_2)^\prime=(0.5, 0.2, 0.1)^\prime$ and $\gamma=1$. The network is derived from Ex.~\ref{sbm}, with $K=2$.}
			\begin{tabular}{c|c|ccc|ccc}\hline\hline
				\multicolumn{2}{c|}{} &\multicolumn{3}{c|}{Size}  & \multicolumn{3}{c}{Power} \\\hline
				\multicolumn{2}{c|}{} &\multicolumn{3}{c|}{}  & \multicolumn{3}{c}{$\alpha=(0.5, 0.2, 0.1)^\prime$, $\gamma=1$} \\\hline
				$N$ & $T$ & 10\% & 5\% & 1\% & 10\% & 5\% & 1\% \\\hline
				8 & 1000 & 0.045 & 0.002 & 0.000 & 1.000 &  1.000 & 1.000 \\
				\hline
				\hline
			\end{tabular}
		\label{sim_gauss_boot}
	\end{table}

	\subsection{Additional empirical results}

	\label{wind}

	In this section we consider a continuous-valued dataset containing $T=721$ wind speeds measured at each of $N=102$ weather stations throughout England and Wales. The weather stations are the nodes of the potential network and two weather stations are connected if they share a border. In this way undirected network is drawn on geographic proximity. The dataset is taken by the \texttt{GNAR} \texttt{R} package \cite{Knightetal(2020)} incorporating the time series data \texttt{vswindts} and the associated network \texttt{vswindnet}. 
	As the wind speed is continuous-valued, the NAR model \eqref{nar_1} is estimated by OLS. The linearity test against model \eqref{nonlinear_cont} is computed and compared with the $\chi^2_1$ distribution. Then, the same tests of Section~\ref{SEC: applications} are evaluated, with $H_0: \alpha=0$ vs. $ H_1: \alpha\neq 0$, for the STNAR model \eqref{stnar} and
	$H_0:\alpha_0=\alpha_1=\alpha_2=0$ vs. $ H_1: \alpha_l\neq0$, for some $l=0,1,2$, in the TNAR model \eqref{tnar}. 
	The results are summarized in Table~\ref{wind_results}.
	All the  estimated coefficients are significant at  standard nominal  levels. In this case, the magnitude of the lagged variable dominates the network effect but the latter  still has considerable  impact. An increasing positive effect is detected for the coefficients. This means that the expected wind speed for a weather station is positively correlated  by its past speed and the past wind speeds detected on neighboring  stations. The estimated variance of the errors is approximately 0.156; it is computed with the moment estimator $\tilde{\sigma}^2=(NT)^{-1}\sum_{t=1}^{T}\sum_{i=1}^{N}(Y_{i,t}-\lambda_{i,t}(\tilde{\theta}))^2$. The existence of nonlinear components is typically encountered in environmental time series studies. The null assumption of linearity is rejected against the  alternative \eqref{nonlinear_cont}. This gives an indication of possible nonlinear shifts in the intercept.
	In addition, the linearity is also rejected when it is tested against the STNAR model \eqref{stnar} and TNAR model \eqref{tnar}, at the usual significance levels, using  the Davies bound and the bootstrap sup-type test. Such results suggest that a nonlinear regime switching effects could be present. 
	
	\begin{table}[H]
		\centering
		\caption{QMLE estimates of the linear model \eqref{nar_1} for wind speed data. Standard errors in brackets.
			Linearity is tested against the nonlinear model \eqref{nonlinear_cont}, with $\chi^2_1$ asymptotic test \eqref{score test}; against the STNAR model \eqref{stnar}, with approximated $p$-values computed by (DV) Davies bound \eqref{Davies bound}, bootstrap $p$-values of sup-type test; 
			and versus TNAR model \eqref{tnar} (only bootstrap). }
			\begin{tabular}{cccc}\hline\hline
				Models & $\beta_0$ & $\beta_1$ & $\beta_2$ \\\hline 
				\eqref{nar_1} & 0.154  & 0.157 & 0.768  \\\hline
				Std. & (0.017) & (0.011) & (0.009)  \\
				\hline
				\hline
				Models & Chi-sq. & DV & Bootstrap  \\
				\eqref{nonlinear} & 131.052 & - & -  \\
				\eqref{stnar} & - & $<0.001$ & 0.02 \\
				\eqref{tnar} & - & - & $<0.001$ \\
				\hline
				\hline
			\end{tabular}
		\label{wind_results}
	\end{table}

\small

\bibliographystyle{plainurl}
\bibliography{lin_test}
\end{document}